\documentclass[a4paper,12pt,twoside,onecolumn,number,review]{ArXiv}

\usepackage{times,url,graphicx,epsfig}
\usepackage{amsmath,amsthm,amssymb}
\usepackage[ruled,vlined,linesnumbered,algoruled]{algorithm2e}
\usepackage{setspace}
\usepackage{booktabs}
\usepackage{threeparttable}
\usepackage{multirow}
\usepackage{fancyhdr}

\usepackage{hhline}
\usepackage{float}
\usepackage[caption=false]{subfig}
\usepackage{threeparttable}
\usepackage{lscape}
\usepackage[usenames, dvipsnames]{color}

\renewcommand{\baselinestretch}{1.3}
\definecolor{cpb}{rgb}{0,0,0}
\newcommand {\pb}[1] {{\color{cpb}{#1}}}

\definecolor{SPAColor}{rgb}{0.0,0.0,0}

\definecolor{SPCColor}{rgb}{0,0,0}

\definecolor{PBColor}{rgb}{0,0.0,0}
\definecolor{RHTcolor}{rgb}{0.9,0.9,0.9}
\definecolor{EVMcolor}{rgb}{0.9,0.9,0.9}
\definecolor{OURcolor}{rgb}{1.0,1.0,1.0}

\newcommand {\arc}[2]   {{\mathcal A}^{\mathbb R}_{#1,#2}} 
\newcommand {\darc}[2]  {{\mathcal A}^{\mathbb Z}_{#1,#2}} 
\newcommand {\rcir}     {{\mathcal C}^\mathbb R}
\newcommand {\dcir}     {{\mathcal C}^\mathbb Z}

\newcommand {\inz}      {\in \mathbb Z}

\newcommand {\zz}       {\mathbb Z^2}
\newcommand {\rr}       {\mathbb R^2}
\newcommand {\sm}       {\smallsetminus}
\newcommand {\care}[3] {\varepsilon_{#1\in#2}^{#3}}

\newcommand{\remove}[1]{}
\newtheorem{theorem}{Theorem}

\newtheorem{corollary}{Corollary}

\begin{document}

\begin{frontmatter}

\title{On Chord and Sagitta in $\zz$: An Analysis towards Fast and Robust Circular Arc Detection\tnoteref{ltitle}}
\tnotetext[ltitle]{A preliminary version of the paper was presented at GREC~2009~\cite{bera_10}.}

\author[isi]{Sahadev Bera}
\ead{sahadevbera@gmail.com}
\author[isi]{Shyamosree Pal}
\ead{shyamosree.pal@gmail.com}
\author[iit]{Partha Bhowmick\corref{corr}}
\ead{pb@cse.iitkgp.ernet.in, bhowmick@gmail.com}
\author[isi]{Bhargab B. Bhattacharya}
\ead{bhargab@isical.ac.in}

\cortext[corr]{Author for correspondence.}
\address[isi]{Advanced Computing and Microelectronics Unit, Indian Statistical Institute, Kolkata, India}
\address[iit]{Computer Science and Engineering Department, Indian Institute Technology, Kharagpur, India\bigskip\\
{\em September 09, 2014}}

\begin{abstract}

Although chord and sagitta, when considered in tandem, may reflect many underlying geometric
properties of circles on the Euclidean plane, their implications on the digital plane are not yet
well-understood. In this paper, we explore some of their fundamental properties on the digital plane
that have a strong bearing on the unsupervised detection of circles and circular arcs in a digital
image. We show that although the chord-and-sagitta properties of a real circle do not readily
migrate to the digital plane, they can indeed be used for the analysis in the discrete domain based
on certain bounds on their deviations, which are derived from the real domain. In particular, we
derive an upper bound on the circumferential angular deviation of a point in the context of chord
property, and an upper bound on the relative error in radius estimation with regard to the sagitta
property. Using these two bounds, we design a novel algorithm for the detection and parameterization
of circles and circular arcs, which does not require any heuristic initialization or manual tuning.
The chord property is deployed for the detection of circular arcs, whereas the sagitta property is
used to estimate their centers and radii. Finally, to improve the accuracy of estimation, the notion
of restricted Hough transform is used. Experimental results demonstrate superior efficiency and
robustness of the proposed methodology compared to existing techniques.
\end{abstract}


\end{frontmatter}

\section{Introduction}
\label{sec:intr}
The properties of a digital circle are well-studied in discrete geometry and have found diverse real-life
applications in various fields of science and engineering.
Fast and accurate recognition of circles or circular arcs in a digital image is a challenging problem with
practical relevance in computer vision~\cite{davi_97, gonz_01, pratt_87, sonk_98},
physics~\cite{carter_05,cher_84,craw_83}, biology and medicine~\cite{bigg_72,paton_70}, and industrial 
engineering~\cite{kasa_76,thomas_89}.
Most of the existing algorithms for detection of circles and circular arcs are based on the properties of
circles on the real or Euclidean plane.
However, the properties of a real circle cannot be readily used for analyzing a digital circle, since the
latter essentially comprises a sequence of points on the integer plane. 
In this paper, we study some of these real-geometric properties of the circle, which can be used to detect
digital circles and circular arcs after considering their impact on $\zz$.

The paper is organized as follows.
A brief review of existing work related with digital (circular) arc segmentation is presented in
Sec.~\ref{ssec:relaWorks}.
The chord-and-sagitta properties a real circle are introduced in Sec.~\ref{sec:chordsagi}.
We derive some important results related to these properties for digital circles, in the respective
subsections.
Sec.~\ref{sec:proposed} describes how these techniques are applied to recognize circular arcs and to estimate
their centers and radii. 
Experimental results are reported in Sec.~\ref{sec:results}. 
Finally, in Sec.~\ref{sec:conclu}, we draw the concluding notes and discuss future research issues.

\subsection{Related Work}
\label{ssec:relaWorks}
There exist several algorithms in the literature for the detection of circular arcs in a digital image.
Most of these algorithms are based on Hough transform (HT) or its variants~\cite{ chiu_05, coeurj_04,
davies_84, duda_75, illin_88, leav_93, kim_01, xu_93}.
Several other techniques have been proposed later to improve the performance of HT, such as Fast Hough
Transform (FHT)~\cite{guil_95, li_86}, Randomized Hough Transform (RHT)~\cite{xu_93, xu_90}, and Adaptive
Hough Transform (AHT)~\cite{illin_87}.
The main objective of these HT-based methods is either to reduce the computation or to reduce the memory requirement.
In one such method, the parameter space is decomposed into several lower dimension parameter spaces~\cite{yip_92}.
It then estimates the parameters of the circles based on local geometrical properties; however, they suffer
from poor consistency and location accuracy because of quantization error.
To overcome these disadvantages, Ho and Chen [6] used the global geometrical symmetry of circle to reduce the
dimension of the parameter space.
Xu {\it et~al.}~\cite{xu_93} presented a randomized Hough transform, which
reduces the storage requirement and computational time significantly compared to other methods based on the
conventional HT.

Li {\it et al.}~\cite{li_86} have developed a fast algorithm for the Hough transform that can be incorporated into the solutions to many problems in computer vision such as line detection, plane detection, segmentation, and motion estimation. 
The fast Hough transform (FHT) algorithm assumes that image space features “vote” for sets of points lying on hyperplanes in the parameter space. 
It recursively divides the parameter space into hypercubes from low to high resolution and performs the Hough
transform only on the hypercubes when their votes exceed a selected threshold.
The decision on whether a hypercube receives a vote from a hyperplane depends on whether the hyperplane intersects the hypercube. 
This hierarchical approach leads to a significant reduction of both computation and storage.
Based on CHT, a size-invariant circle-detection algorithm was proposed~\cite{athe_99}. 

Some methods use randomized selection of edge points and geometrical properties of circle instead of using the information of edge pixels and evidence histograms in the parameter space. 
Kim {\it et~al.}~\cite{kim_01} have proposed a two-step circle detection algorithm, given a pair of intersecting chords, in which the first step is to compute the center of the circle using 2D-HT.
In the second step, the 1D radius histogram is used to identify the circle and to compute its radius.

Although HT is robust against noise, clutter, object defect, and shape distortion, its main drawback is high computational time and space.
So, several other techniques have also been proposed, which do not use histograms in the parameter space, such as least-squares fitting~\cite{cher_05, thomas_89}, randomization and geometry~\cite{chen_01b, chun_07, ho_95}, and genetic algorithm~\cite{nagao_93}.
These non-HT-based algorithms can extract circles and arcs faster than the HT-based methods, as they do not use histograms in the parameter space.
Some of these algorithms are discussed below.

Xu {\it et al.}~\cite{xu_90} presented an approach that randomly selects three pixels. 
The method selects three non-collinear edge pixels and votes for the circle parameters which are found by using the circle equation. 
In order to improve the performance, Chen and Chung~\cite{chen_01b} proposed an efficient randomized
algorithm (RCD) for detecting circles that does not use HT as well as the accumulator needed for saving the
information of related parameters.
The underlying concept in RCD is to first select four edge pixels randomly in the image and then to use a distance criterion to determine whether there might exist a possible circle, and finally to collect further evidence for determining whether or not, it is indeed a circle.

Gradient information of each edge pixel is used in~\cite{rad_03} to reduce the computing time or the memory
requirement.
After computing the gradient of the whole image, the pair of anti-parallel vectors are searched to detect
circle.
Chiu {\it et~al.}~\cite{chiu_05} proposed an effective voting method for circle detection, which also does not
use HT.
Rather, it reduces the data set by a sampling technique; 
once the first two points are chosen, the third one is chosen following certain criteria of circularity.
\pb{The UpWrite method~\cite{mcla_98} works with local models of the pixels within small neighborhoods, 
computed by a spot algorithm, to classify these local models by their geometric features for circle detection.}

Recently Jia {\em et al}.~\cite{jia_11} present a fast randomized circle detection algorithm, which can be
applied to determine the centers and radii of circular components. 
Firstly, the gradient of each pixel in the image is computed using Gaussian template. 
Then, the edge map of the image, obtained by applying Canny edge detector, is tackled to acquire the curves
consisting of \pb{$8$-connected} edge points.  
Subsequently, for the detection of the center, edge points for each curve are picked up, and the point,
passed through by \pb{most of the} gradient lines of the edge points, corresponds to a center.  
The radius \pb{is estimated from the distances of} the corresponding edge points
\pb{from the center}. 
The algorithm performs much better in terms of efficiency compared to randomized circle detection algorithm
(RCD). 
\pb{More recent work related to line and circle detection can be seen in \cite{Kole12,Kole14}.
In~\cite{Kole12}, Kolesnikov proposed dynamic programming solutions to multi-model curve approximation
constrained by a given error threshold.
In~\cite{Kole14}, Kolesnikov and Kauranne have proposed a parameterized model of rate-distortion curve, which
produces a multi-model approximation by minimizing the associated cost function.}

\subsection{Our Contribution}
\label{ssec:conti}
All prior work use different properties of real circles while detecting circles and circular arcs in the
digital plane.
The {\em chord property}~\cite{weis_ch} and the {\em sagitta property}~\cite{weis_sa} are two most important
properties of real circles, whose deviation in $\zz$ has to be analyzed so that they can be efficiently used
in the detection of digital circles and arcs.
In this paper, we study, for the first time, some digital-geometric properties of chord and sagitta.
Based on these properties, we propose a novel technique for recognizing digital circles and circular arcs.
The chord property in $\zz$ is first used to identify the circular curve segments.
The sagitta property is then used to determine the radii and the centers of the circular curve segments.
The proposed technique not only detects isolated circles and circular arcs, but also joins the concentric arc
segments with the same or nearly same radius to make out a full circle or a larger arc.
For improving the accuracy of the radii and the centers, a {\em restricted Hough transform} (rHT) is applied.

\section{Chord and Sagitta Properties in $\zz$}
\label{sec:chordsagi}
In order to study the properties of chord and sagitta in $\zz$, we start with some definitions from the
literature~\cite{gonz_01,klette_04}.
A pixel $p$ is a point in $\zz$. 
Two pixels $p(x_p,y_p)$ and $q(x_q,y_q)$, $p\neq q$, are $8$-neighbors of each other if and only if $\max\{|x_p-x_q|,|y_p-y_q|\}\leq 1$.
A sequence of pixels defines a (8-connected) digital curve if and only if each pixel, excepting the first one and the last one, has exactly two neighbors in its $8$-neighborhood.
For an open digital curve, the first and the last pixel have one $8$-neighbor each;
for a closed digital curve, they have two $8$-neighbors each.

\subsection{Chord Property}
\label{ssec:chord}
\begin{figure}[t]\center
\includegraphics[width=\textwidth]{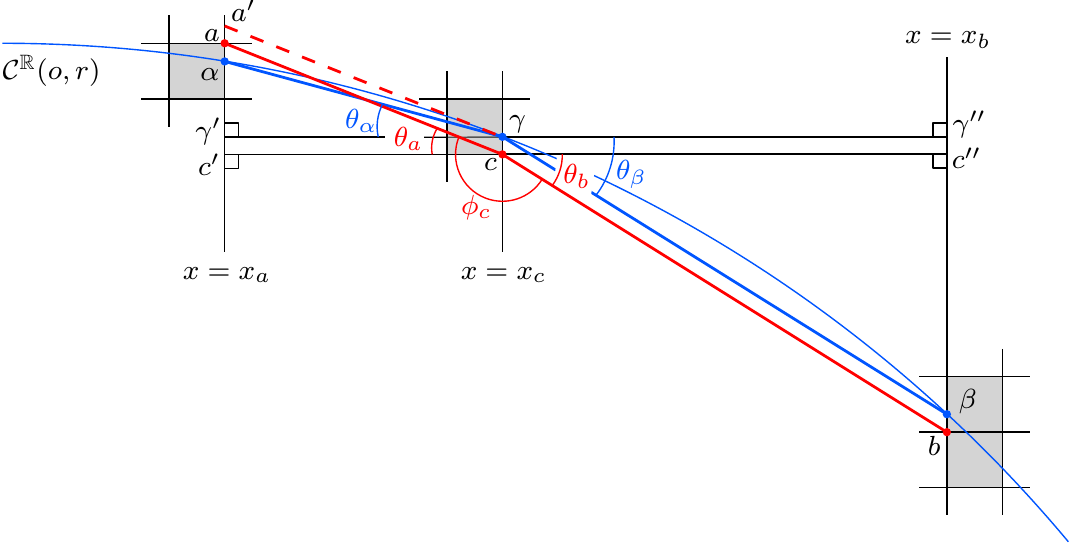}
\caption{Deviation of the chord property (Case~1).
Points in $\rr$ ($\alpha, \beta, \ldots$) or on the real circle $\rcir(o,r)$ are shown in blue, 
and points shown in red ($a, b, c$) belong to the digital circle, $\dcir(o,r)$.
$\arc{\alpha}{\beta}$ is an arc of $\rcir(o,r)$, which corresponds to the given 
digital (circular) arc $\darc{a}{b}$.
As $c$ changes its place along $\darc{a}{b}$ such that $\left|y_\gamma-y_c\right|<\frac12$, 
the  angle $\phi_c$ gets deviated by $\pm\delta_{\phi}$.}
\label{fig:chord-dev}
\end{figure}

We show here the deviation of chord property for digital circular arcs.
See Figure~\ref{fig:chord-dev} for an illustration.
Let $\rcir(o,r)$ be the real circle having center at $o(0,0)$ and radius $r\inz^{+}$.
Let $\arc{\alpha}{\beta}$ be an arc of $\rcir(o,r)$ having endpoints 
$\alpha(x_\alpha,y_\alpha)$ and $\beta(x_\beta,y_\beta)$.
Let $\gamma(x_\gamma,y_\gamma)$ be another point on $\arc{\alpha}{\beta}$.
Then, by the {\em chord property} of real circle, ${\alpha\beta}$ subtends a fixed angle 
$\phi_\gamma$ at $\gamma$, irrespective of its position on $\arc{\alpha}{\beta}\sm\{\alpha,\beta\}$~\cite{weis_ch}.
However, this is not true for a digital (circular) arc, as shown below. 

Let $\dcir(o,r)$ be the digital circle corresponding to $\rcir(o,r)$,
and $\darc{a}{b}$ be the digital arc corresponding to $\arc{\alpha}{\beta}$,
where $a(x_a,y_a)$, $b(x_b,y_b)$, $c(x_c,y_c)$ are the respective integer points corresponding to 
$\alpha,\beta,\gamma$;
$\phi_c$ be the angle subtended by the line segment ${ab}$ at $c$;
$c'$ and $\gamma'$ be the respective feet of the perpendiculars dropped from $c$ and $\gamma$ to the
line $x=x_a$, and $c''$ and $\gamma''$ be those from $c$ and $\gamma$ to $x=x_b$.
Let $\theta_a$, $\theta_\alpha$, $\theta_b$, and $\theta_\beta$ be the acute angles subtended at $c$
and $\gamma$ by the corresponding perpendiculars.
Then, the angle $\phi_c$ deviates from $\phi_\gamma$ by an amount 
depending on the position of $c$ in $\darc{a}{b}$.
In~particular, we have the following theorem.

\begin{theorem}
\label{thm:chord}
If $\darc{a}{b}$ is a digital circular arc with a point $c$ on it, then the
circumferential angular deviation of $c$ from its corresponding point $\gamma$ on the real arc is 
less than $\sin^{-1}\frac{1}{{ac}} + \sin^{-1}\frac{1}{{cb}}$.
\end{theorem}

\begin{proof}
Let ${x_{ac}}=x_c-x_a$, ${y_{ca}}=y_a-y_c$, ${x_{\alpha\gamma}}=x_{\gamma}-x_{\alpha}$,  ${y_{\gamma\alpha}}=y_{\alpha}-y_{\gamma}$, ${x_{cb}}=x_b-x_c$, ${y_{bc}}=y_c-y_b$, 
${x_{\gamma\beta}}=x_{\beta}-x_{\gamma}$,  ${y_{\beta\gamma}}=y_{\gamma}-y_{\beta}$.

As shown in \cite{bhow_08}, a real circle with integer center and integer radius never intersects 
a grid line at the middle between two consecutive grid points.
Hence, we get 
\begin{equation}
\left|{x_{ac}}-{x_{\alpha\gamma}}\right|<1, \ \left|{y_{ac}}-{y_{\gamma\alpha}}\right|<1, \left|{x_{cb}}-{x_{\gamma\beta}}\right|<1, \left|{y_{bc}}-{y_{\beta\gamma}}\right|<1.
\label{eqn:y_bounds}\end{equation} 

Now, we have two possible cases as follows, which are illustrated in Figure~\ref{fig:chord-dev}
and Figure~\ref{fig:chord-dev1}.

\paragraph{Case 1} $x$-coordinates of both $\alpha$ and $\gamma$ are integers.
W.l.o.g., let $\alpha$ and $\gamma$ be in Octant~1. 
Then, $x_a=x_{\alpha}$ and $x_c=x_{\gamma}$, whence $x_{ac}=x_{\alpha\gamma}$. 
So, for maximum deviation, $y_a>y_{\alpha}$ and $y_{\gamma}>y_c$.
Now, consider the line $\gamma a'$ parallel to ${ac}$. 
The circumferential angular deviation at $c$ is $\theta_{a\alpha}=\angle a'\gamma\alpha$.
Using simple geometric property in $\Delta a'\gamma\alpha$, we get ${a'\gamma}\sin\theta_{a\alpha}\leq\alpha a'$.
By Eqn.~\ref{eqn:y_bounds}, we have $\alpha a'= \alpha a+aa'=\alpha a+c\gamma= |{y_{ac}}-{y_{\gamma\alpha}}|<1$.
As $a'\gamma=ac$, we get $\theta_{a\alpha}<\sin^{-1}\frac{1}{{ac}}$.

When $y$-coordinates of both $\alpha$ and $\gamma$ are integers, we also get $\theta_{a\alpha}<\sin^{-1}\frac{1}{{ac}}$, the proof being similar as above.
And when $\beta$ and $\gamma$ have integer $x$- or $y$-coordinates, the proof is again similar, and we get  ${\theta_{b\beta}}<\sin^{-1}\frac{1}{{bc}}$.

\begin{figure}[!t]\center
\includegraphics[width=\textwidth]{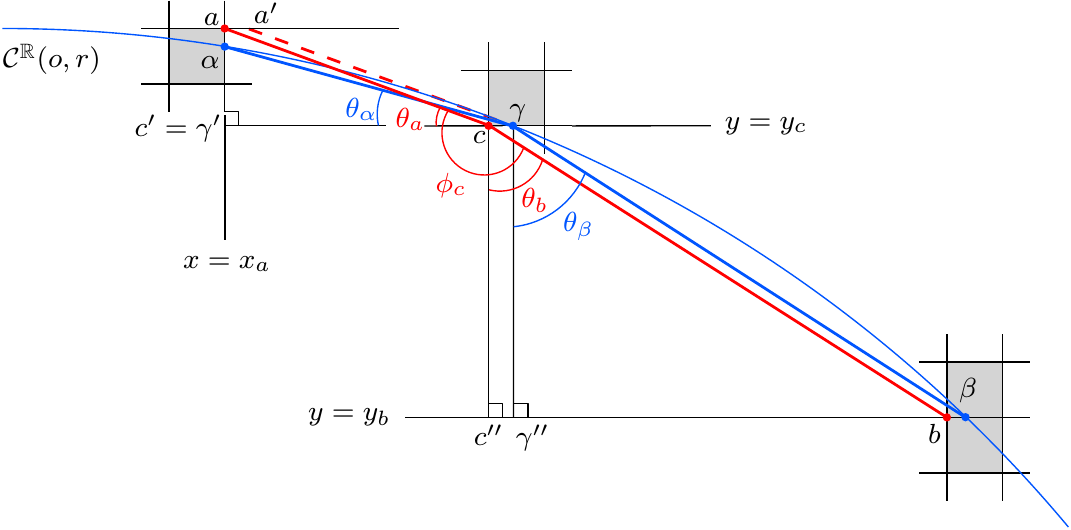}
\caption{Deviation of the chord property (Case~2).}
\label{fig:chord-dev1}
\end{figure}

\paragraph{Case 2} 
$x$-coordinate of $\alpha$ and $y$-coordinate of $\gamma$ are integers.
W.l.o.g., let $\alpha$ be in Octant~1 and $\gamma$ in Octant~2. 
Then, $x_a=x_{\alpha}$, $\left|y_a-y_{\alpha}\right|<\frac12$, $y_c=y_{\gamma}$, $\left|x_c-x_{\gamma}\right|<\frac12$.
For maximum deviation, $y_a>y_{\alpha}$ and $x_{\gamma}>x_c$.
Now, consider the line $\gamma a'$ parallel (and equal in length) to ${ac}$. 
The circumferential angular deviation at $c$ is $\theta_{a\alpha}=\angle a'\gamma\alpha$.
Similar to Case~1, in $\Delta a'\gamma\alpha$, we get ${a'\gamma}\sin\theta_{a\alpha}\leq\alpha a'$.
As $aa'=c'c''=c\gamma<\frac12$ and $\alpha a'<aa'+a\alpha<\frac12+\frac12=1$, we get $\theta_{a\alpha}<\sin^{-1}\frac{1}{{ac}}$.

For a different combination on integer values of $x$- and $y$-coordinates, the proof is similar.
\end{proof}

From the above theorem, it is clear that the circumferential angular deviation is dominated by the smaller value between $ac$ and $cb$.
We have the following corollary if they are equal.

\begin{corollary}
If $m$ is the middle pixel (one of the two if it is not unique) of $\darc{a}{b}$, then the maximum possible deviation of $\phi_m$ from $\phi_\gamma$ is given by
\begin{eqnarray}
\delta_{\phi}< 2\sin^{-1}\frac{1}{am}< 2\sin^{-1}\frac{2}{ab},\ \mbox{since ${am}+{mb}>{ab}$}.\end{eqnarray}
\end{corollary}

\begin{figure}[!t]
\centering
\includegraphics[width=.8\textwidth]{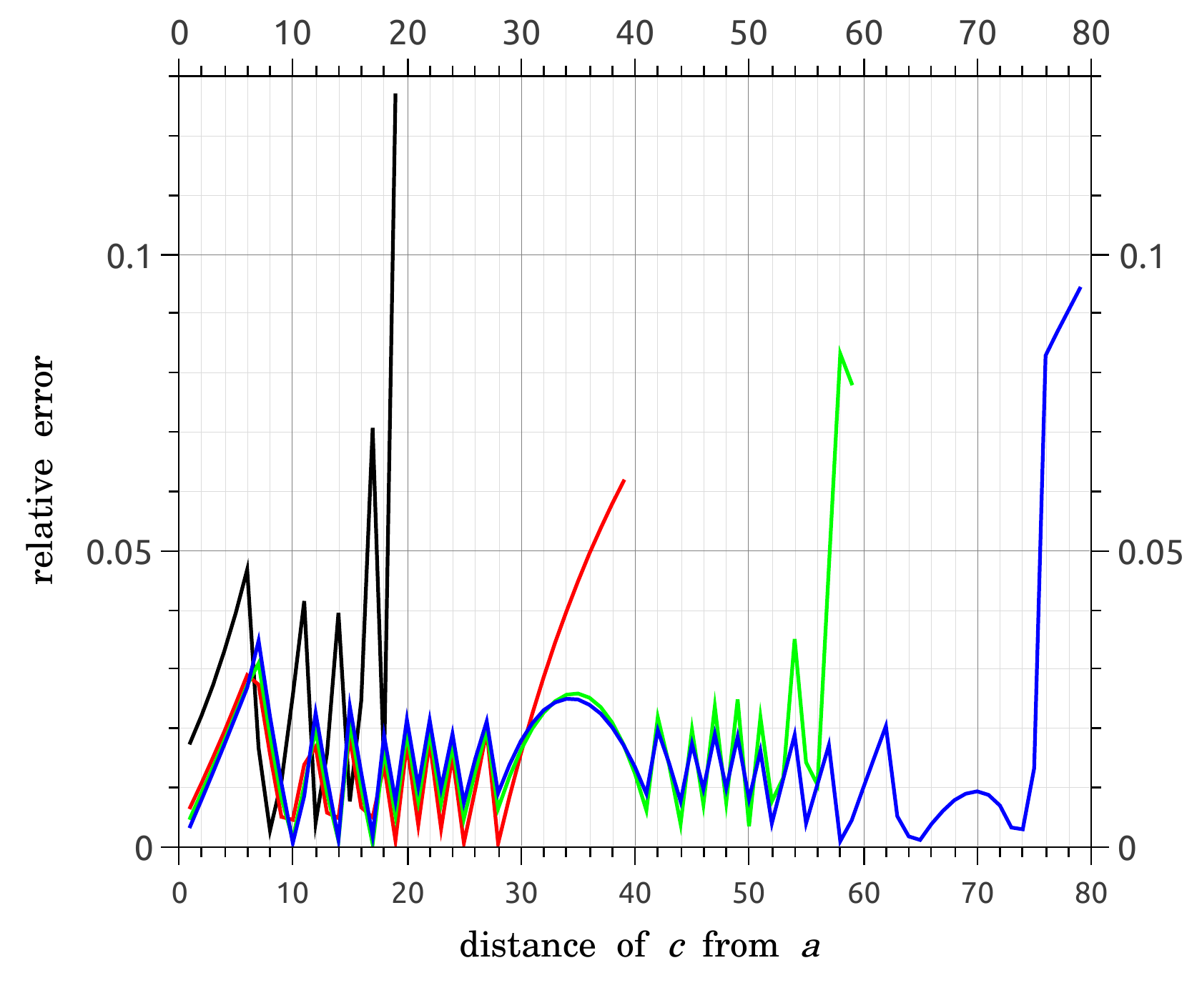}
\caption{Circumferential angle relative error $\care{c}{ab}{r}$ versus distance of $c$ from the endpoint $a$ of the digital arc $ab$ of circle having radius $r=50$. Four different digital arcs of length $20$ (black), $40$ (red), $60$ (green), and $80$ (blue)  are considered.}
\label{fig:chord:err}
\end{figure} 

Theorem~\ref{thm:chord} also implies that if ${ac}$ or ${cb}$ is sufficiently small, then the total error would be unacceptably large.
To measure this, we define {\em circumferential angle relative error} at the pixel $c$ lying on the digital circular arc $ab$ of radius $r$ as $\care{c}{ab}{r}=\frac{|\phi_\gamma-\phi_c|}{\phi_\gamma}$.
For a given digital arc $ab$, we vary the position of $c$ on $ab$ and compute the corresponding value of $\care{c}{ab}{r}$.
Figure~\ref{fig:chord:err} shows a set of four such plots of $\care{c}{ab}{r}$ for four different digital circular arcs of length $20$, $40$, $60$, and $80$, taken from a digital circle of radius $r=50$.
It is clear from these plots that the error is abnormally large at or near the two endpoints, whereas in the intermediate portion, it is much less.
Hence, such errors have to be kept in consideration for obtaining proper results.
In our algorithm, we have considered this, resulting to satisfactory performance in terms of both precision and robustness.

\subsection{Sagitta Property}
\label{ssec:sagi}

\begin{figure}[!t]
\centering
\includegraphics[width=\textwidth]{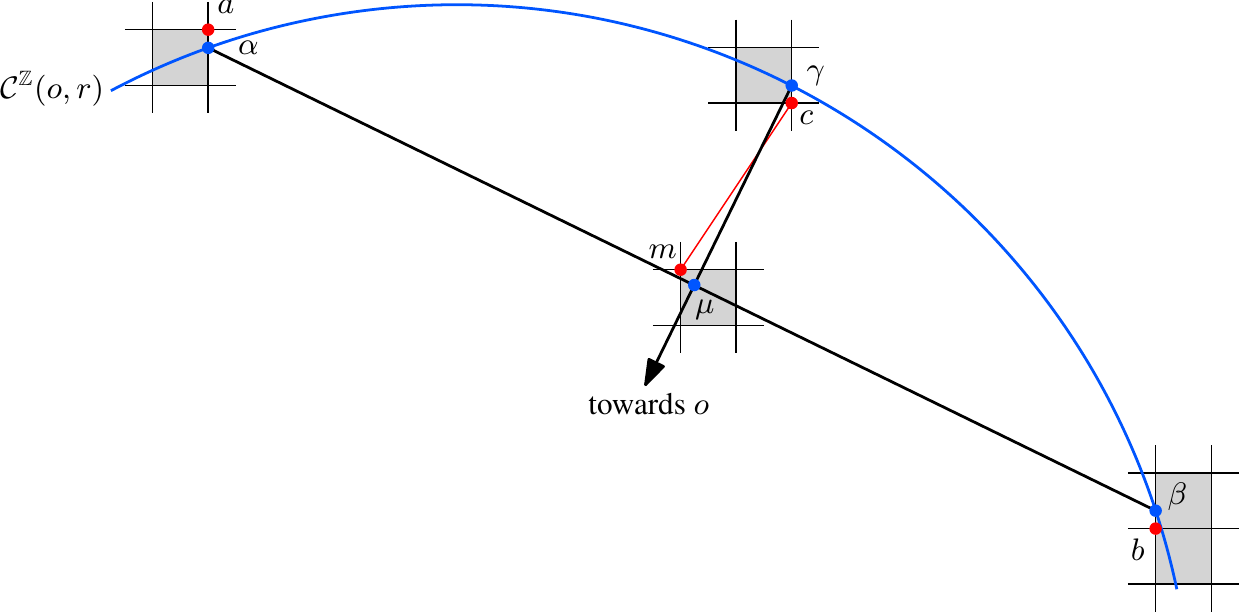}
\caption{Deviation of the sagitta property.
Points in $\rr$ or on $\rcir(o,r)$ are in blue, and those in $\zz$ or in $\dcir(o,r)$ are in red.
\pb{The real sagitta here is $\mu\gamma$, and its discrete counterpart is $mc$.}}
\label{fig:sagitta-dev}
\end{figure}

We show here the deviation of sagitta property for digital circular arcs.
See Figure~\ref{fig:sagitta-dev} for an illustration.
Let $\rcir(o,r)$ be the real circle having center at $o(0,0)$ and radius $r\inz^{+}$.
\pb{Let $\alpha(x_\alpha,y_\alpha)$ and $\beta(x_\beta,y_\beta)$ be two points on the boundary of 
$\rcir(o,r)$ such that $x_\alpha$ and $x_\beta$ are integers.
Let $\arc{\alpha}{\beta}$ be the arc of smaller length between the two arcs of $\rcir(o,r)$ with endpoints 
$\alpha$ and $\beta$.}
Then the {\em sagitta}~\cite{weis_sa} of the circular arc $\arc{\alpha}{\beta}$ is the straight line segment drawn perpendicular to the chord ${\alpha\beta}$, which connects the midpoint $\mu(x_\mu,y_\mu)$ of ${\alpha\beta}$ with $\arc{\alpha}{\beta}$, as shown in Figure~\ref{fig:sagitta-dev}.
Let the sagitta intersect the arc $\arc{\alpha}{\beta}$ at $\gamma$; then

\begin{eqnarray}
&r^2&=(r-d(\mu,\gamma))^2+d^2(\alpha,\mu)\nonumber\\
&\Rightarrow 2rd(\mu,\gamma)&=\left(\frac{d(\alpha,\beta)}{2}\right)^2+d^2(\mu,\gamma)\nonumber\\
&\Rightarrow r&=\frac{d^2(\alpha,\beta)}{8d(\mu,\gamma)} + \frac{d(\mu,\gamma)}{2}\label{eqn:sagitta}
\end{eqnarray}

The above equation is known as the {\em sagitta property} for real circle, where $d(\alpha,\beta)$ denotes the Euclidean distance between the points $\alpha$ and $\beta$, and $d(\mu,\gamma)$ is that between $\mu$ and $\gamma$.
The sagitta property gives the radius of the arc, irrespective of its length in real domain. 
However, this is not true for a digital (circular) arc, as shown below. 

Let $\dcir(o,r)$ be the digital circle corresponding to $\rcir(o,r)$,
and $\darc{a}{b}$ be the digital arc corresponding to $\arc{\alpha}{\beta}$,
where $a$, $b$, $c$, and $m$ are the respective integer points corresponding to $\alpha,\beta,\gamma$, and $\mu$.
The sagitta of the digital arc $\darc{a}{b}$ is $\tilde s={mc}$ and hence using the sagitta property the estimated radius of the arc is $\tilde r=\frac{d^2(a,b)}{8d(m,c)} + \frac{d(m,c)}{2}$.
Then, depending on the value of $d(a,b)$ and $d(m,c)$, the radius $\tilde r$ varies 
between  $r-\delta$ and $r+\delta$, \pb{$\delta$ being the {\em absolute error} in radius estimation.}
In~particular, we have the following theorem.

\begin{theorem}\label{thm:sagitta}
The relative error $\varepsilon_r$ in radius estimation using sagitta property is given by
\begin{eqnarray}
    \frac{|r-\tilde{r}|}{r}\leqslant\left|1-\frac{\tilde{s}}{2r}\right|,
\end{eqnarray}
where $\tilde{r}$ is the estimated radius using sagitta property and $\tilde{s}$ is the length of the \pb{discrete sagitta.}
\end{theorem}

\begin{proof}
Observe that $2\tilde{r}\geqslant d(m,c)=\tilde{s}$.
So, relative error in radius estimation is given by
$\varepsilon_r=\frac{|r-\tilde{r}|}{r}\leqslant\left|1-\frac{\tilde{s}}{2r}\right|$.
\end{proof}

\enlargethispage*{-20pt}

\renewcommand{\tabcolsep}{2pt}

\begin{figure}[!t]\center
\begin{tabular}{@{\,}c@{\,}c@{\,}}
\fbox{\includegraphics[width=.48\textwidth]{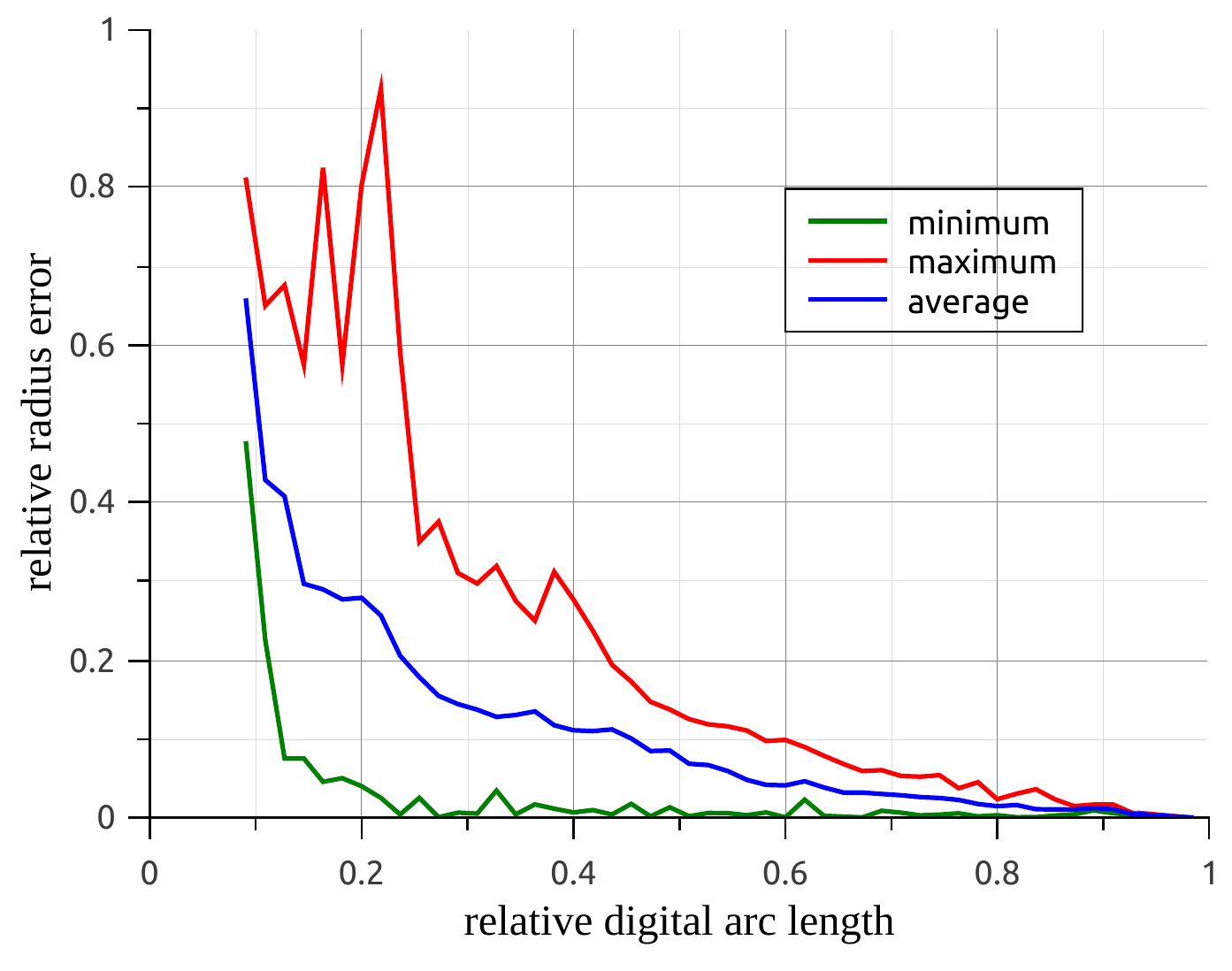}}&
\fbox{\includegraphics[width=.48\textwidth]{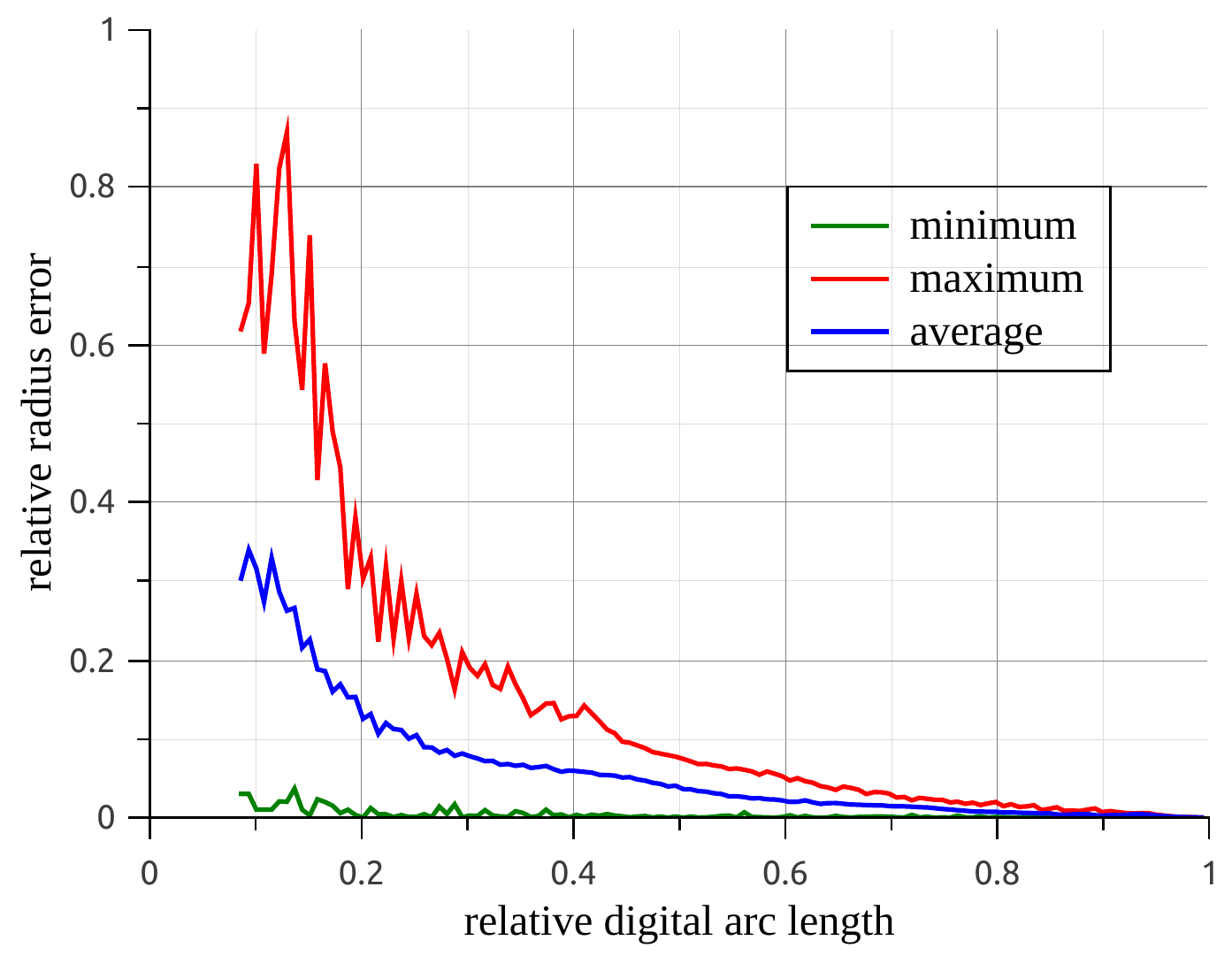}}\\
\fbox{\includegraphics[width=.48\textwidth,height=.398\textwidth]{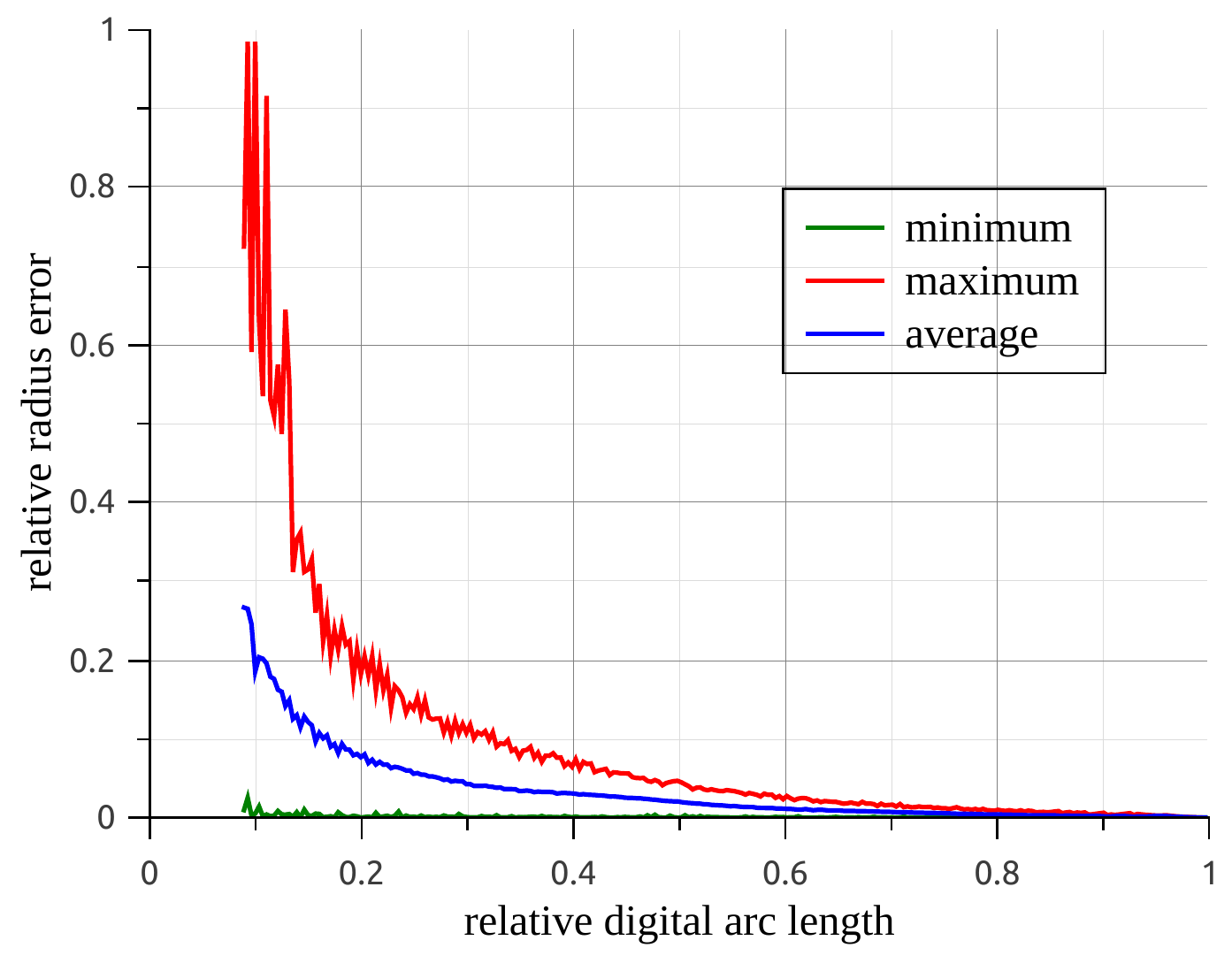}}&
\fbox{\includegraphics[width=.48\textwidth]{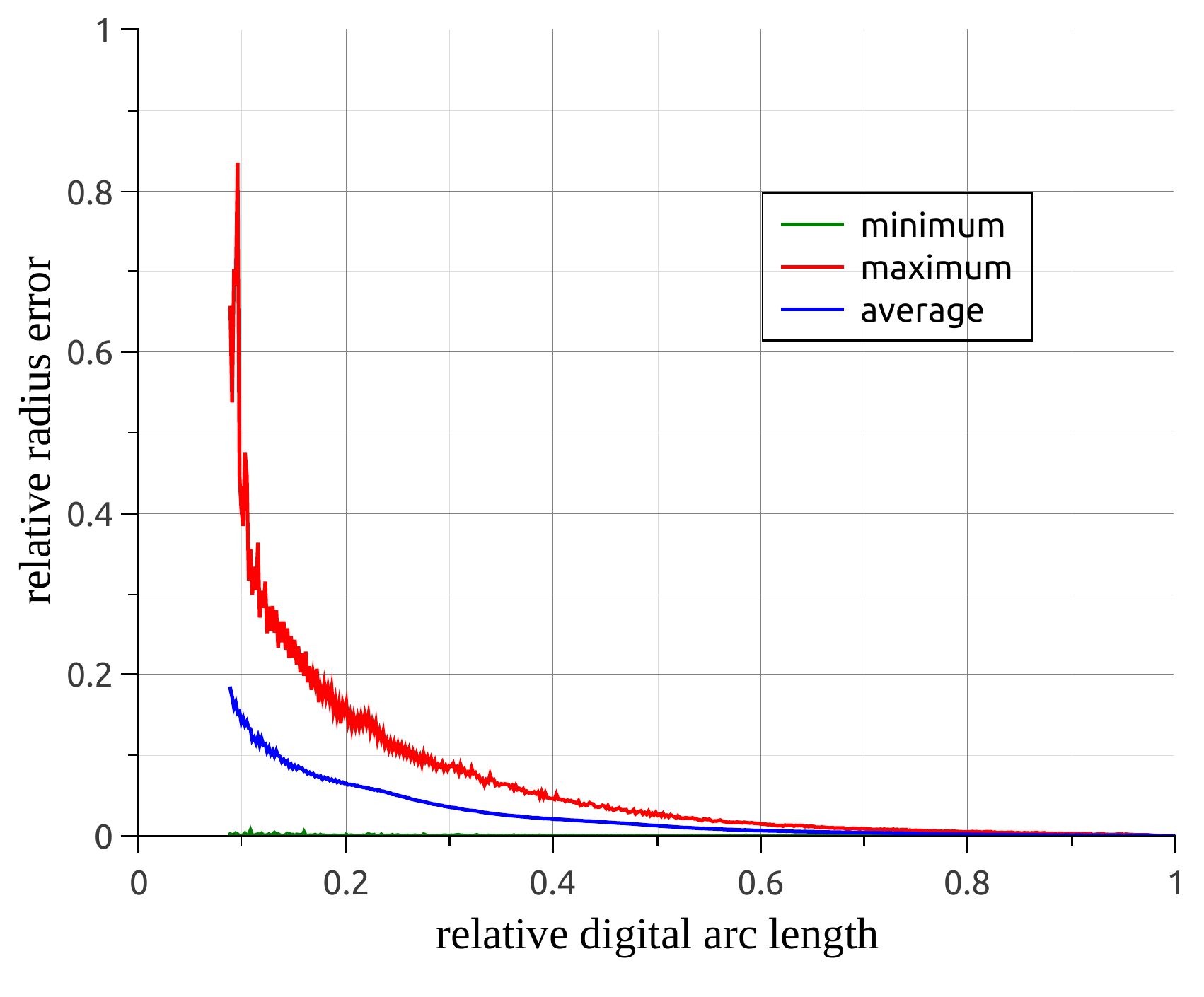}}
\end{tabular}
\caption{Relative radius error $\varepsilon_r$ versus relative digital arc length for $r=20,50,100,200$,
as estimated using sagitta property.}
\label{fig:sagitta:err}
\end{figure}

Theorem~\ref{thm:sagitta} implies that for a fixed value of $r$, as $\tilde{s}$ increases, the relative radius error decreases.
But if $\tilde{s}$ is small relative to $r$, then the relative radius error would be large.
Maximum possible value of $\tilde{s}$ is $r$, and then the relative radius error is $0$.
Figure~\ref{fig:sagitta:err} shows a set of four plots of relative radius error versus relative arc lengths for four different digital circular arcs of radius $20$, $50$, $100$, and $200$.
The {\em relative arc length} is defined as the ratio of the length of the digital arc to its semi-perimeter, measured in number of pixels.
Minimum relative error, maximum relative error, and average relative error are plotted with green, red, and blue colors respectively.
It is clear from these plots that for a fixed radius, as the arc length increases, the corresponding sagitta $\tilde{s}$ increases in length, and so the relative error decreases, which we have proved in Theorem~\ref{thm:sagitta}.
This property is considered in our algorithm to obtain the desired result of circular arc detection.

\section{CSA: Proposed Algorithm by Chord and Sagitta Analysis}
\label{sec:proposed}

Let $\mathcal I$ be the input digital image containing various digital curves like straight line segments, circles, and circular arcs.
Our algorithm checks each curve segment separately for its circularity.
For this, it first extracts all the curve segments.
As  the image $\mathcal I$ may contain thick curves segments, we use thinning~\cite{gonz_01} as preprocessing before applying the algorithm.
The subsequent steps are as follows.

\subsection{\pb{Removing the Straight Segments}}
\label{ssec:straight}
The digital curve segments are first extracted from the thinned image and stored in a list of segments, ${\mathcal L}$.
Each entry in~${\mathcal L}$ contains the coordinates of two endpoints defining the curve segment, and a pointer to the list of curve points.
The center and the radius are also stored in it after their computation.
To identify the circular arc segments, we first remove the digital straight line segments from the list~${\mathcal L}$.
It may be mentioned that there are several techniques to determine digital straightness available in the literature \cite{bhow_07e,klette_04,klette_04a,rosin_97}.
We have used the concept of area deviation~\cite{wall_84}, which is realizable in purely integer domain using a few primitive operations only.
The method is as follows.

Let $S:=\langle a=c_1,c_2,\ldots,c_k=b\rangle$ be a digital curve segment with endpoints $a$ and $b$.
Let $c_i$ $(2\leq i\leq k-1)$ be any point on the segment $S$ other than $a$ and $b$.
Let $h_i$ be the distance of the point $c_i$ from the real straight line segment ${ab}$.
Then $S$ is considered to be a single digital straight line segment starting from $a$ and ending at $b$, provided the following condition is satisfied.

\begin{eqnarray}\label{eqn:straight_line}
\max\limits_{2\leqslant i\leqslant k-1} \left|\triangle\left(a,c_i,b\right)\right|
\leq    \tau_h d_\top\left(a,b\right)
\end{eqnarray}

Here, $\left|\triangle\left(a,c_i,b\right)\right|$ denotes twice the magnitude of area of the triangle with vertices $a:=(x_1,y_1)$, $c_i:=(x_i,y_i)$, and $b:=(x_k,y_k)$, and $d_\top(a,b):=\max(|x_1-x_k|,|y_1-y_k|)$ is the maximum isothetic distance between the points $a$ and $b$, and $\tau_h=2$ in our experiments.
Since all these points are in two-dimensional digital space, the above measures are computable in the integer domain as follows.
\begin{equation}\label{eqn:Delta}
\triangle\left(a,c_i,b\right)=\left|
\begin{array}{ccc}
1 & 1 & 1\\ x_1 & x_i & x_k\\ y_1 & y_i & y_k
\end{array}
\right|
\end{equation}
As $\triangle\left(a,c_i,b\right)$ gives twice the signed area of the triangle with vertices $a$, $c_i$, and $b$,
the digital curve segment $S$ forms a single straight line segment, provided the (maximum) area of the triangle
having $ab$ as the base and the third vertex as the point of $S$ farthest from $ab$, does not exceed the area of the triangle with base-length $d_\top(a,b)$ (isothetic length) and height $\tau_h$.

\subsection{Verifying the Circularity}
\label{ssec:verify_circularity}
\begin{figure}[!t]
\centering
\includegraphics[width=4.30in]{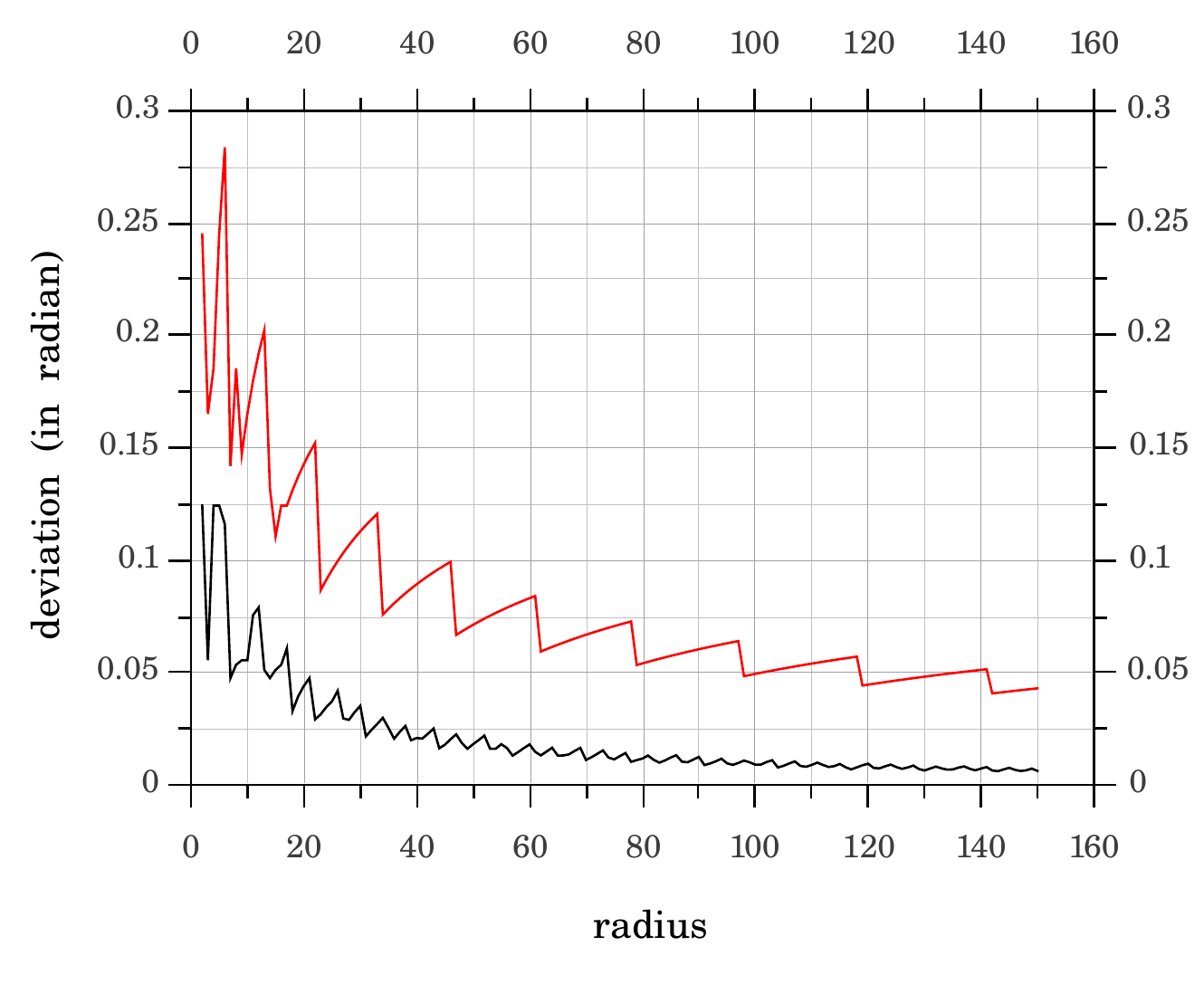}
\caption{Plot of maximum deviation of circumferential angle (in radian) of a digital circle from its real
counterpart. The black curve corresponds to the central region, and red to the remaining part of
the semi-circle excepting its two endpoints.}
\label{fig:dev-angle}
\end{figure}

As explained in Sec.~\ref{ssec:straight}, \pb{after removal of the digitally straight segments,} the remaining segments present in the list ${\mathcal L}$ are not digitally straight.
That is, each segment $S$ in ${\mathcal L}$ is made of \pb{one or more circular segments with or without one or more intervening straight parts}.
So, for each segment $S$, we check its circularity using the chord property in $\zz$, as explained in Sec.~\ref{ssec:chord}.
If the segment $S$ consists of both circular and straight components, then we extract its circular part(s) only from $S$ and discard its straight portion, store these circular segment(s) in the list ${\mathcal L}$ with necessary updates, and remove the original segment $S$ from ${\mathcal L}$.
As explained in Sec.~\ref{ssec:chord}, circumferential angular deviation is too high near the endpoints of an arc.
So, we exclude some pixels from both ends for chord property checking.
We define the {\em central region} of an arc as the sequence of pixels lying in its central one-third portion.
We verify the circularity for the central region of an arc.
The remaining points (i.e., one-third from either end) of $S$ are disregarded from circularity test as they are prone to high deviation of the chord property.
The count of pixels in a semicircle of radius $1$ is $3$, and hence it is trivially accepted.
However, such occurrences are not found as valid circles or circular arcs in a digital drawing. 
The count of pixels in a semicircle of radius $2$ is $7$, and hence the chord property is checked for arcs having length $\tau_c=7$ or more.
After deleting the arcs of length less than $\tau_c$ from the list ${\mathcal L}$, the chord property is checked for each of the remaining arcs.
Let $S:=\langle a=c_1,c_2,\ldots,c_k=b\rangle$ be an arc in the list ${\mathcal L}$.
We verify the circularity for the central region of $S$, namely
$S':=\left\langle c_{\lfloor k/3\rfloor}, c_{\lfloor k/3\rfloor+1},\ldots,c_{m-1},c_m,c_{m+1},
\ldots,c_{\lfloor 2k/3\rfloor-1},c_{\lfloor 2k/3\rfloor} \right\rangle$.
Hence, if $c_m$ ($m=\lfloor k/2\rfloor$) is the midpoint of $S$ and the angle subtended by the chord ${ab}$ at $m$ is estimated to be $\phi_m$, then $S$ is considered to be satisfying the chord property in $\zz$, provided the angle $\phi_c$ subtended by ${ab}$ at each point $c\in S'$ satisfies the following equation.
\begin{eqnarray}
\pb{\max\limits_{c\in S'} \left\{\left|\phi_c-\phi_m\right|\right\} \leq \delta_{\phi}}
\label{eqn:circularity}\end{eqnarray}

As shown in Figure~\ref{fig:dev-angle}, the deviation of a circumferential angle in the central region is less than $\pi/18\approx 0.1745$ radian.
Hence in our experiments, we have taken $\delta_{\phi}=\pi/18$ radian.
If $S$ is not found to be circular, then we divide $S$ into two equal parts and recursively check for digital straightness and the chord property on each part.
The process is continued until a part is smaller than $\tau_c$ in length or it satisfies digital straightness or satisfies chord property.

\subsection{Parameter Estimation}
\label{ssec:combining}
The centers and radii of the detected circular arcs are computed using the sagitta property of the circle, as explained in Sec.~\ref{ssec:sagi}.
While combining the circular arcs, necessary care has to be taken for the inevitable error that creeps~in owing to the usage of sagitta property, which is a property of real circles only.
Since we deal with digital curve segments, the {\em cumulative error} of the effective radius computed for a combined/growing circular arc using the aforesaid sagitta property is very likely to increase with an increase in the number of segments constituting that arc.
Hence, to enhance accuracy, we merge two digital circular segments $S$ and $S'$ into $S'':=S\cup S'$, if (i)\,$S$ and $S'$ have a common endpoint in ${\mathcal L}$ and (ii)\,$S''$ satisfies the {\em chord property}.
Since the node corresponding to each segment in the list ${\mathcal L}$ contains endpoints, center, radius, and a pointer to the list of curve points, the attributes of the segment $S$ are updated by those of $S''$, and the data structure ${\mathcal L}$ is updated accordingly.

\subsection{Parameter Finalization}
\label{ssec:hough}

In spite of the treatments to reduce discretization errors while employing chord property to detect circular arcs and while employing sagitta property to combine two or more circular arcs and to compute the effective radius and center, some error may still be present in the estimated values of the radii and the center.
To remove such errors, we apply a {\em restricted Hough transform} (rHT) on each circular arc $S\in{\mathcal L}$ with a small parameter space~\cite{chen_01b}.
Let $q(x_q,y_q)$ and $r$ be the respective center and radius of $S$ estimated using the sagitta property.
As explained in Sec.~\ref{ssec:sagi}, the relative radius error can have a maximum value of $1$.
Hence, the restricted parameter space is considered as $[x_q-\delta,x_q+\delta]\times[y_q-\delta,y_q+\delta]\times[\tilde{r}-\delta,\tilde{r}+\delta]$, where $\delta=\tilde{r}$.
A 3D integer array, $H$, is taken corresponding to this parameter space, each of whose entry is initialized to zero.
For every three points $c$, $c'$, and $c''$ from $S$, with $c$ lying in its left region, $c'$ lying in its central region, and $c''$ lying in its right region, we estimate the center $q'(x',y')$ and the radius $r'$ of the (real) circle passing through $c$, $c'$, and $c''$.
The corresponding entry in $H$ is incremented accordingly.
Finally, the entry in $H$ corresponding to the maximum frequency provides the final center and radius of $S$.

\subsection{Demonstration of CSA}
\label{sec:demo}

\begin{figure}[!t]\center\footnotesize
\begin{tabular}{@{\,}c@{\,}c@{\,}}
\fbox{\includegraphics[width=.46\textwidth,viewport=40 110 590 550,clip]{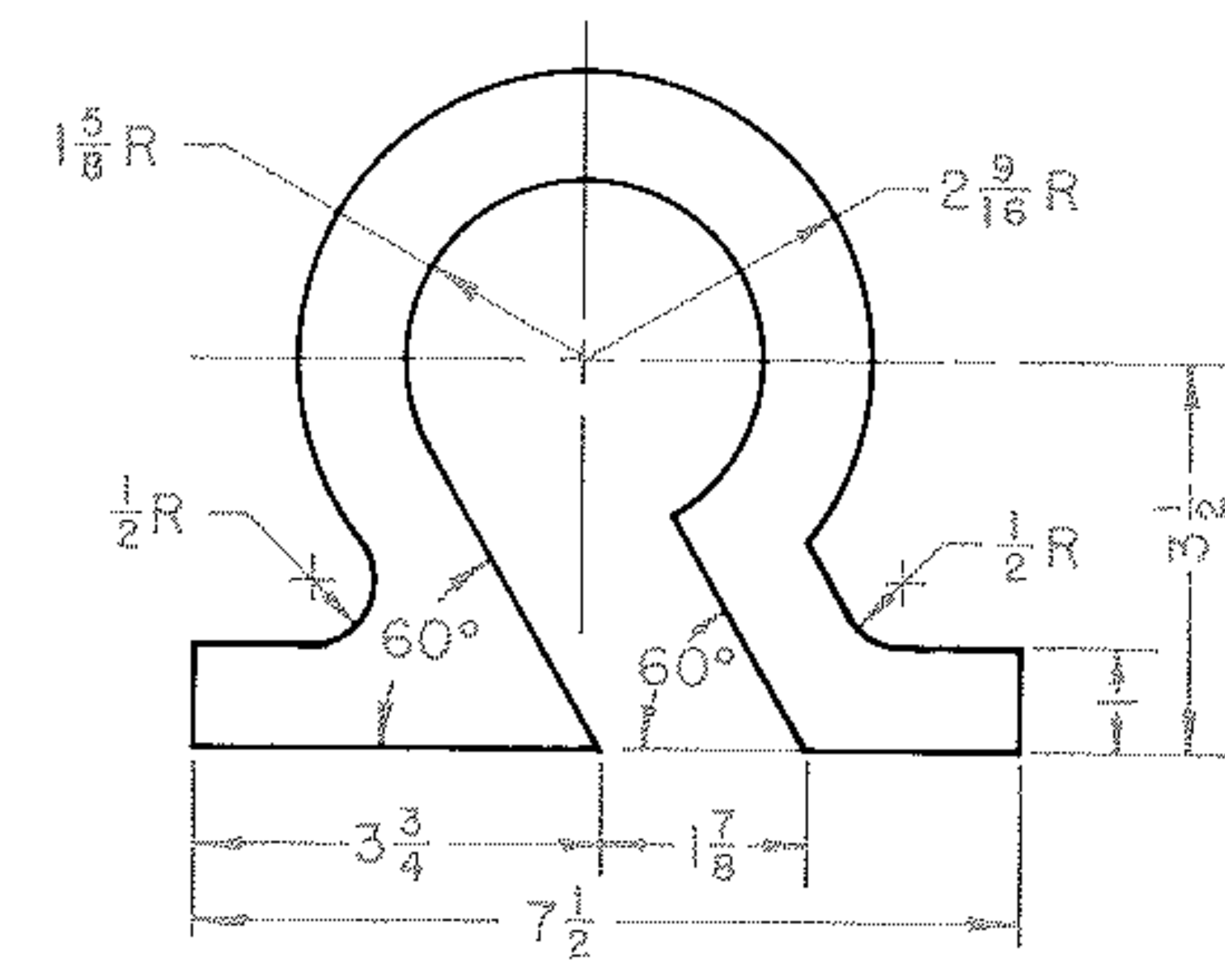}}&
\fbox{\includegraphics[width=.46\textwidth,viewport=40 110 590 550,clip]{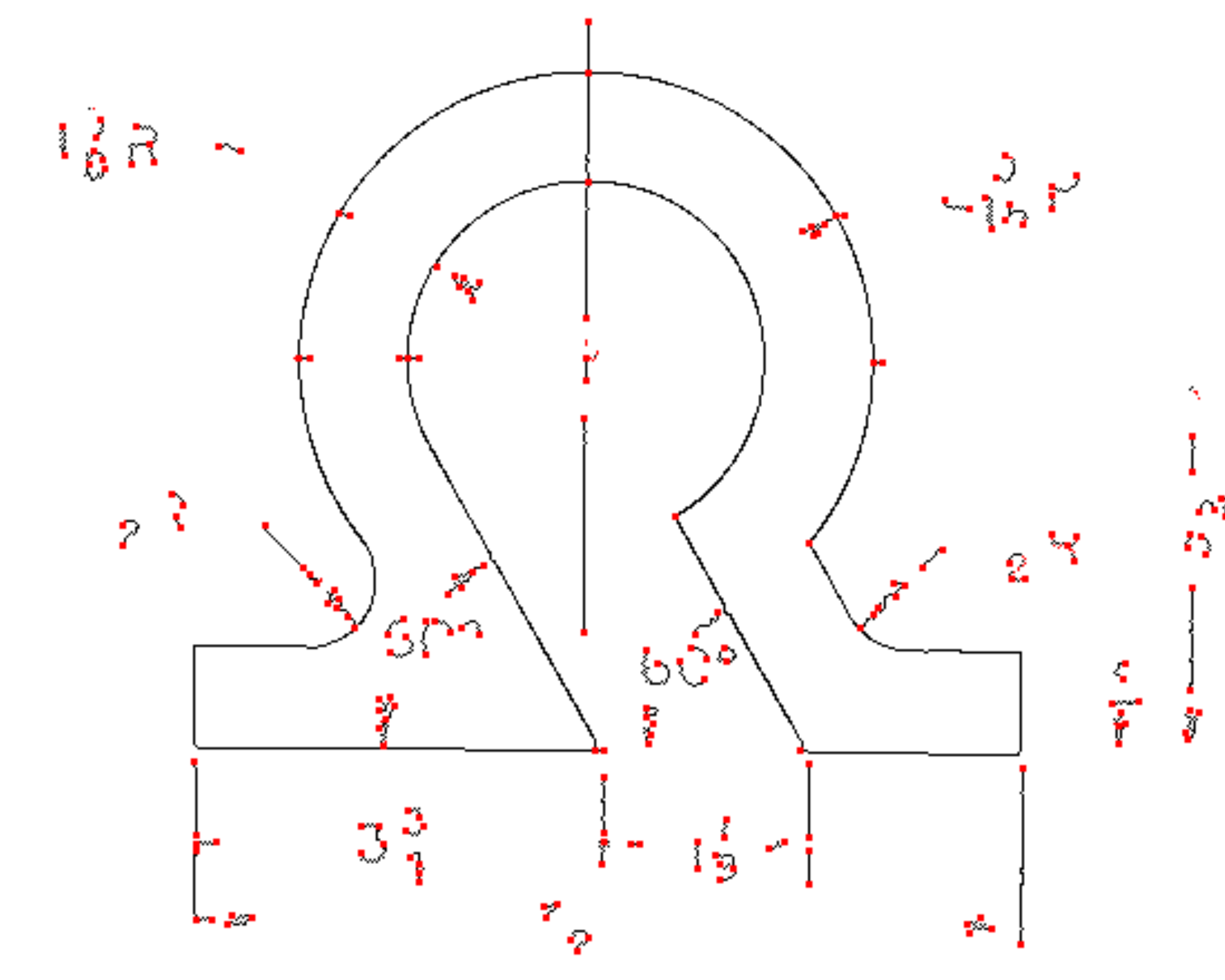}}\\
\parbox[t]{.25\textwidth}{\centering(a)}&\parbox[t]{.25\textwidth}{\centering(b)}\smallskip\\
\fbox{\includegraphics[width=.46\textwidth,viewport=40 110 590 550,clip]{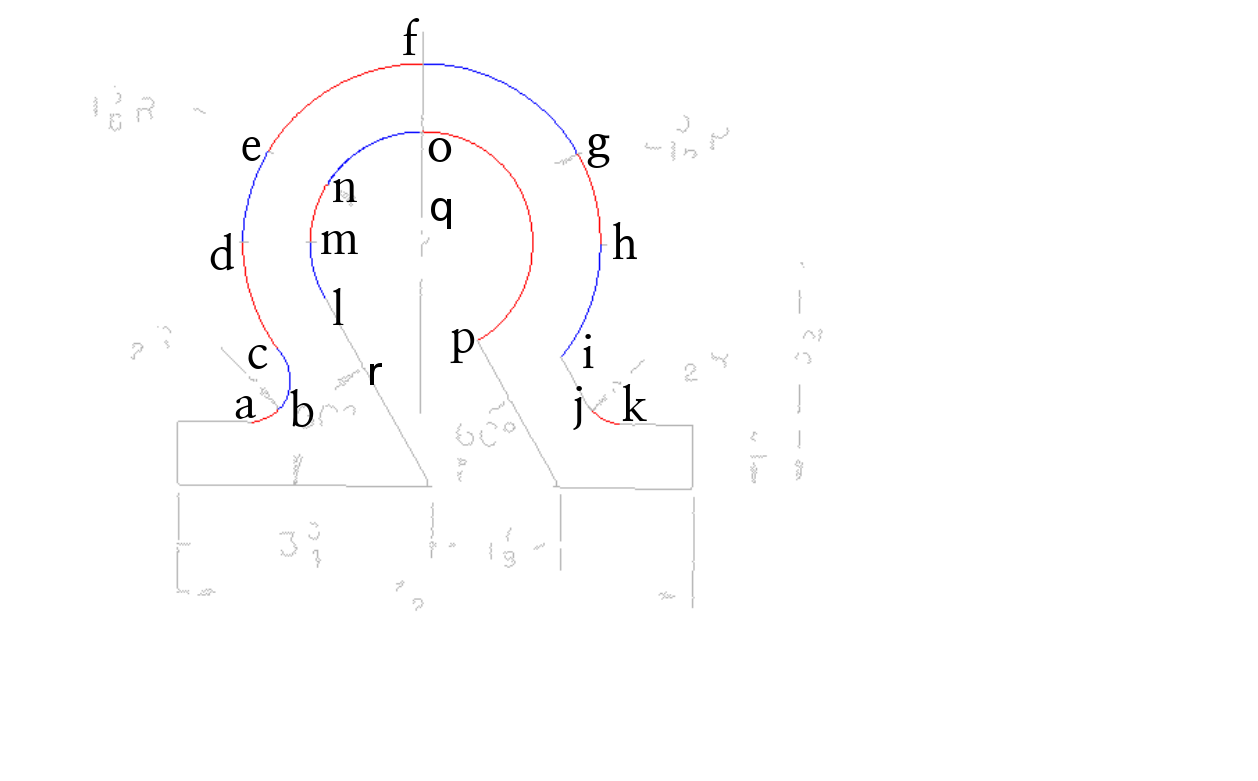}}&
\fbox{\includegraphics[width=.46\textwidth,viewport=40 110 590 550,clip]{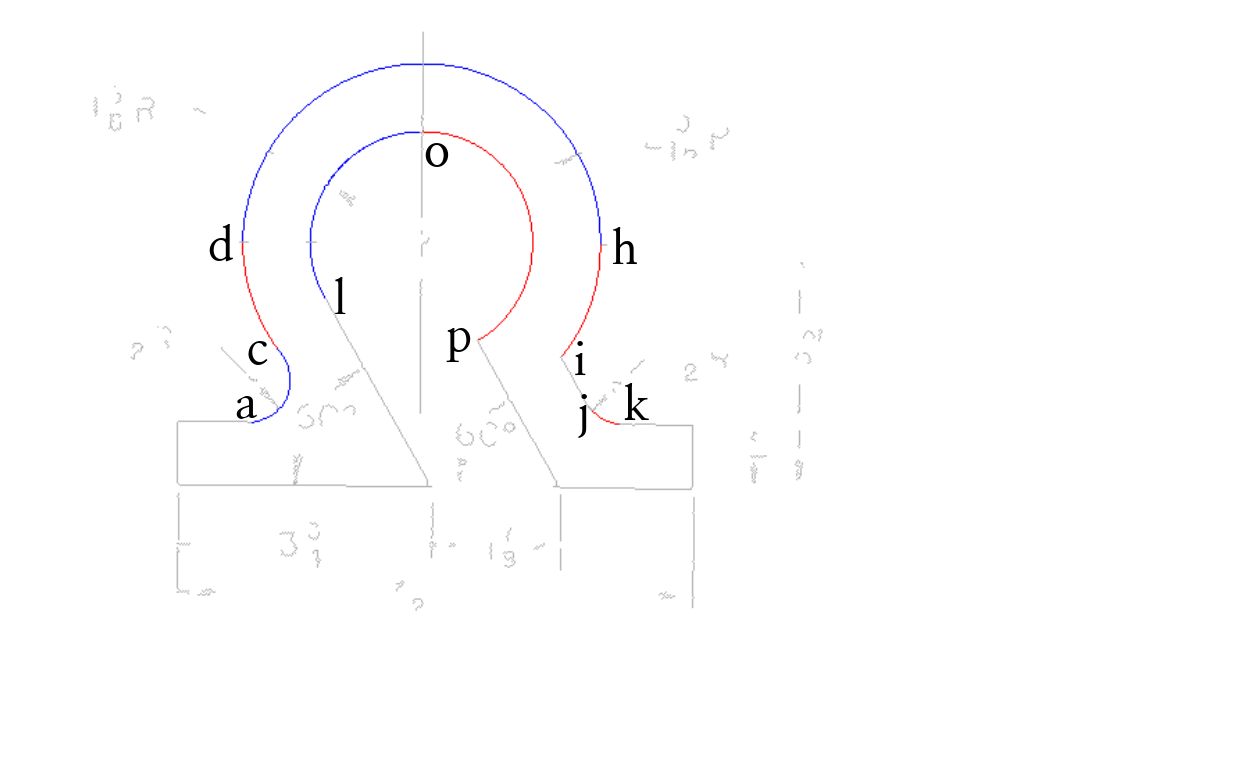}}\\
\parbox[t]{.25\textwidth}{\centering(c)}&\parbox[t]{.25\textwidth}{\centering(d)}\smallskip\\
\fbox{\includegraphics[width=.46\textwidth,viewport=40 110 590 550,clip]{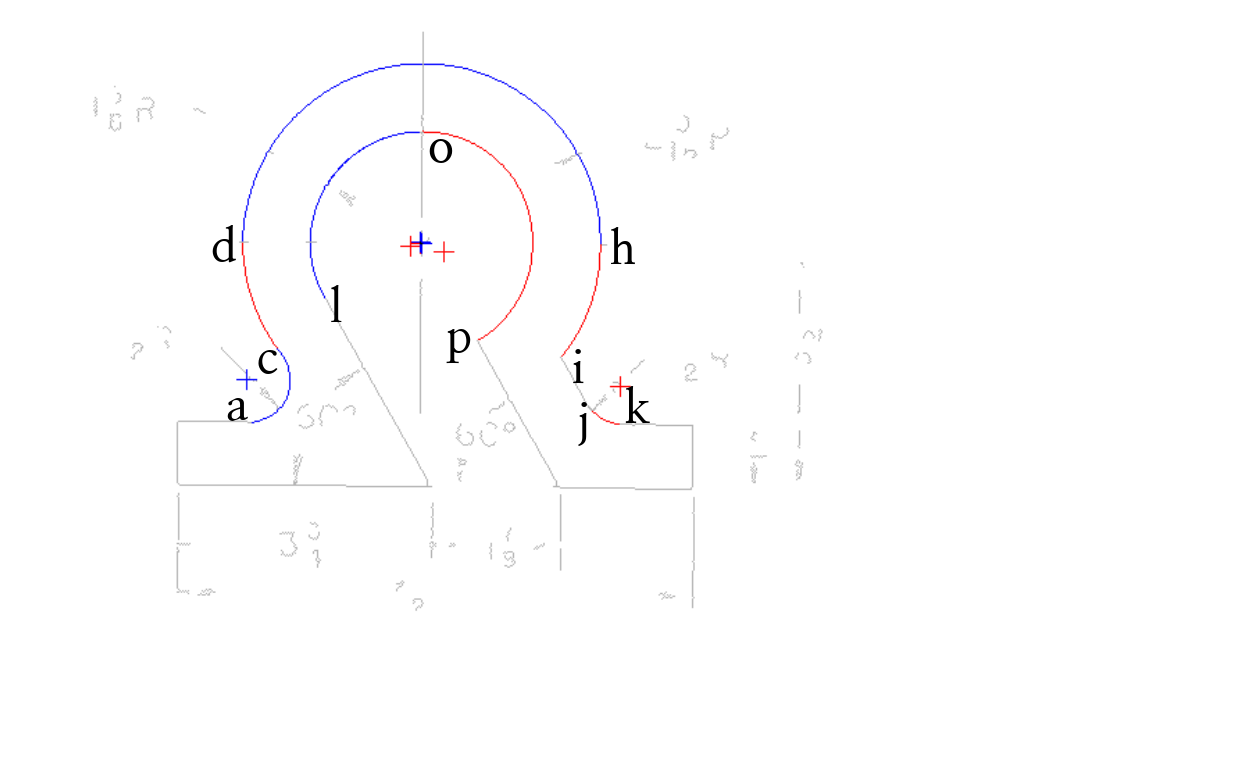}}&
\fbox{\includegraphics[width=.46\textwidth,viewport=40 110 590 550,clip]{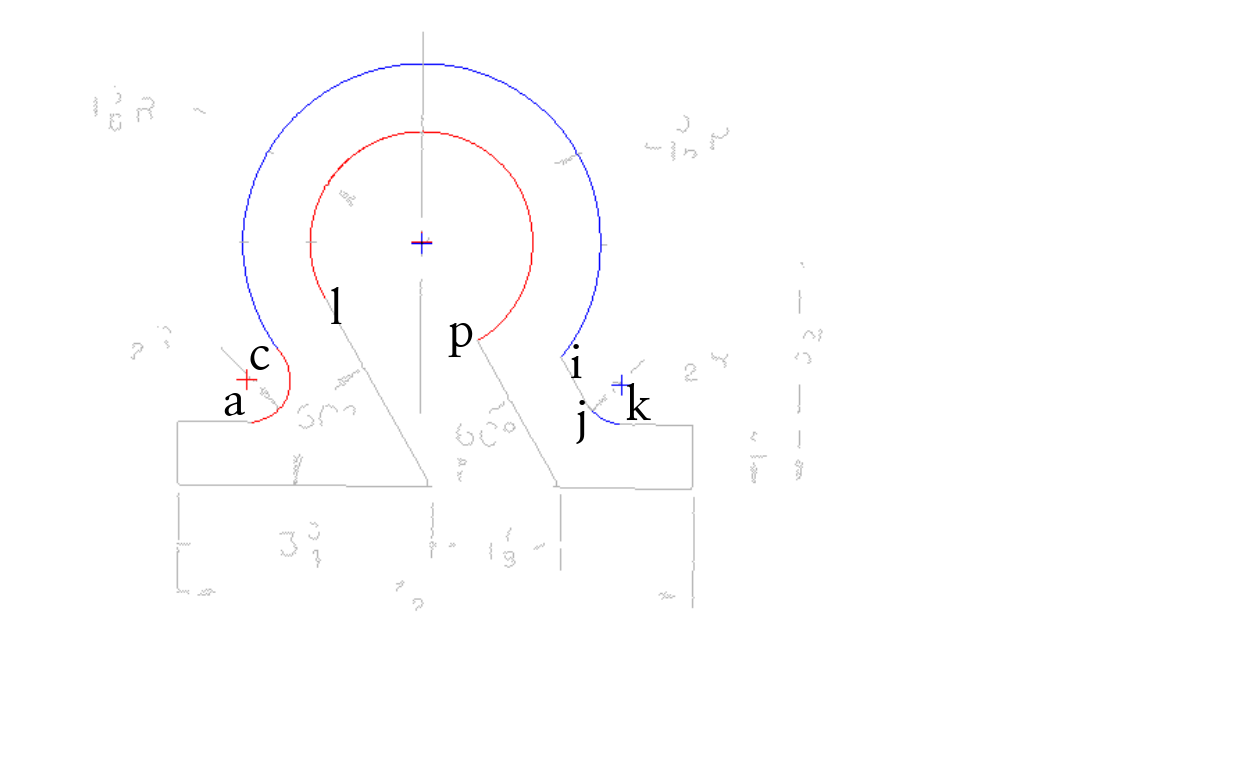}}\\
\parbox[t]{.25\textwidth}{\centering(e)}&
\parbox[t]{.25\textwidth}{\centering(f)}\smallskip\\
\multicolumn{2}{r}{(Continued to next page.)}
\end{tabular}
\label{fig:demo1}
\end{figure}

\begin{figure}[!t]\center\footnotesize
\begin{tabular}{@{\,}c@{\,}c@{\,}}
\fbox{\includegraphics[width=.46\textwidth,viewport=40 110 590 550,clip]{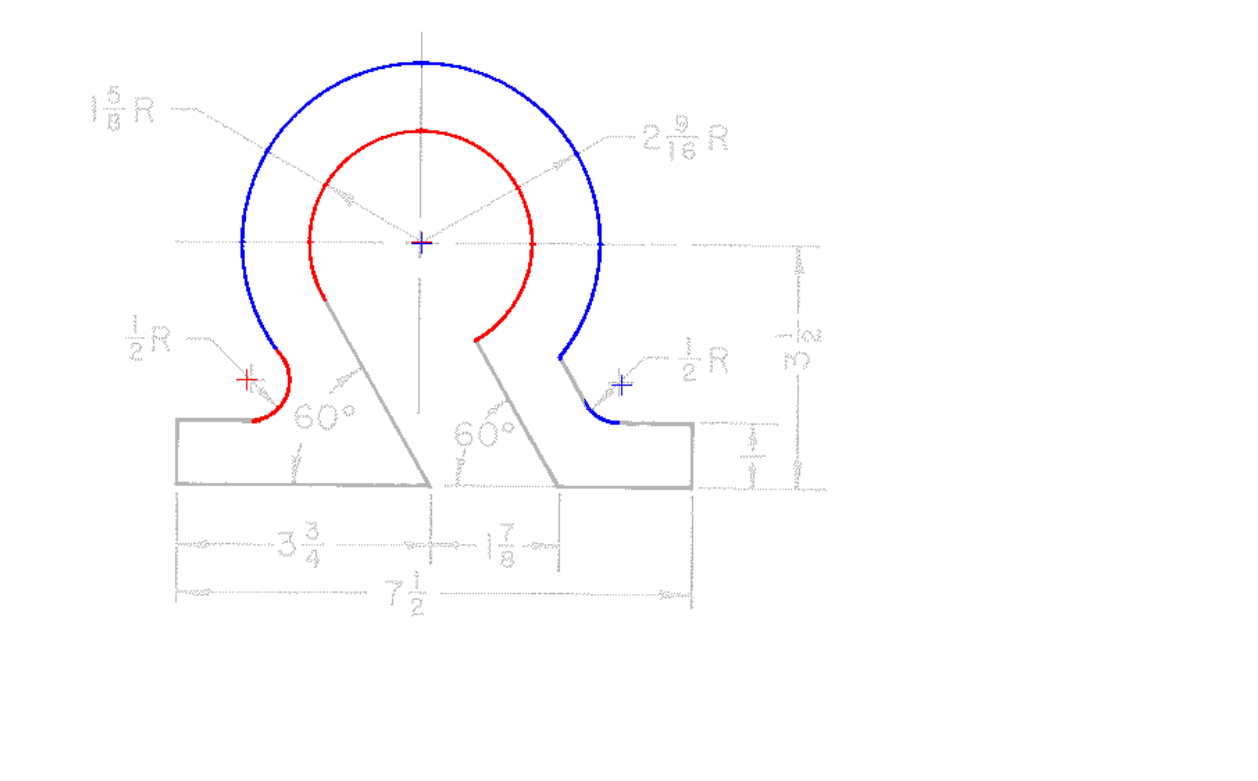}}&
\fbox{\includegraphics[width=.46\textwidth,viewport=40 110 590 550, clip]{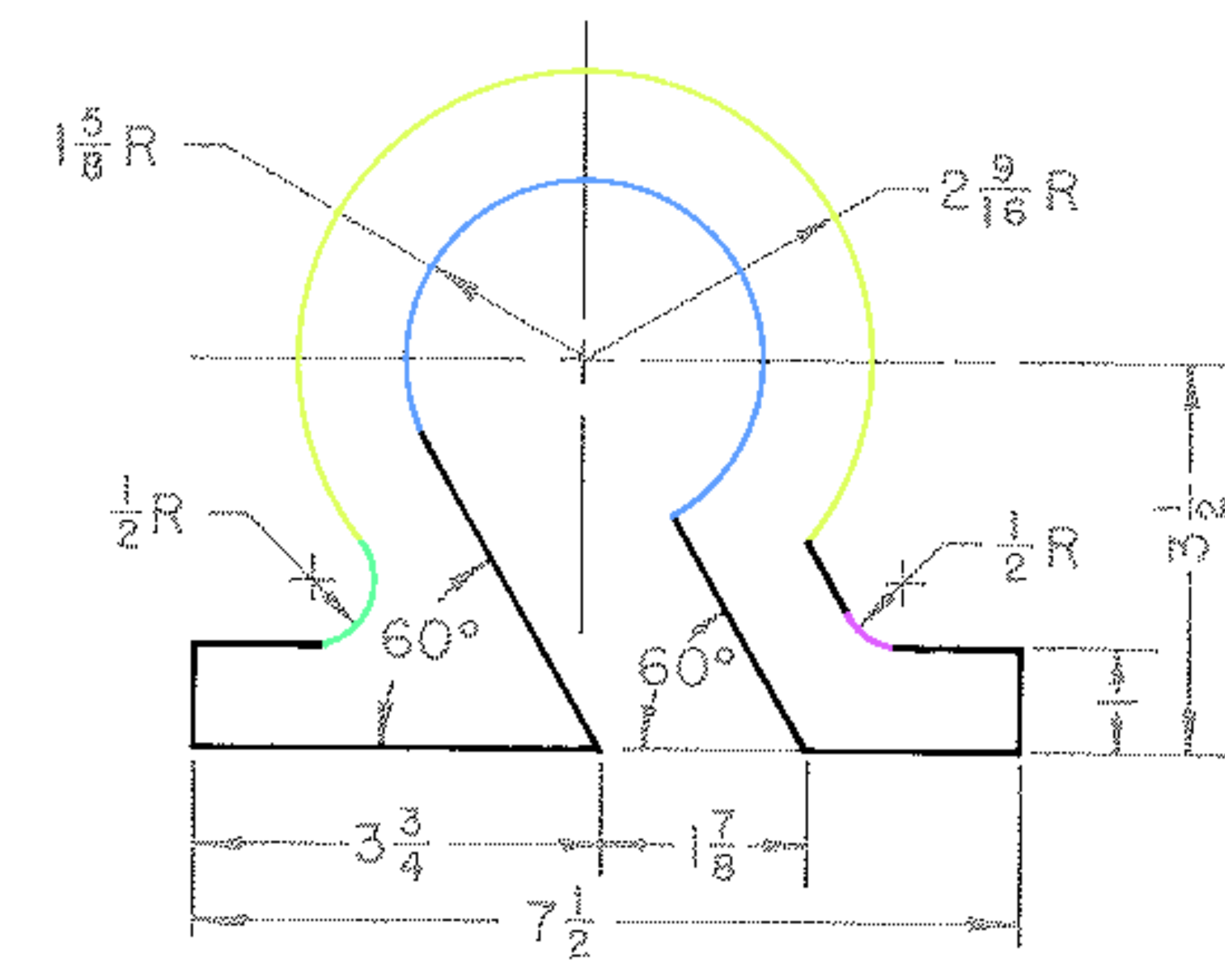}}\\
\parbox[t]{.25\textwidth}{\centering(g)}&
\parbox[t]{.25\textwidth}{\centering(h)}
\end{tabular}
\caption{Step-wise snapshots of the algorithm on {\tt g07-tr6.tif} from GREC2007 dataset~\cite{grec07}: (a)\,input image; (b)\,segments after thinning; (c)\,circular arcs by {\em chord property}; (d)\,after combining adjacent arcs; (e)\,centers by
{\em sagitta property}; (f)\,after applying {\em restricted Hough transform}; (g)\,final result; (h)\,ground-truth;}
\label{fig:demo}
\end{figure}

A demonstration of the proposed algorithm (CSA) on a sample image is shown in Figure~\ref{fig:demo}.
All the digital curve segments in the image are extracted and stored in the list ${\mathcal L}$ (Figure~\ref{fig:demo}(b)).
The straight line segments are removed from ${\mathcal L}$ using the straightness properties.
For example, ${fo}$ and ${oq}$ are two of the straight line segments that are removed (Figure~\ref{fig:demo}(c)).
Then using the {\em chord property}, the circular segments are detected with necessary updates in the list ${\mathcal L}$.
For example, the digital curve segment $db$ consists of two circular segments.
After their extraction, the segment $dc$ is stored in the node of the original segment and the other one, i.e., $cb$, is stored in a newly created node in the list ${\mathcal L}$.
Similarly, for the segment $mr$, the circular arc $ml$ is extracted, inserted in the node of $mr$,  and the straight part ${lr}$ is removed.

Two or more adjacent arcs are combined if they jointly satisfy the {\em chord property} in order to get larger arcs for reducing the computational error in the next step while applying the {\em sagitta property} (Sec.~\ref{ssec:combining}).
After this combining/merging, the number of circular segments gets significantly reduced, as reflected in Figure~\ref{fig:demo}(d).
The radius and  the center of each arc in ${\mathcal L}$ are computed using the {\em sagitta property} and stored in the node of the corresponding arc (Figure~\ref{fig:demo}(e)).
Next, we apply rHT on these arcs (Sec.~\ref{ssec:hough}).
The resultant image is shown in Figure~\ref{fig:demo}(f).
Finally, we consider the detected circular arcs, and for each pixel on a detected
arc, the object pixels in its 8-neighborhood are iteratively marked as pixels of the corresponding circular arc.
Figure~\ref{fig:demo}(g) shows the final circular arcs detected by our algorithm, and the corresponding ground-truth is shown in Figure~\ref{fig:demo}(h).

\section{Experimental Results}
\label{sec:results}

We have implemented our algorithm CSA, and also some existing algorithms (Sec.~\ref{ss:comparison}), in C on the openSUSE{\scriptsize \texttrademark}
OS Release 11.0 HP xw4600 Workstation with Intel{\scriptsize \textregistered} Core{\tiny \texttrademark}2 Duo, 3~GHz processor.
We have performed tests on several datasets including the GREC datasets \cite{grec07,grec13}.
The results of our algorithm on some of the images from these datasets are shown here.

\begin{figure}[!t]
\begin{tabular}{@{\,}c@{\,}c@{\,}c}
\fbox{\includegraphics[width=.31\textwidth]{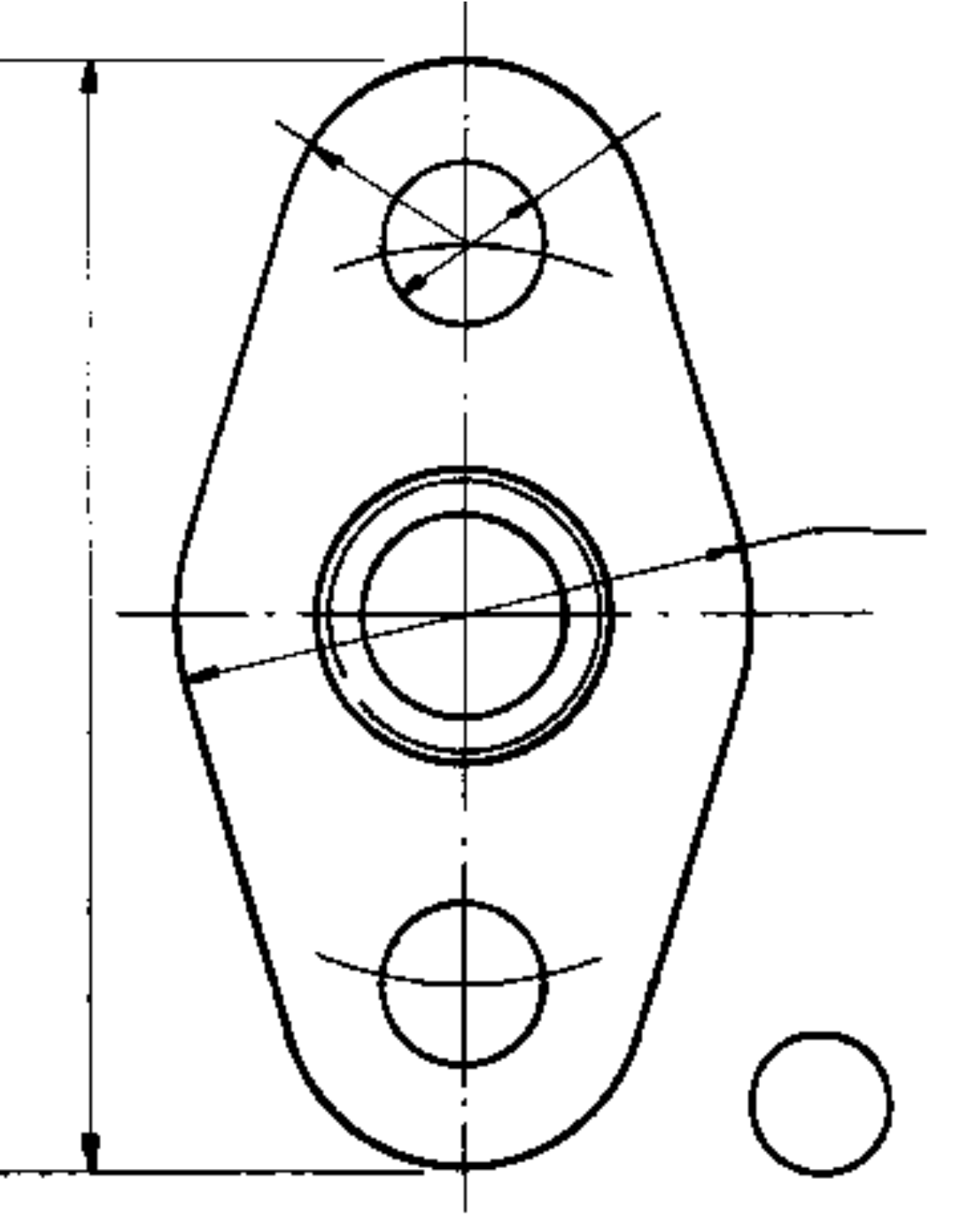}}&
\fbox{\includegraphics[width=.31\textwidth]{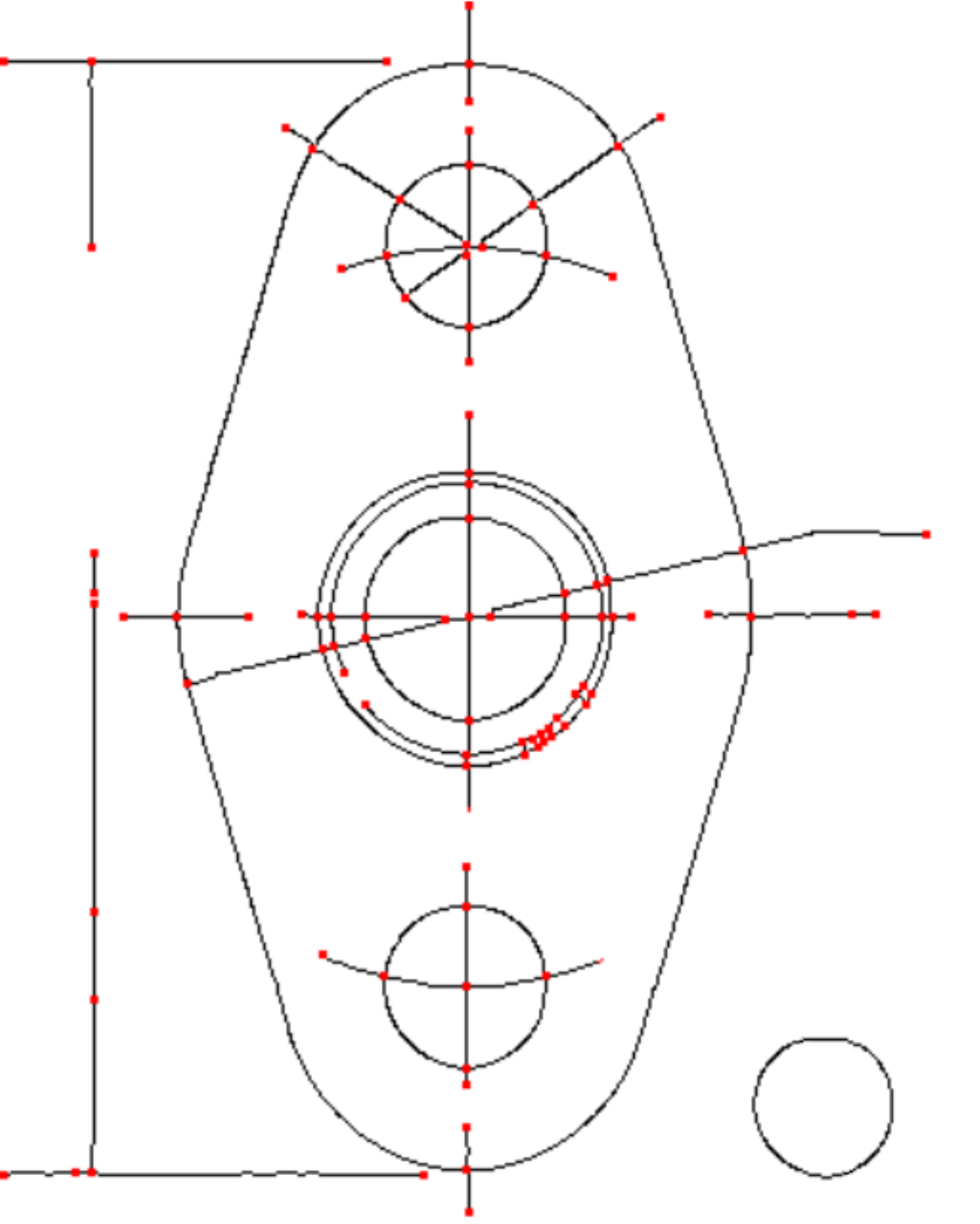}}&
\fbox{\includegraphics[width=.31\textwidth]{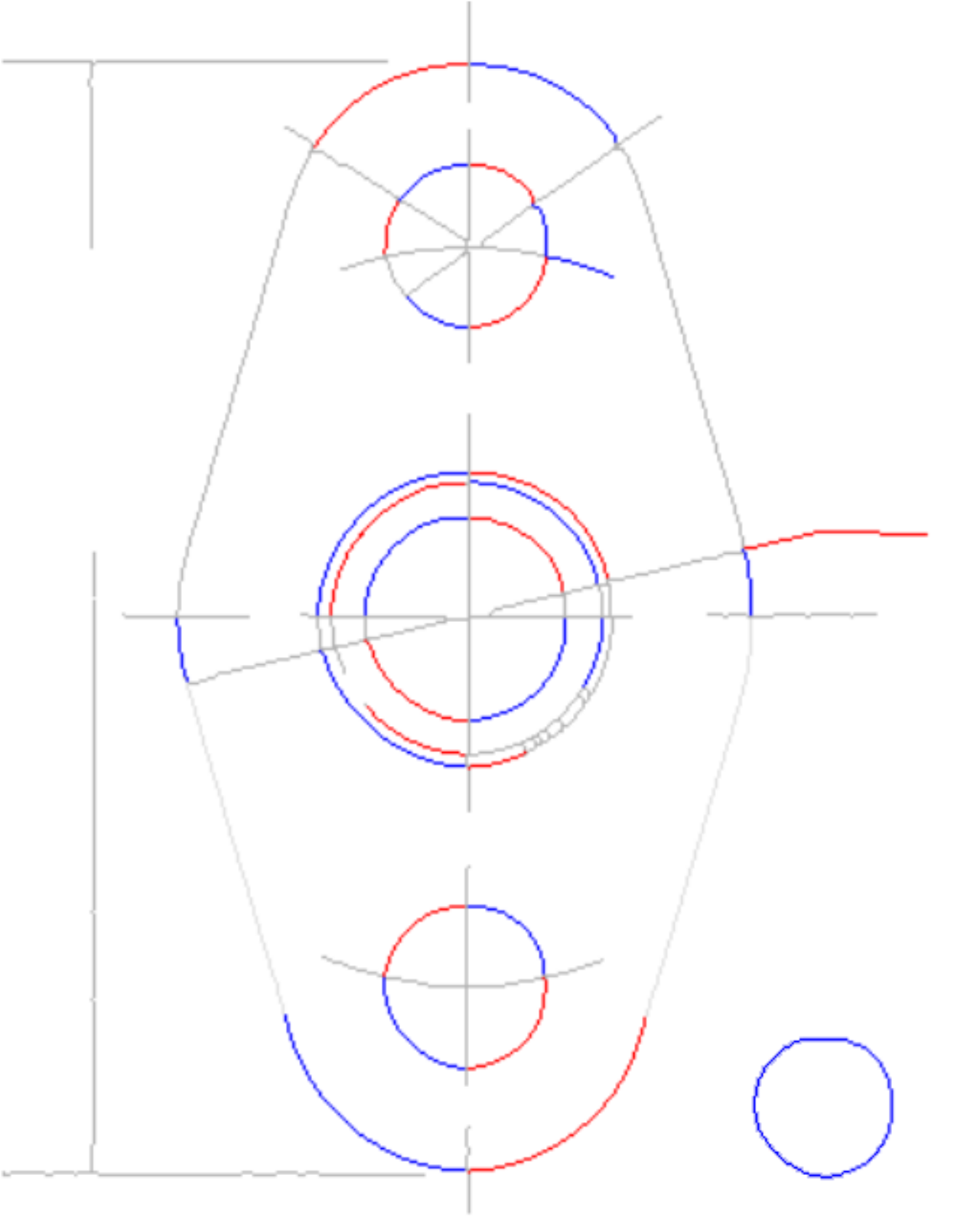}}\\
\parbox[t]{.31\textwidth}{\centering(a)}&
\parbox[t]{.31\textwidth}{\centering(b)}&
\parbox[t]{.31\textwidth}{\centering(c)}
\smallskip\\
\fbox{\includegraphics[width=.31\textwidth]{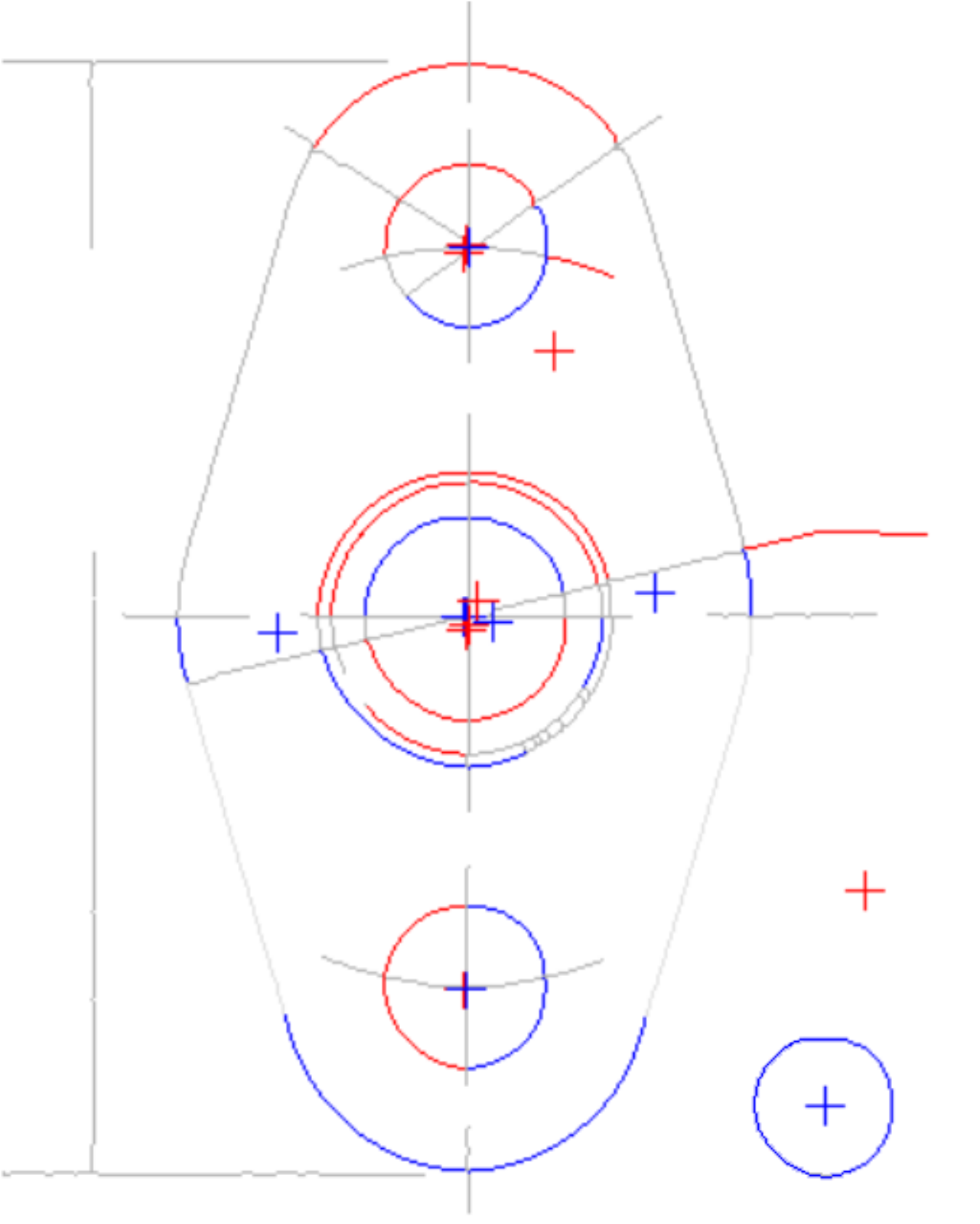}}&
\fbox{\includegraphics[width=.31\textwidth]{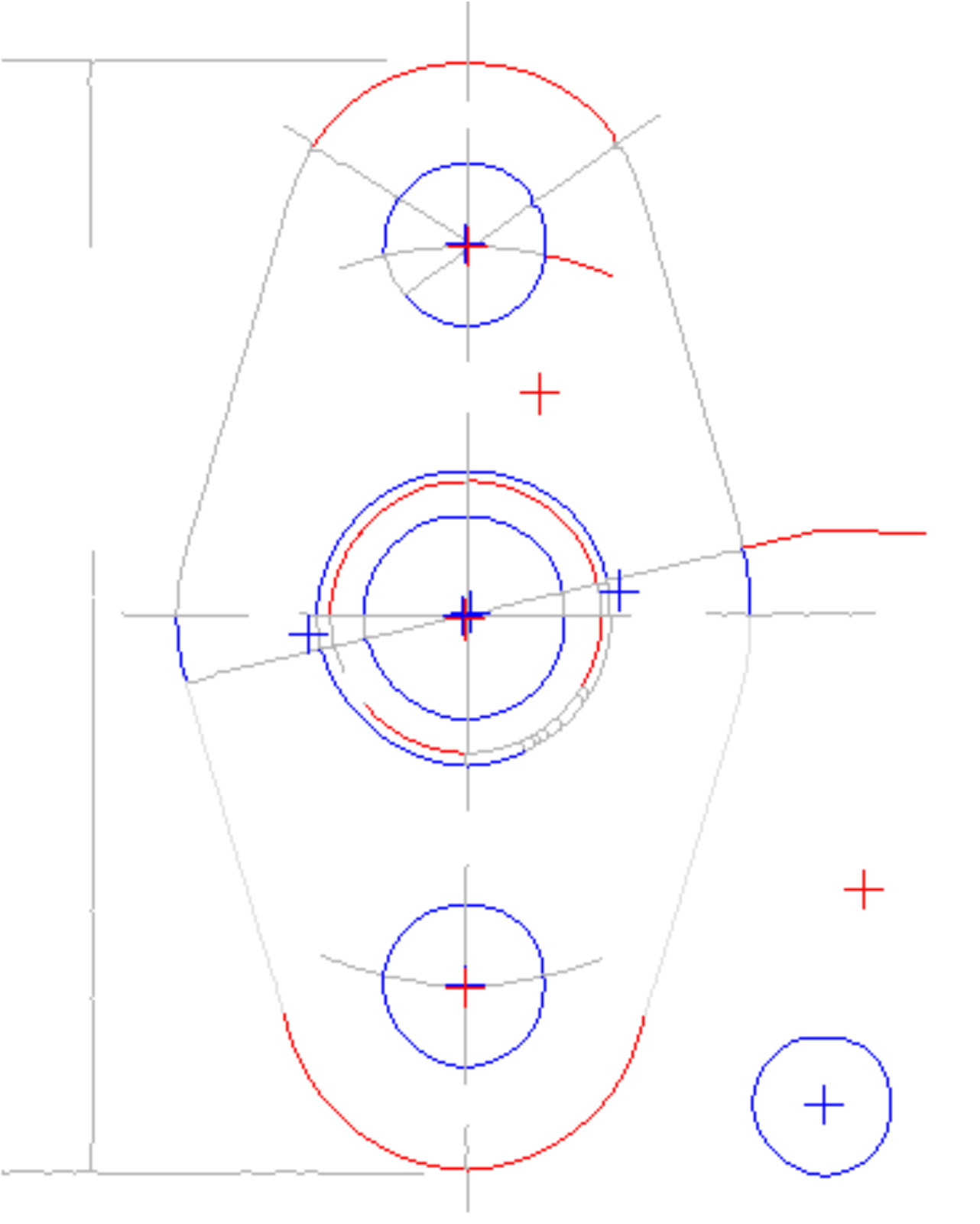}}&
\fbox{\includegraphics[width=.31\textwidth]{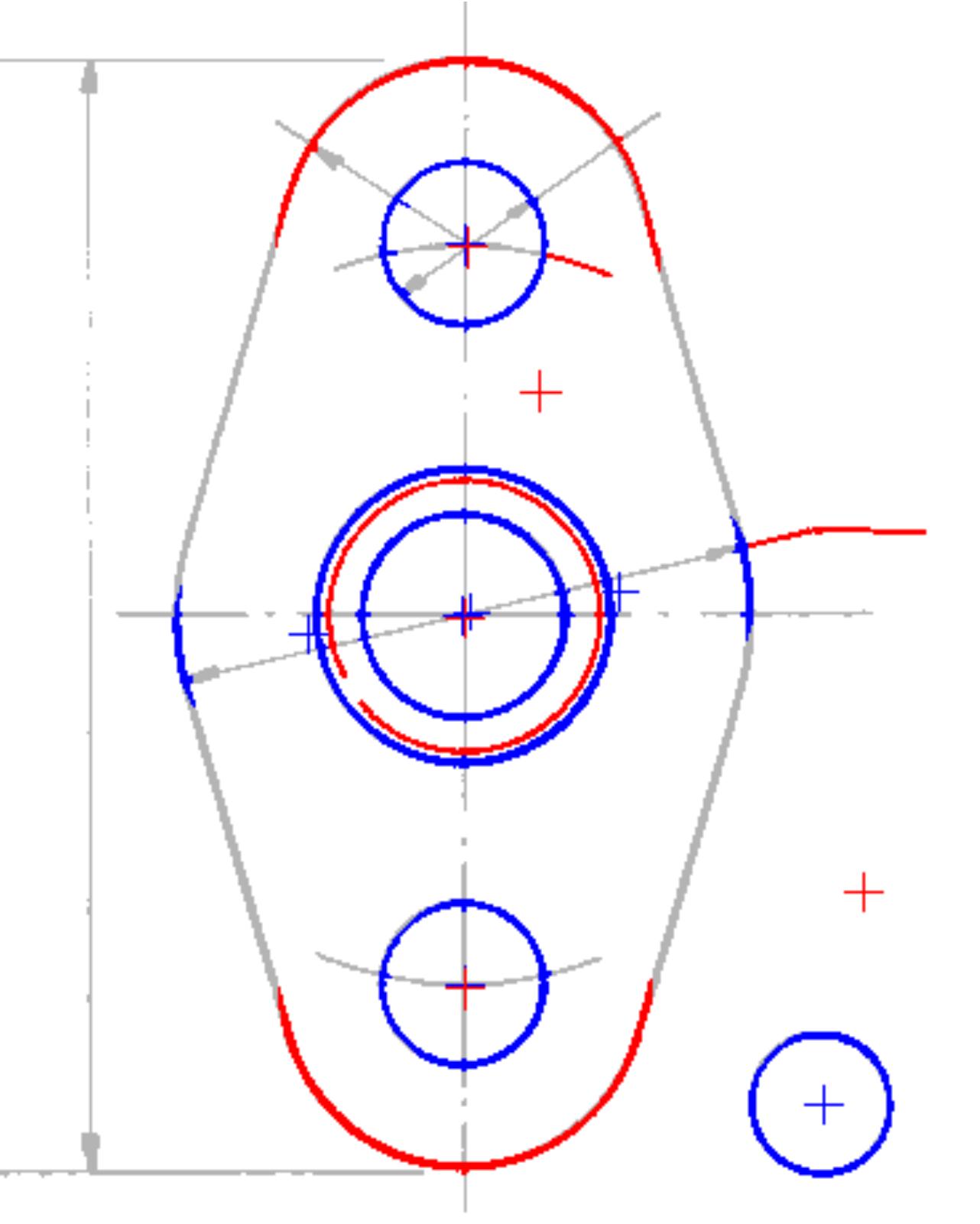}}\\
\parbox[t]{.31\textwidth}{\centering(d)}&
\parbox[t]{.31\textwidth}{\centering(e)}&
\parbox[t]{.31\textwidth}{\centering(f)}
\end{tabular}\medskip\\
\caption{Step-wise snapshots of our experiment on {\tt g07-tr1.tif} from GREC2007 dataset~\cite{grec07}: (a)\,input image; (b)\,segments after thinning; (c)\,circular arcs by {\em chord property} and after combining adjacent arcs; (d)\,centers by
{\em sagitta property}; (e)\,after applying {\em restricted Hough transform}; (f)\,final result.}
\label{fig:g07-tr1}
\end{figure}

{\def\baselinestretch{1.1}
\begin{figure}[!t]\center
\fbox{\includegraphics[width=0.98\textwidth,viewport=0 0 1710 245,clip]{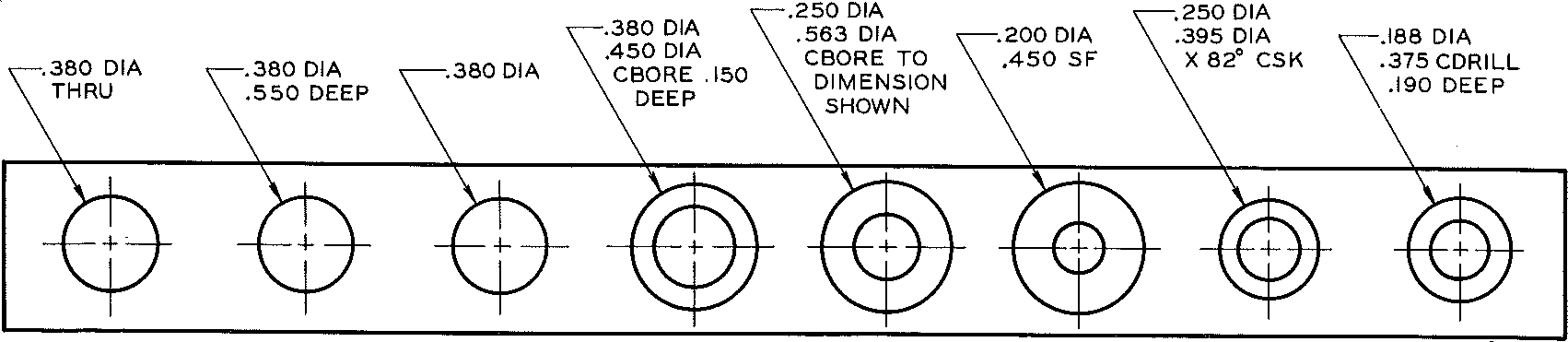}}\\
\,(a)\medskip\\
\fbox{\includegraphics[width=0.98\textwidth,viewport=0 0 1710 245,clip]{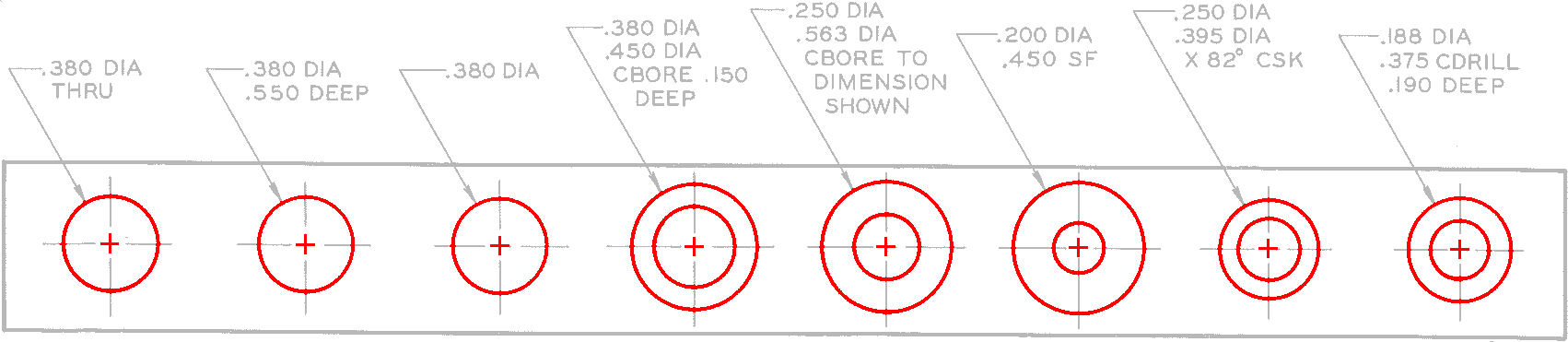}}\\
\,(b)\\
\caption{An example of perfect result by our algorithm on the image cropped from {\tt g07-tr4.tif} of GREC2007 dataset, which contains only full circles.}
\label{fig:g07-tr4}
\end{figure}}

{\def\baselinestretch{1.1}
\begin{figure}[!t]\center
\begin{tabular}{@{}c@{\,}@{\,}c@{}}
\fbox{\includegraphics[width=0.4\textwidth]{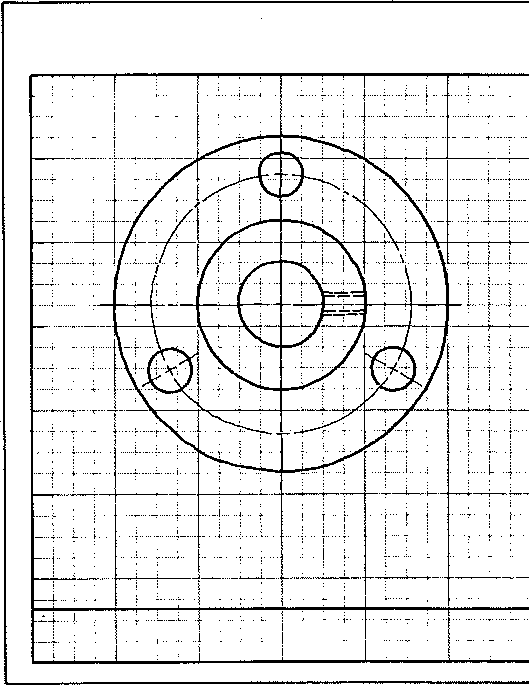}}&
\fbox{\includegraphics[width=0.4\textwidth]{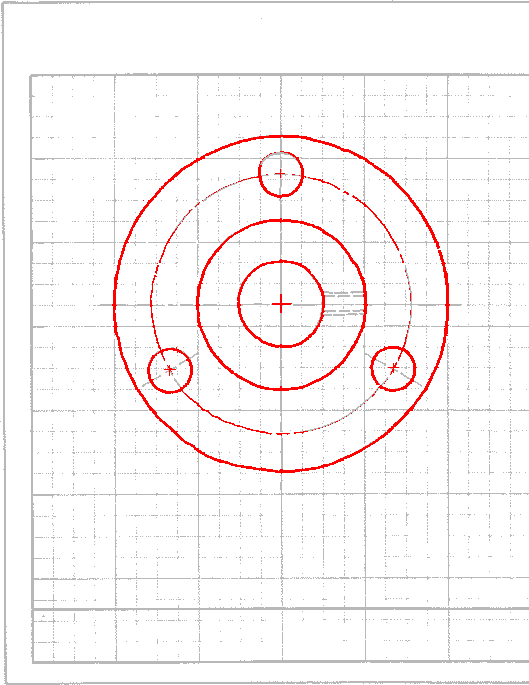}}\\
\,(a)&\,(b)\\
\fbox{\includegraphics[width=0.45\textwidth]{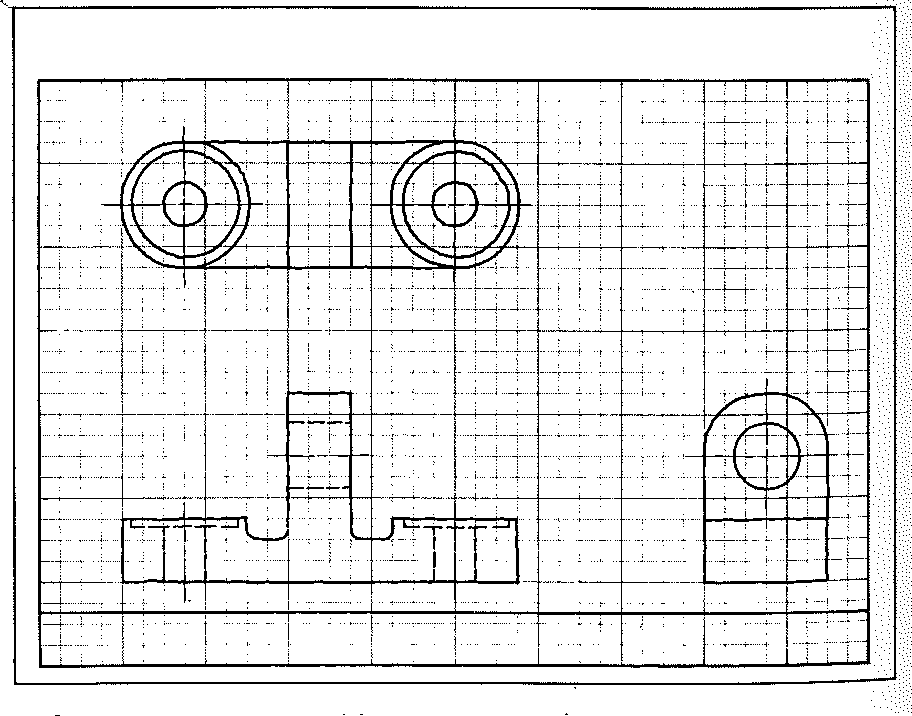}}&
\fbox{\includegraphics[width=0.45\textwidth]{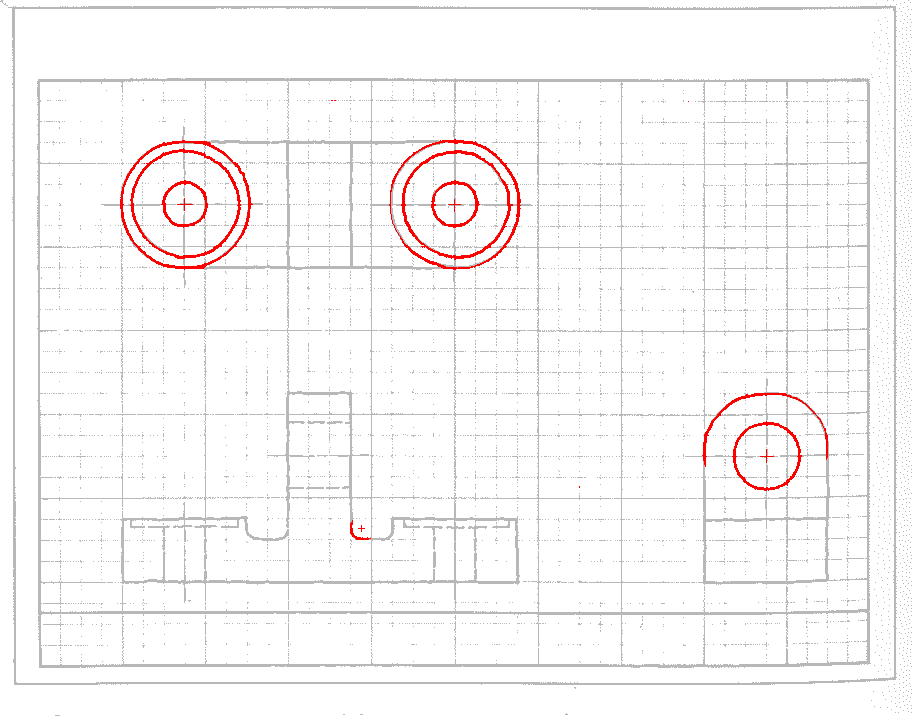}}\\
\,(c)&\,(d)\\
\end{tabular}
\caption{Results by our algorithm on two difficult images from GREC2007 dataset.
{\bf Top:}\,{\tt g07-tr2.tif}. {\bf Bottom:}\,{\tt g07-tr3.tif}. (a) and (c) are the respective inputs, and (b) and (d) are the respective outputs.}
\label{fig:g07-tr23}
\end{figure}}

{\def\baselinestretch{1.1}
\begin{figure}[!t]\center
\begin{tabular}{@{}c@{\,}@{\,}c@{}}
\includegraphics[width=0.38\textwidth]{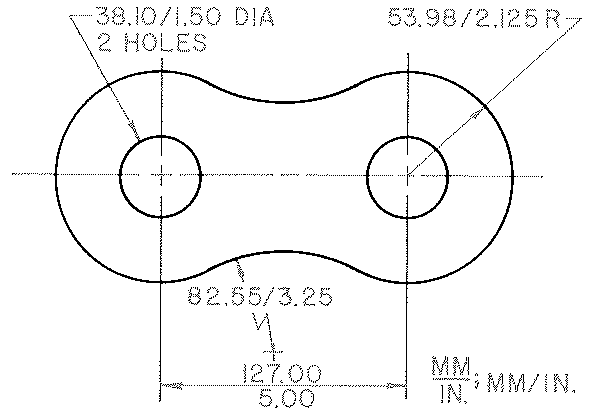}&
\includegraphics[width=0.38\textwidth]{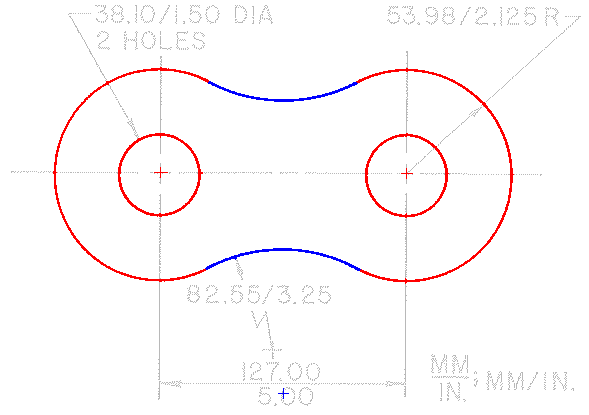}\\
\includegraphics[width=0.45\textwidth]{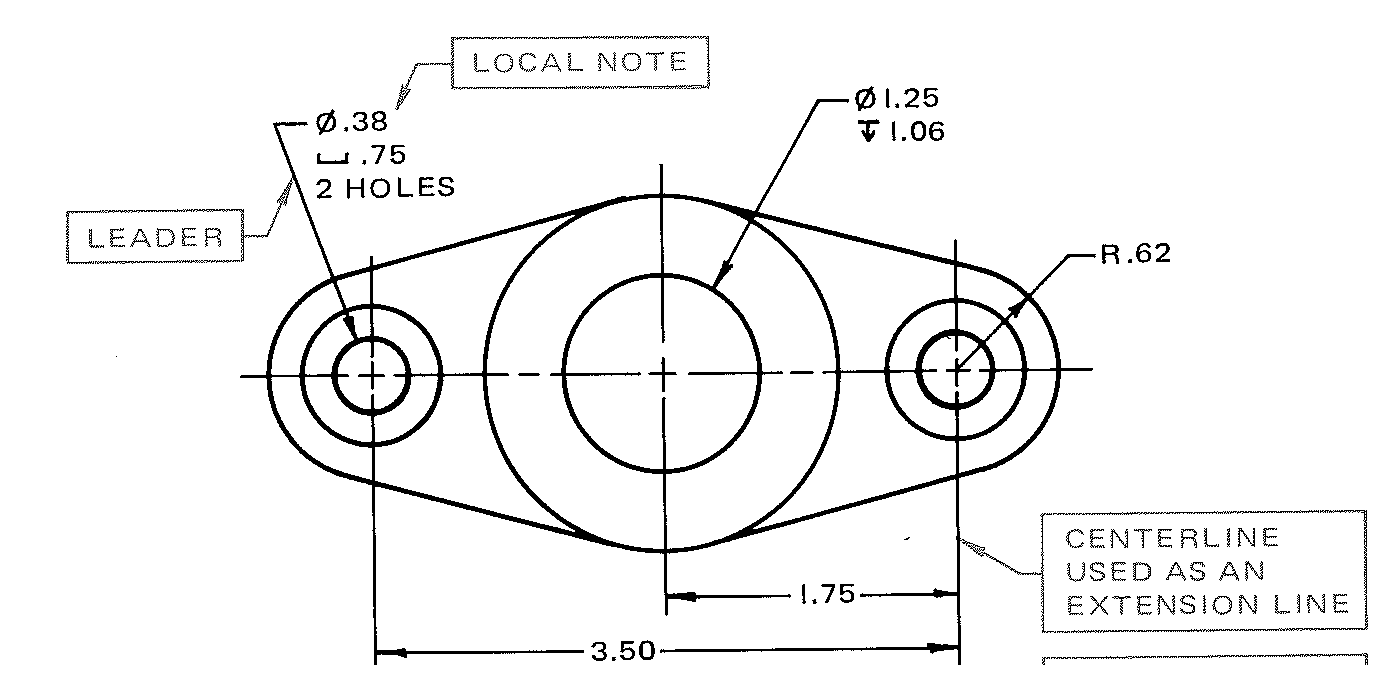}&
\includegraphics[width=0.45\textwidth]{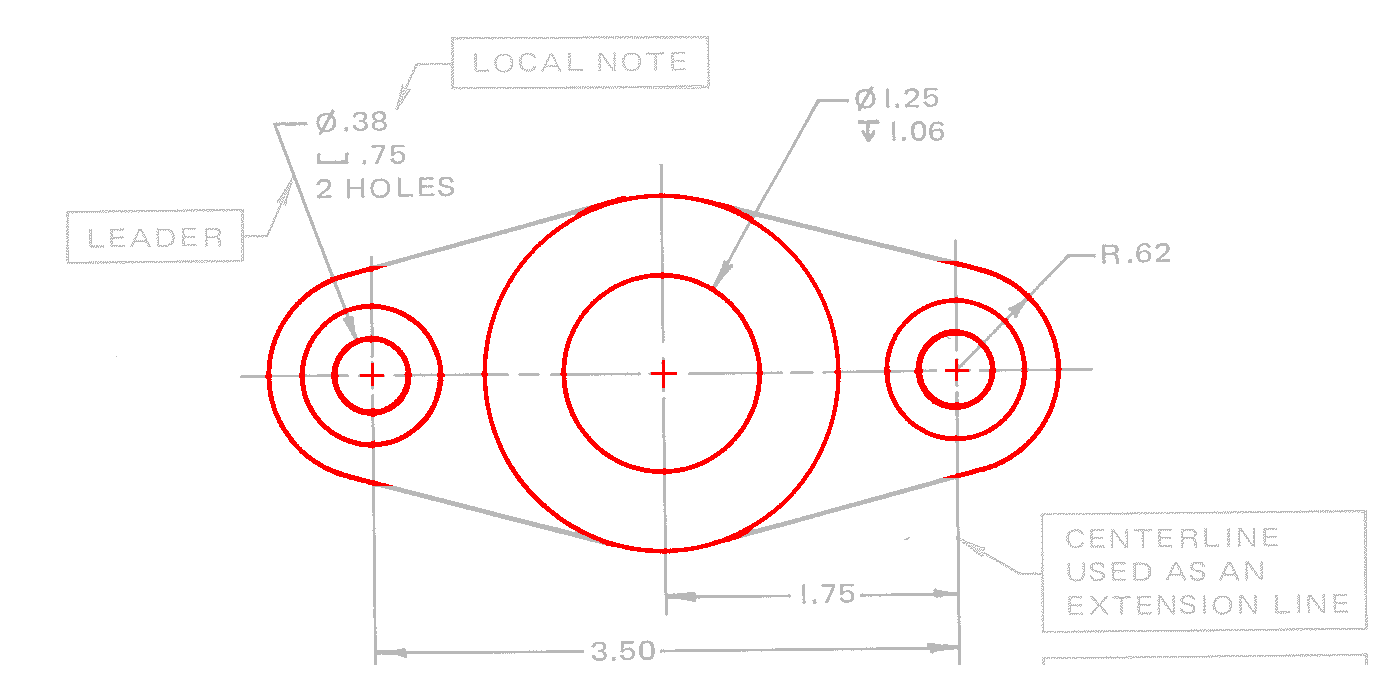}
\end{tabular}
\caption{Results by our algorithm on two images from GREC2007 dataset.
{\bf Top:}\,{\tt g07-tr5.tif}. {\bf Bottom:}\,{\tt g07-tr9.tif}.}
\label{fig:g07-tr59}
\end{figure}}

Figure~\ref{fig:g07-tr1} shows the step-by-step output of our algorithm on the image {\tt g07-tr1.tif} from GREC2007 dataset.
The input image contains closely placed concentric circles and circular arcs, which have been accurately detected by our algorithm.
Another set of results is shown in Figure~\ref{fig:g07-tr4}.
Its input image is {\tt g07-tr4.tif} from GREC2007 dataset, which is a comparatively simpler image containing only full circles, for which our algorithm gives perfect output.
Even for noisy or unclear input images, our algorithm detects the circles and the circular arcs, as evident from the results on two images shown in Figure~\ref{fig:g07-tr23}.
Figure~\ref{fig:g07-tr59} shows the result on another pair of images from GREC2007 dataset.
Although the images contain annotation text, the arcs are correctly detected by our algorithm.
\pb{Some more results on GREC2013 datasets and our own dataset (SMP) are given in Appendix.
We have prepared the dataset SMP by scanning engineering drawing books, e.g.,~\cite{simmons_10}.}

\begin{figure}[!t]
\begin{tabular}{@{\,}c@{\,}c@{\,}c@{\,}}
\fbox{\includegraphics[width=.31\textwidth]{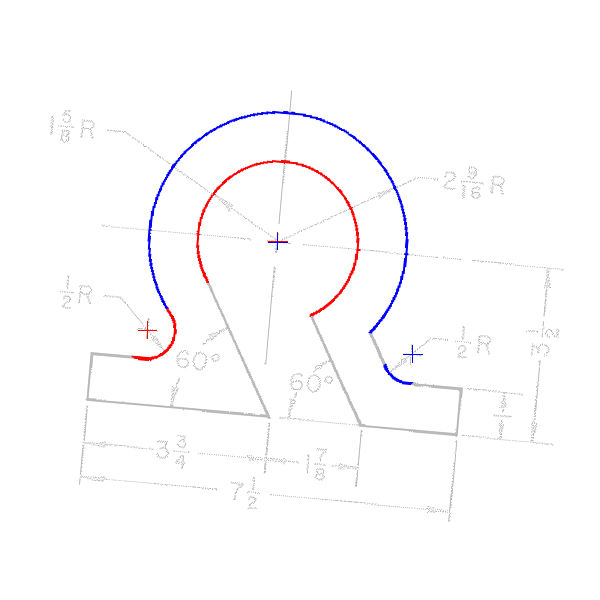}}&
\fbox{\includegraphics[width=.31\textwidth]{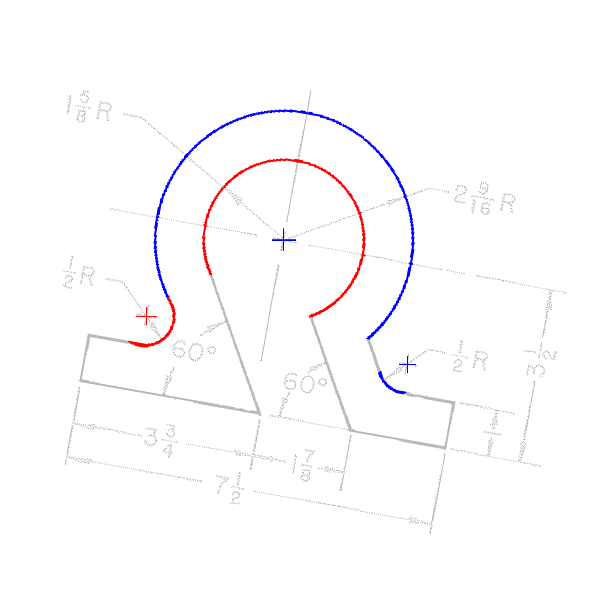}}&
\fbox{\includegraphics[width=.31\textwidth]{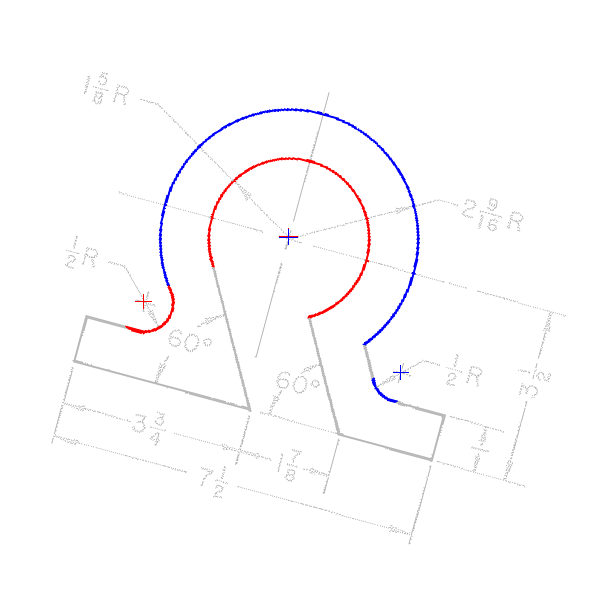}}\\
\fbox{\includegraphics[width=.31\textwidth]{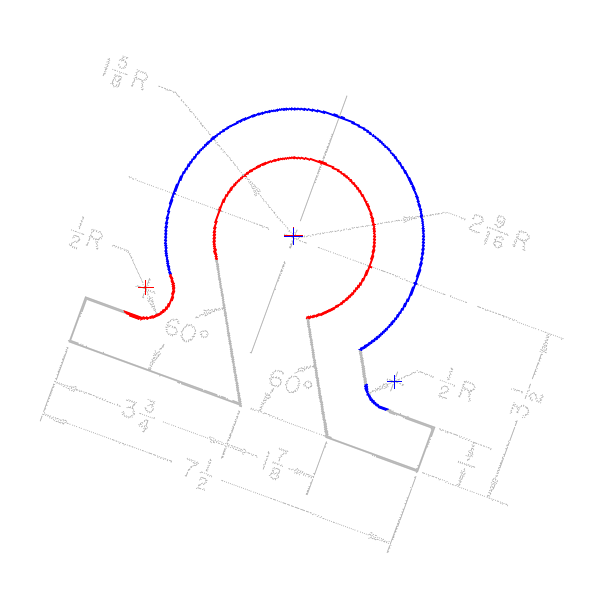}}&
\fbox{\includegraphics[width=.31\textwidth]{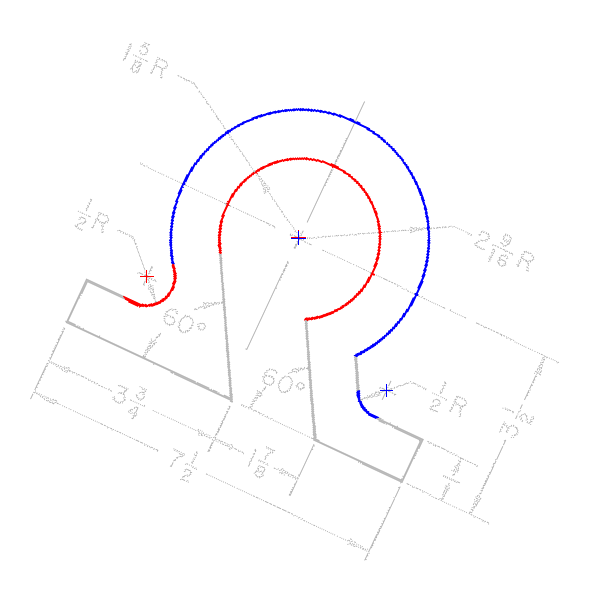}}&
\fbox{\includegraphics[width=.31\textwidth]{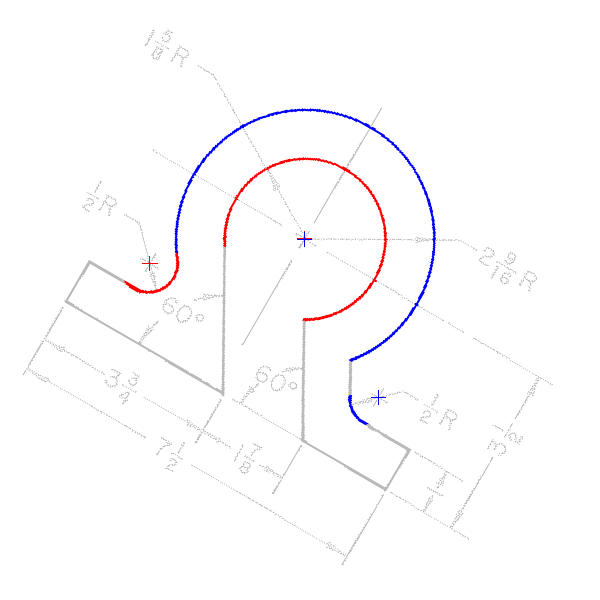}}\\
\fbox{\includegraphics[width=.31\textwidth]{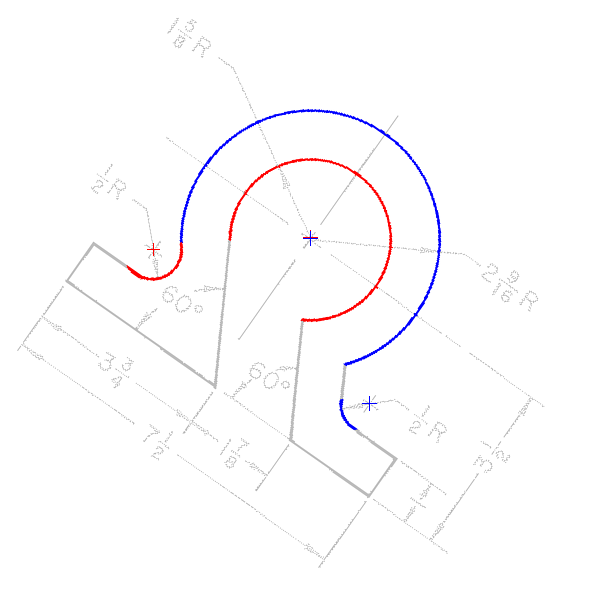}}&
\fbox{\includegraphics[width=.31\textwidth]{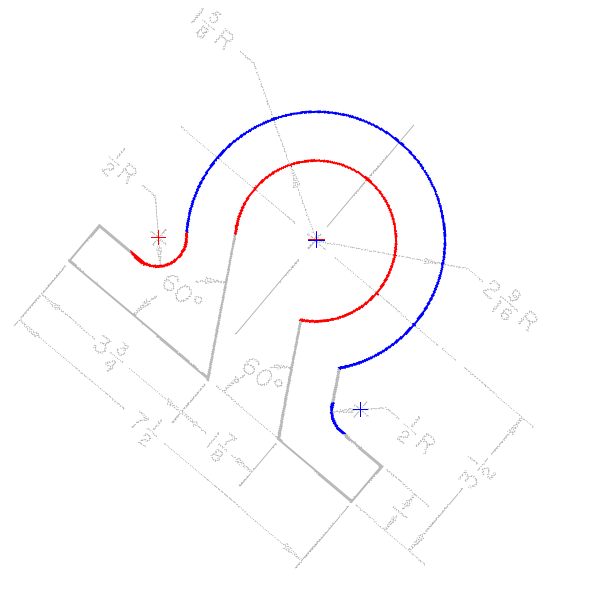}}&
\fbox{\includegraphics[width=.31\textwidth]{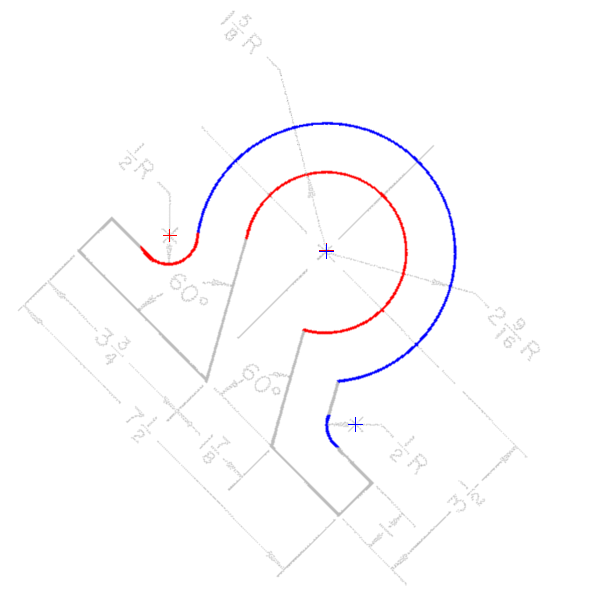}}
\end{tabular}
\caption{Results showing robustness of our algorithm against rotation.
The output images are for {\tt g07-tr6.tif}, rotated clockwise by $5^o$, $10^o$, $15^o,\ldots,45^o$,
shown here in row-major order.}
\label{fig:rotation}
\end{figure}

\begin{figure}[!t]
\begin{tabular}{@{\,}c@{\,}c@{\,}}
\fbox{\includegraphics[width=.46\textwidth,viewport=40 110 590 550,clip]{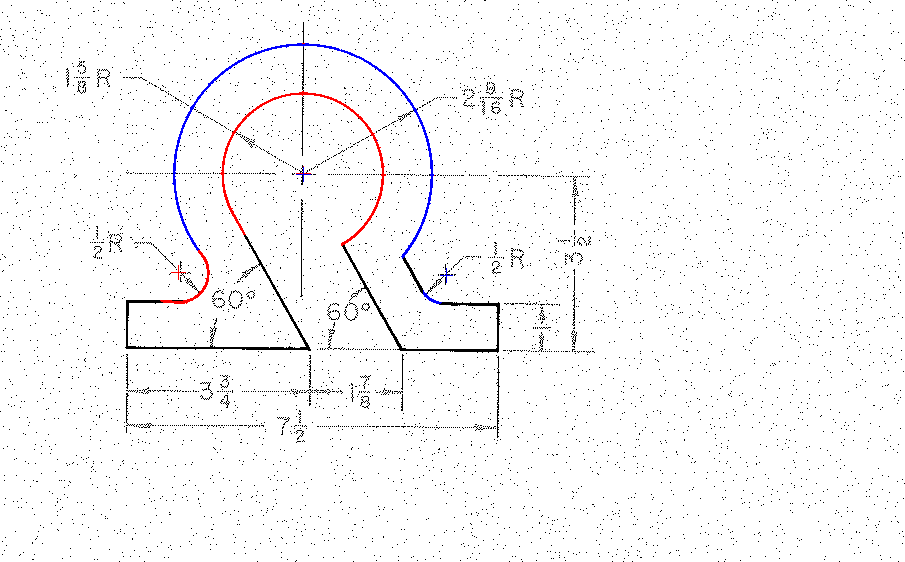}}&
\fbox{\includegraphics[width=.46\textwidth,viewport=40 110 590 550,clip]{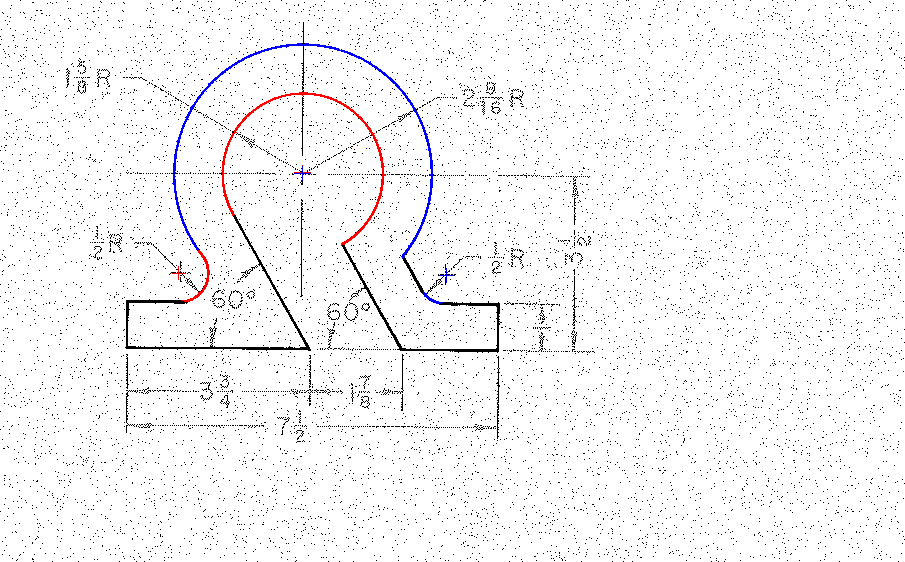}}\\
\fbox{\includegraphics[width=.46\textwidth,viewport=40 110 590 550,clip]{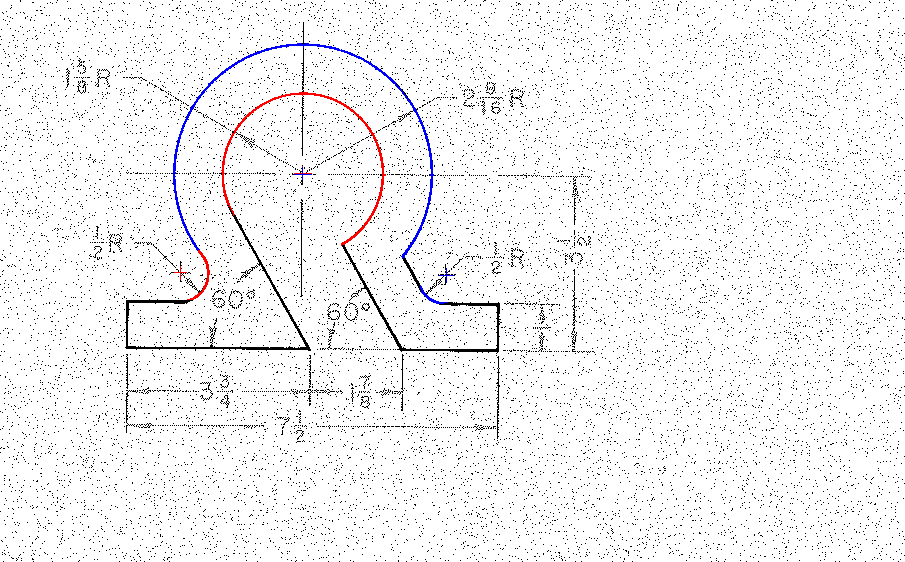}}&
\fbox{\includegraphics[width=.46\textwidth,viewport=40 110 590 550,clip]{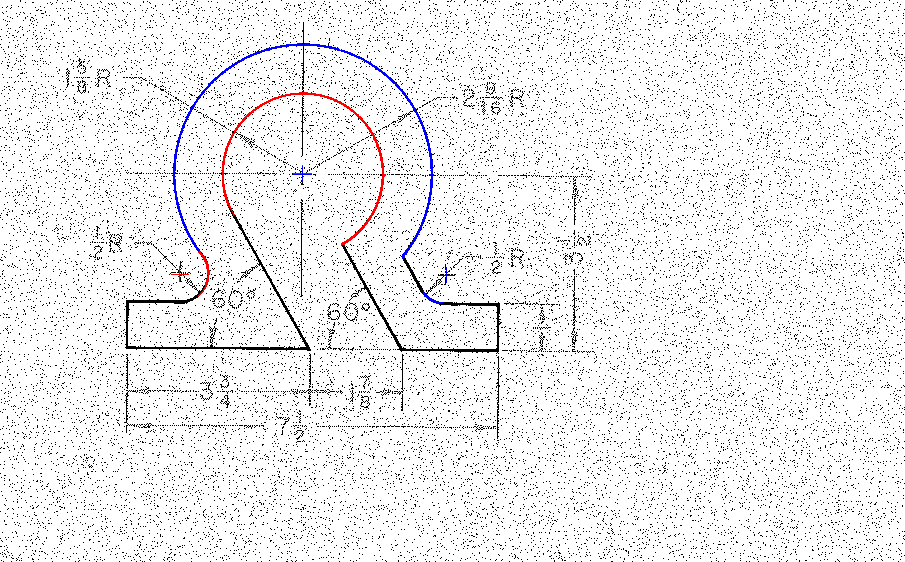}}\\
\fbox{\includegraphics[width=.46\textwidth,viewport=40 110 590 550,clip]{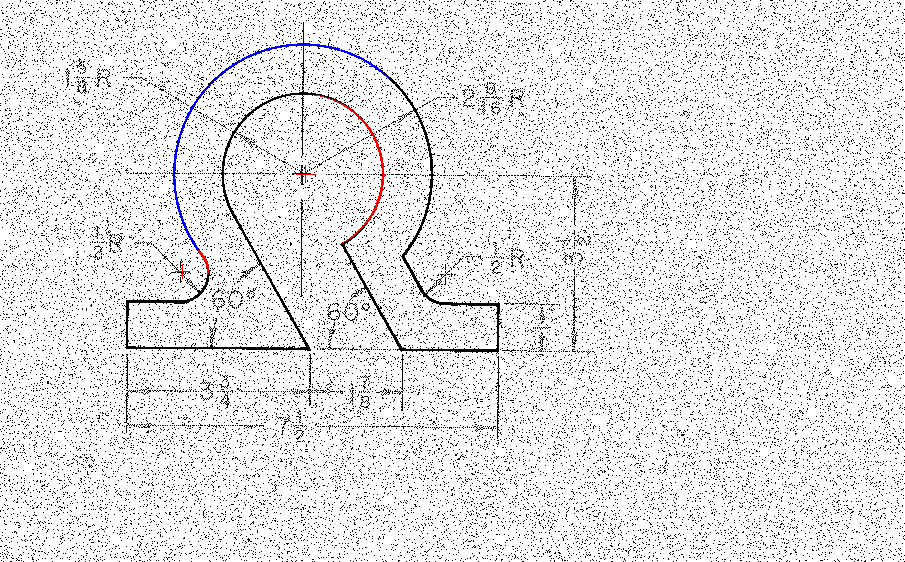}}&
\fbox{\includegraphics[width=.46\textwidth,viewport=40 110 590 550,clip]{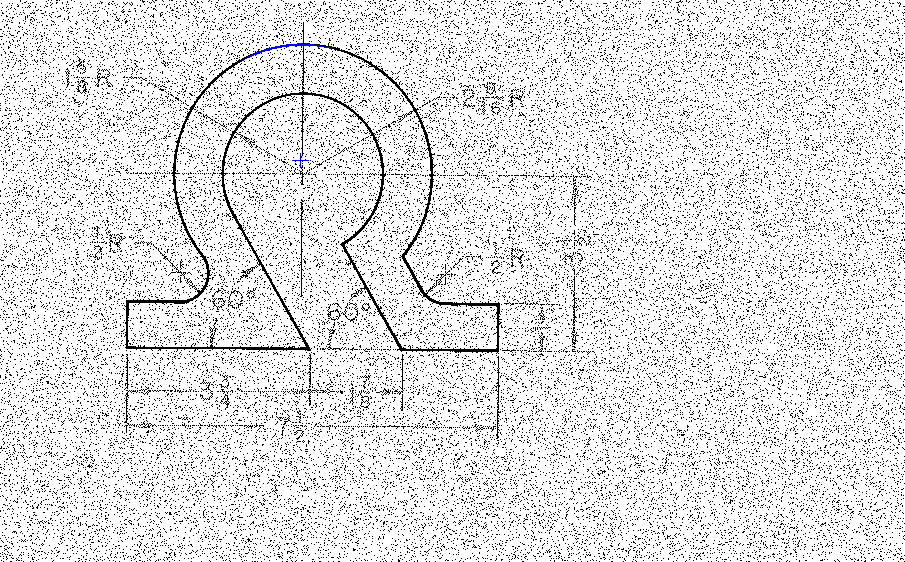}}\\
\end{tabular}
\caption{Results of our algorithm after adding salt-and-paper noise in the image {\tt g07-tr6.tif} 
to different levels (1\%, 2\%, 3\%, 5\%, 8\%, 10\%, shown in row-major order).}
\label{fig:noise}
\end{figure}

We have tested our algorithm to evaluate its robustness against rotation.
Results on the image {\tt g07-tr6.tif} for different angles of rotation are shown in Figure~\ref{fig:rotation}.
These results indicate that the proposed algorithm has a reasonable robustness against rotation.
In order to evaluate its performance on noisy images, we have also run our algorithm on some images after
adding salt-and-pepper noise.
The results on the image {\tt g07-tr6.tif} for different noise levels are shown in Figure~\ref{fig:noise}.
From these results, we can notice that the algorithm gives an acceptable result  up~to 5\% noise.

\begin{table}[!t]\center\footnotesize
\caption{Results for images from GREC2007 dataset for the proposed algorithm.}
\begin{tabular}{@{}@{\ }c@{\ }|@{\ }c@{\ }|@{\ }c@{\ }|@{\ }c|r|r|r|c|c|c|c@{\,}}\hline
Image   &\#rows$\times$ & \multirow{2}{*}{{\centering $N_c$}}&  \multirow{2}{*}{{\centering  $N_g$}}&
\multirow{2}{*}{{\centering $N_p$}} & \multirow{2}{*}{{\centering $N_{fa}$}} &\multirow{2}{*}{{\centering $N_{fn}$}}
& \multirow{2}{*}{{\centering E1}}&\multirow{2}{*}{{\centering E2}} &\multirow{2}{*}{{\centering AD}}&Time\\
{\tt g-07-}&\#columns   &       &  &  &&&&&&(sec.)\\\hline\hline
{\tt 2} & $ 792 \times  662$ & $ 29072$ & $  7820$ & $  7243$ & $  72$ & $ 649$ & $0.921$ & $ 8.299$ & $0.975$ & $0.121$ \\\hline
{\tt 3} & $ 924 \times 1167$ & $ 53899$ & $  6663$ & $  6045$ & $ 282$ & $ 900$ & $4.232$ & $13.507$ & $0.978$ & $0.180$\\\hline
{\tt 4} & $ 638 \times 2046$ & $ 53156$ & $ 16245$ & $ 16398$ & $ 155$ & $   2$ & $0.954$ & $ 0.012$ & $0.997$ & $0.217$\\\hline
{\tt 5} & $ 590 \times  977$ & $  8478$ & $  4326$ & $  4355$ & $  35$ & $   6$ & $0.809$ & $ 0.139$ & $0.995$ & $0.125$\\\hline
{\tt 6} & $ 562 \times  905$ & $  8321$ & $  2639$ & $  2884$ & $ 247$ & $   2$ & $9.360$ & $ 0.076$ & $0.970$ & $0.120$\\\hline
{\tt 7} & $ 779 \times  907$ & $ 14817$ & $ 11435$ & $ 11494$ & $  69$ & $  10$ & $0.603$ & $ 0.088$ & $0.995$ & $0.145$\\\hline
{\tt 8} & $ 982 \times 1064$ & $ 25358$ & $ 12026$ & $ 12845$ & $ 821$ & $   2$ & $6.827$ & $ 0.017$ & $0.968$ & $0.166$\\\hline
{\tt 9} & $ 700 \times 1400$ & $ 44717$ & $ 18650$ & $ 18849$ & $ 347$ & $ 148$ & $1.861$ & $ 0.794$ & $0.989$ & $0.205$\\\hline
{\tt 10}& $ 862 \times  853$ & $ 15189$ & $  9930$ & $  9923$ & $  53$ & $  60$ & $0.534$ & $ 0.604$ & $0.993$ & $0.152$\\\hline
{\tt 11}& $1043 \times  900$ & $ 19182$ & $ 12493$ & $ 12369$ & $  14$ & $ 138$ & $0.112$ & $ 1.104$ & $0.992$ & $0.231$\\\hline
\end{tabular}
\label{tab:compare-sb}
\end{table}

For a quantitative evaluation, we have analyzed its performance using some conventional empirical  measures, namely {\em Type~I error} (E1), {\em Type~II error} (E2), and {\em accuracy of detection} (AD).
These parameters are evaluated in terms of the following variables:
\begin{itemize}
\item $N_c=$ the number of curve pixels in the original image.
\item $N_g=$ the number of pixels on circular arcs in the ground-truth image.
\item $N_p=$ the number of pixels on circular arcs detected by the proposed algorithm.
\item $N_{fa}=$ the number of false-acceptance pixels.
\item $N_{fr}=$ the number of false-rejection pixels.
\end{itemize}

Note that both $N_{fa}$ and $N_{fr}$ represent erroneous output of the proposed algorithm. 
The value of $N_{fa}$ gives the number of pixels detected as part of circular arcs by
the proposed algorithm; these pixels do not correspond to circular arcs in the ground-truth image.
On the other hand, the value of $N_{fr}$ gives the number of pixels that are incorrectly detected as
non-circular arc pixels by the proposed algorithm when they actually lie on circular arcs in the ground-truth
image.
The error estimates are defined using these variables as follows.
\begin{itemize}
\item {\em Type~I error}: $\mbox{E1} = \frac{\mbox{$N_{fa}$}}{\mbox{$N_g$}}\times 100\%$.
\item {\em Type~II error}: $\mbox{E2} = \frac{\mbox{$N_{fr}$}}{\mbox{$N_g$}}\times 100\%$.
\item {\em Accuracy of detection}:\\
\[\begin{array}{ll}
\mbox{AD} &= \frac{\mbox{\# pixels correctly detected on circular and non-circular arcs}} 
{\mbox{\# curve pixels in the original image}}\smallskip\\
&=\frac{\mbox{$N_c$}-(\mbox{$N_{fa}$}+\mbox{$N_{fr}$})}{\mbox{$N_c$}}. \end{array}\]
\end{itemize}
A comparative study of some of the images from the GREC2007 dataset is given in Table~\ref{tab:compare-sb}. 

\subsection{Comparison with Existing Methods}
\label{ss:comparison}

In this section, we compare our work with two existing methods, namely, {\em randomized Hough transform} (RHT)
\cite{xu_93} and the HT-based circular arc detection using {\em effective voting method} (EVM) \cite{chiu_05}.
Both these methods are basically improvisations on HT, which is by far the most popular arc detection
technique.
We first provide short descriptions of these two methods, and then compare them with the proposed method (CSA).

\paragraph{Comparison with RHT \cite{xu_93}}
\label{para:rht}

The improvement of HT in terms of computational time is effected in this method by randomly selecting three
points from the image $I$ at each step and mapping them into a point in the parameter space.
The parameter space is implicitly represented by a set $P$, each of whose elements contains a real-valued vector (circle parameter values) and an integer score.
At each step, the parameter space is updated  by checking whether the mapped point already exists in $P$; if so, then its score is incremented by one;
and if not, then a new element is added to $P$, with its score set to one.
In the process of updating, if an element obtains a score equal to the threshold $n_t$ ($=2$ or $3$), then it is considered as a candidate for a true circle.
If the number of points of $I$ lying on this circle is greater than a predefined threshold $T_r$, then it is reported and all the concerned points are removed from $I$.
The above process is repeated until the stopping criterion is satisfied.

\paragraph{Comparison with EVM \cite{chiu_05}}
\label{para:chiu-liaw}

In this method, point triplets are selectively chosen by considering the point pairs from the sampled object points, $M$.
While obtaining a triplet, all point pairs in $M$ are considered; and for each pair $(p, q) \in M^2$, the
triplet $(p, q, r)$ is chosen such that $r$ is an object point in the image $I$ and ${pq} =
{qr}$.
For each selected triplet, the entire image $I$ is searched for all the object points lying on the circle $C$ represented by it, and the {\em existing rate} of $C$ is determined.
The existing rate of $C$ is the ratio of the number of object points lying on $C$ to the circumference of $C$.
Based on the existing rate values, each object point in $I$ votes for only that circle which has the highest existing rate among all the circles passing through it.
Finally, all the circles that have existing rate values higher than the given threshold value $T_e$ are reported.

{\def\baselinestretch{1.1}
\begin{figure}[!t]\center
\begin{tabular}{@{}c@{\,}|@{\,}c@{}|@{\,}c@{}}\hline
\includegraphics[width=0.3\textwidth]{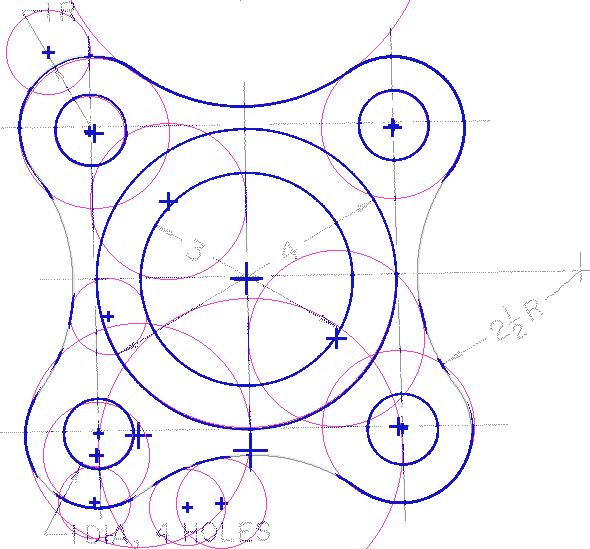}&
\includegraphics[width=0.3\textwidth]{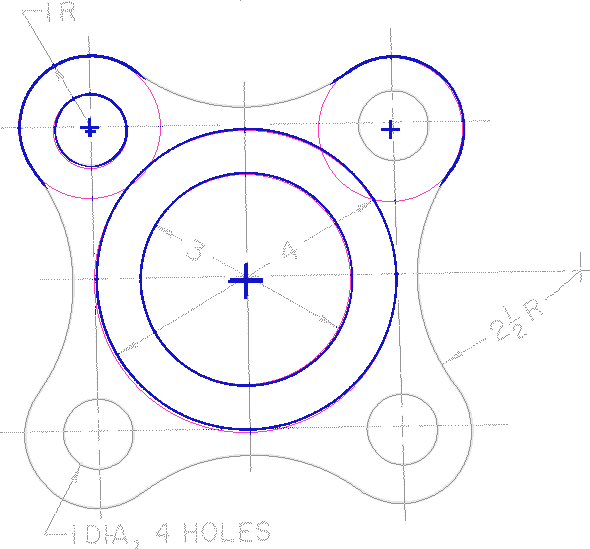}&
\includegraphics[width=0.3\textwidth]{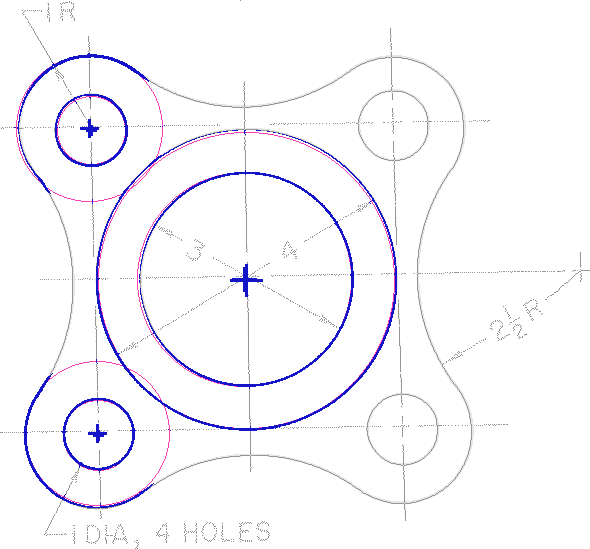}\\\hline
\,(a)&\,(b)&\,(c)\\\hline
\includegraphics[width=0.3\textwidth]{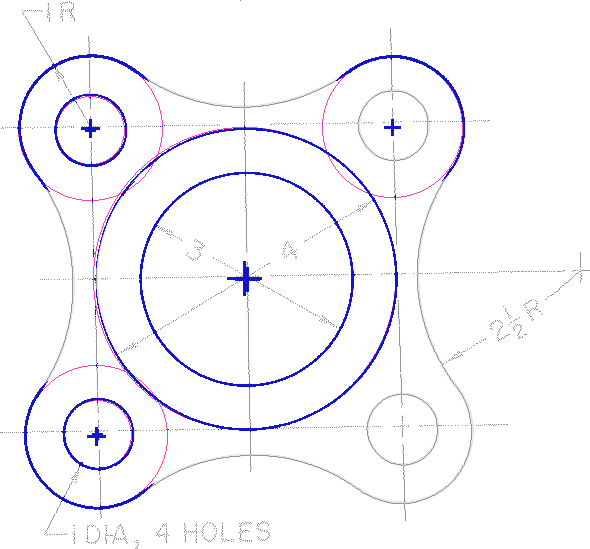}&
\includegraphics[width=0.3\textwidth]{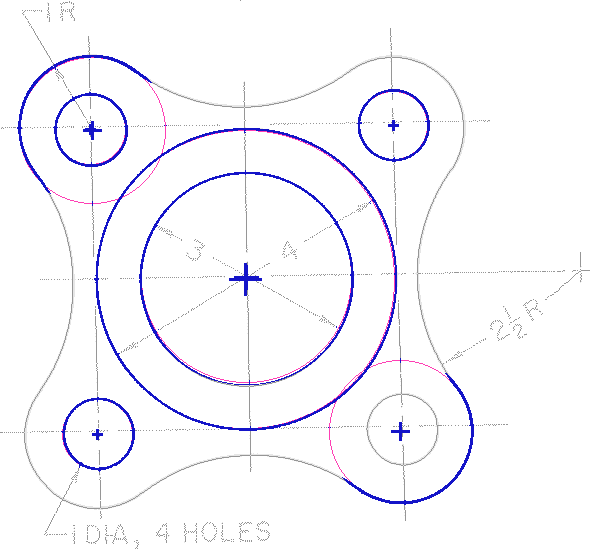}&
\includegraphics[width=0.3\textwidth]{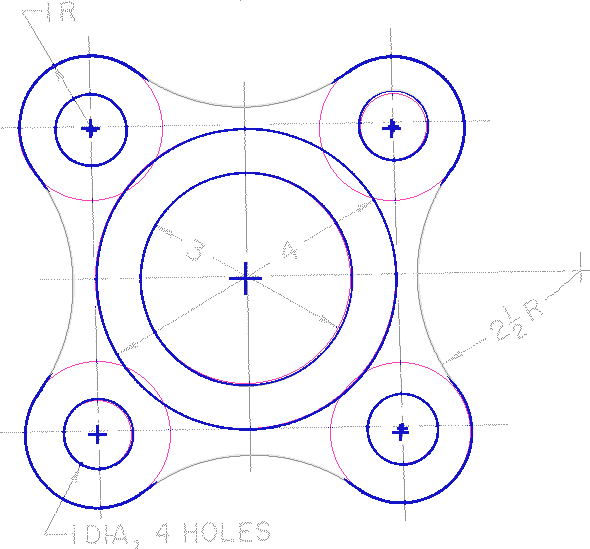}\\\hline
\,(d)&\,(e)&\,(f)\\\hline
\end{tabular}
\caption{Results of different (randomized) runs of RHT for (a)~$T_r=0.23$ and (b--f)~$T_r=0.46$ on the image {\tt g07-tr7.tif}. 
All the different runs for $T_r=0.46$ produce different outputs due to randomization. 
See Table~\ref{tab:g07-tr7-rht} for statistical details.}
\label{fig:g07-tr7-rht}
\end{figure}

{\def\baselinestretch{1.1}
\begin{figure}[!t]\center\footnotesize
\begin{tabular}{@{}c@{\,}|@{\,}c@{}|@{\,}c@{}}\hline
\includegraphics[width=0.3\textwidth]{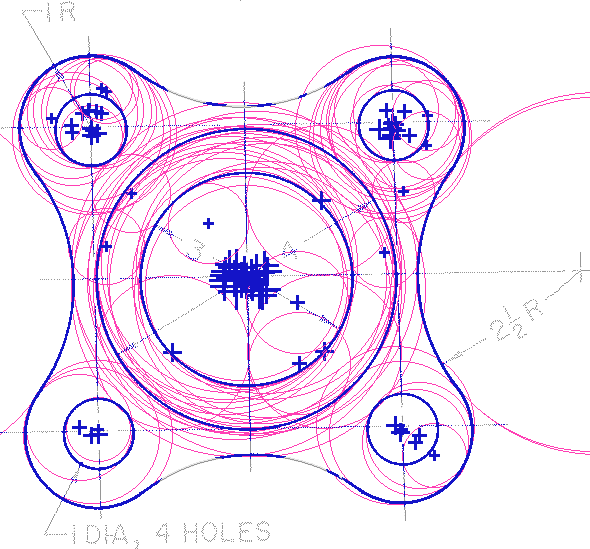}&
\includegraphics[width=0.3\textwidth]{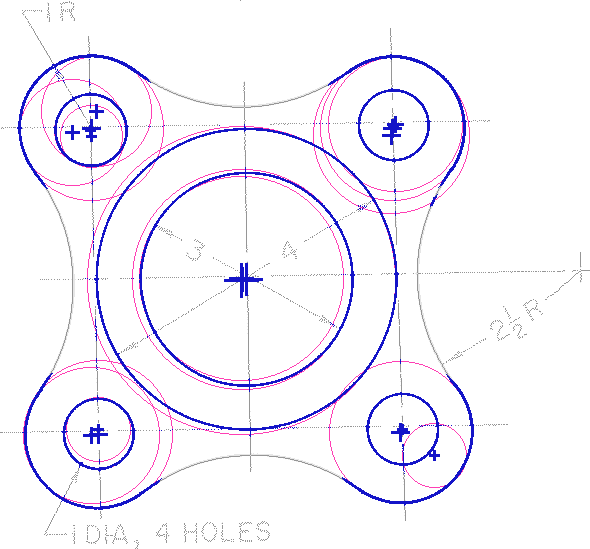}&
\includegraphics[width=0.3\textwidth]{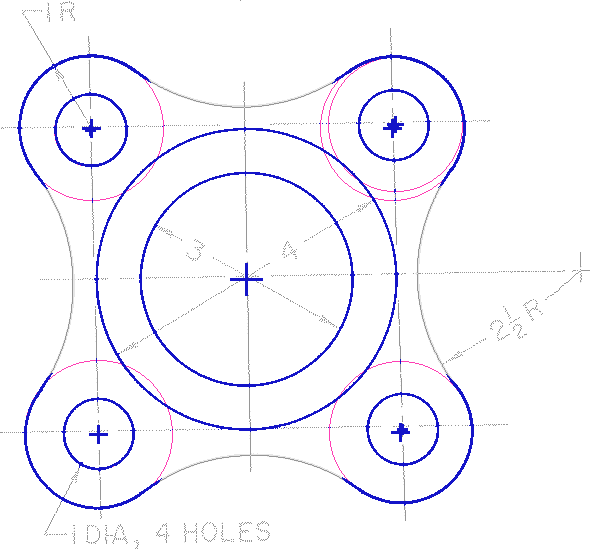}\\\hline
\,(a)\,$T_e = 0.20$&\,(b)\,$T_e = 0.25$&\,(c)\,$T_e = 0.30$\\\hline
\includegraphics[width=0.3\textwidth]{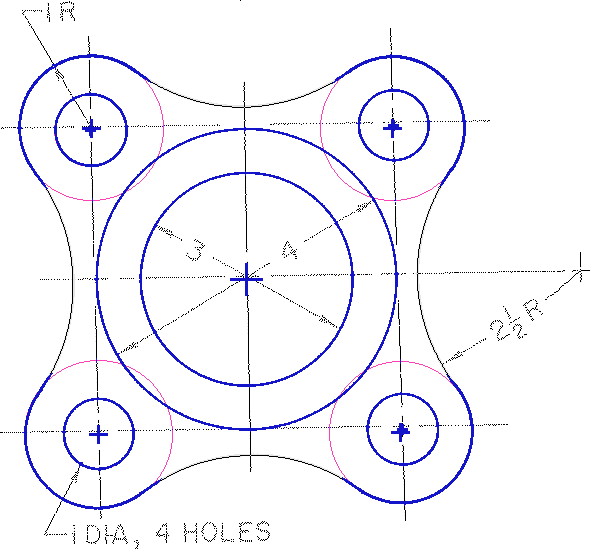}&
\includegraphics[width=0.3\textwidth]{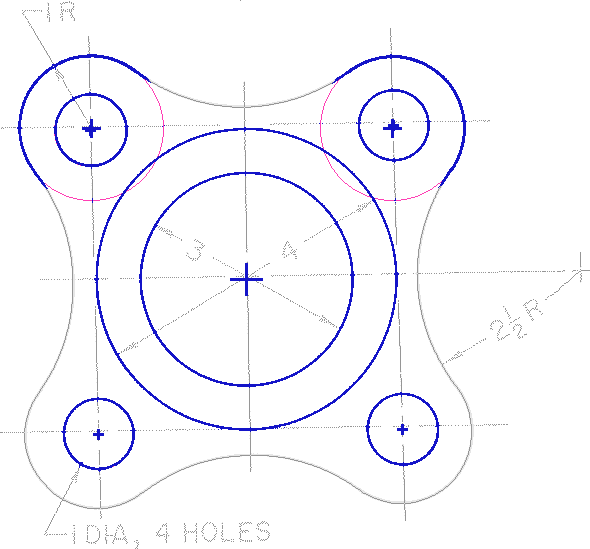}&
\includegraphics[width=0.3\textwidth]{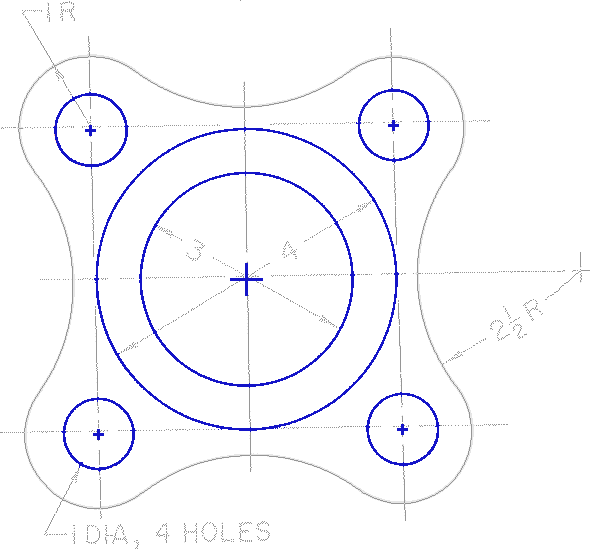}\\\hline
\,(d)\,$T_e = 0.35$&\,(e)\,$T_e = 0.45$&\,(f)\,$T_e = 0.50$\\\hline
\end{tabular}
\caption{Results of running the algorithm EVM for different values of $T_e$ on the image {\tt g07-tr7.tif}.
Notice that for low values of $T_e$, many extraneous arcs are reported along with the correct arcs; and for higher values of $T_e$, the output image cannot detect all the circular arcs accurately.
See Table~\ref{tab:g07-tr7-evm} for statistical details.}
\label{fig:g07-tr7-evm}
\end{figure}}

{\def\baselinestretch{1.1}
\begin{figure}[!t]\center\footnotesize
\begin{tabular}{@{}c@{\,}|@{\,}c@{}|@{\,}c@{}}\hline
\includegraphics[width=0.3\textwidth]{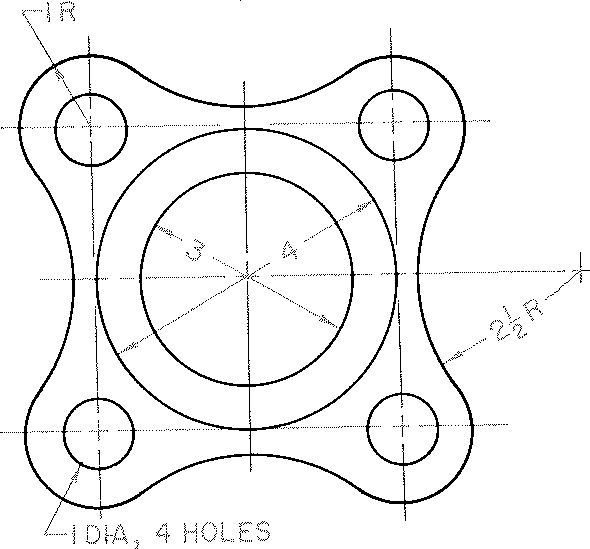}&
\includegraphics[width=0.3\textwidth]{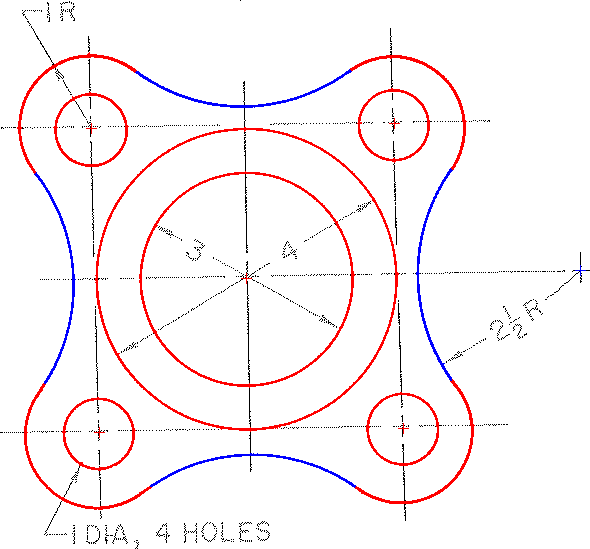}&
\includegraphics[width=0.3\textwidth]{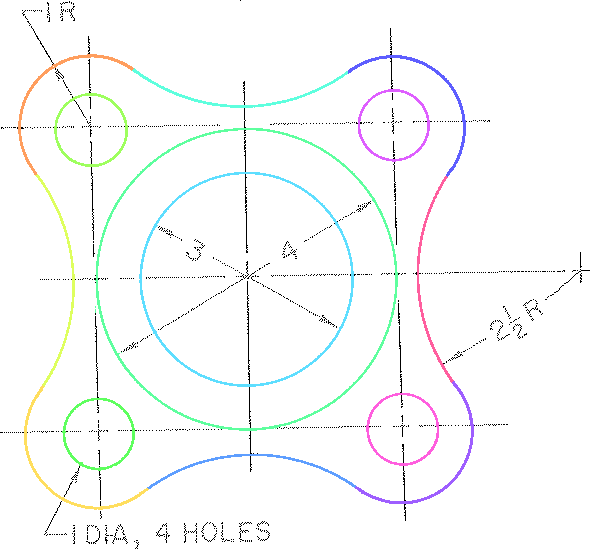}\\\hline
\,(a)&\,(b)&\,(c)\\\hline
\end{tabular}
\caption{Result by our algorithm on (a)\,{\tt g07-tr7.tif} from GREC2007 dataset.
Unlike EVM and RHT, our algorithm does not require any parameter tuning, and (b)\,its result is close to
(c)\,the bench-mark result.}
\label{fig:g07-tr7-sb}
\end{figure}}

{\def\baselinestretch{1.1}
\begin{table}[!t]\scriptsize
\caption{Results of running EVM for different values of $T_e$ on the image {\tt g07-tr7.tif}. 
The best results are obtained for $T_e = 0.32$--$0.43$ when ten circles (circular arcs) are detected
out of the fourteen circles/circular arcs present in the input image (Figure~\ref{fig:g07-tr7-sb}(b)).}
\centering
\begin{tabular}{c|c|c|r|r|r|r|r|r|r}\hline
\multirow{3}{*}{$T_e$}&\#circles&\#correct circles&& \multicolumn{4} {c|}{ Error} &&Time\\\cline{5-8}
&\& arcs&\& arcs&$N_{p}$&\multirow{2}{*}{$N_{fa}$ }&\multirow{2}{*}{$N_{fr}$ }&
\multirow{2}{*}{E1 (\%)}&\multirow{2}{*}{E2 (\%)}&AD&(sec.)\\
&detected& detected $/$ $14$ & &&&&&&  \\
\hline
 $ 0.20$ & $77$ &$12$& $ 12368$ & $1541$ & $ 608$ & $13.476$ & $5.317$ & $0.855$ & $54.272$\\\hline
$0.25$ & $21$ &$10$& $  9787$ & $ 399$ & $2047$ & $3.489$ & $17.901$ & $0.835$ &$54.272$\\\hline
$0.30$ &$11$ &$10$ & $  9478$ & $ 138$ & $2095$ & $1.207$ & $18.321$ & $0.849$ & $54.272$\\\hline
 $ 0.32$--$0.43$ & $10$ &$10$& $  9467$ & $ 127$ & $2095$ & $1.111$ & $18.321$ & $0.850$ & $54.272$\\\hline
$ 0.45$ & $8$ &$8$& $  7981$ & $ 106$ & $3560$ & $0.927$ & $31.132$ & $0.753$ & $54.272$\\\hline
$ 0.50$--$0.85$ & $6$ &$6$& $  6457$ & $  79$ & $5057$ & $0.691$ & $44.224$ & $0.653$ & $54.272$\\\hline
 $ 0.90$ & $1$ &$1$& $   588$ & $  10$ & $10857$ & $0.088$ & $94.945$ & $0.267$ & $54.272$\\\hline
 $ 0.95$ & $0$ &$0$& $     0$ & $   0$ & $11435$ & $0.000$ & $100.000$ & $0.228$ & $54.272$\\\hline
\end{tabular}
\label{tab:g07-tr7-evm}
\end{table}}

{\def\baselinestretch{1.1}
\begin{table}[!t]\scriptsize
\caption{Results of different (randomized) runs of RHT for $T_r=0.46$ on the image {\tt g07-tr7.tif}.}
\centering
\begin{tabular}{c|c|c|c|r|r|r|r|r|r|r}\hline
\multirow{3}{*}{$T_r$}&&\#circles&\#correct circles&& \multicolumn{4} {c|}{ Error} &&Time\\\cline{6-9}
&Run \#&\& arcs&\& arcs&$N_{p}$&\multirow{2}{*}{$N_{fa}$ }&\multirow{2}{*}{$N_{fr}$ }&
\multirow{2}{*}{E1 (\%)}&\multirow{2}{*}{E2 (\%)}&AD&(sec.)\\
&&detected& detected $/$ $14$ & &&&&&&  \\
\hline
 \multirow{10}{*}{$0.46$}&$1$ & $7$ &$7$ & $  7021$ & $  97$ & $4511$ & $0.848$ & $39.449$ & $0.689$ & $ 7.749$\\\cline{2-11}
& $2$ & $7$ &$7$ & $  7089$ & $  93$ & $4439$ & $0.813$ & $38.819$ & $0.694$ & $9.709$\\\cline{2-11}
& $3$ & $8$ &$8$ & $  7875$ & $  98$ & $3658$ & $0.857$ & $31.989$ & $0.746$ & $6.471$\\\cline{2-11}
& $4$ & $7$ &$7$ & $  7251$ & $  95$ & $4279$ & $0.831$ & $37.420$ & $0.705$ & $1.273$\\\cline{2-11}
& $5$ & $5$ &$5$ & $  5995$ & $  91$ & $5531$ & $0.796$ & $48.369$ & $0.621$ & $1.637$\\\cline{2-11}
& $6$ & $10$ &$10$ & $  9084$ & $ 121$ & $2472$ & $1.058$ & $21.618$ & $0.825$ & $1.638$\\\cline{2-11}
& $7$ & $7$ &$7$ & $  7267$ & $ 100$ & $4268$ & $0.874$ & $37.324$ & $0.705$ & $6.223$\\\cline{2-11}
& $8$ & $6$ &$6$ & $  5763$ & $  84$ & $5756$ & $0.735$ & $50.337$ & $0.606$ & $6.992$\\\cline{2-11}
& $9$ & $6$ &$6$ & $  5876$ & $  85$ & $5644$ & $0.743$ & $49.357$ & $0.613$ & $7.463$\\\cline{2-11}
& $10$ & $10$ &$10$ & $  9072$ & $ 120$ & $2483$ & $1.049$ & $21.714$ & $0.824$ & $11.123$\\\hline
\end{tabular}
\label{tab:g07-tr7-rht}
\end{table}}

Thus, we observe that the above methods require a predefined threshold value ($T_r$ for RHT and $T_e$ for EVM)
for detecting the true circles in the image.
If only complete circles are present in the image, then high threshold values suffice in both the cases. 
However, for the images, which also contain (partial) circular arcs, the thresholds have to be lowered
sufficiently to get proper result.
So, for best results, the optimum threshold value for each image should be set individually depending on the
nature of the circular arcs present in the image.  
But obtaining an optimal threshold value is quite difficult, and hence this issue becomes a major drawback of both EVM and RHT methods.

Our method, on the contrary, suffers from no threshold-related weakness, as it requires no threshold value and outputs an optimal or near-optimal result.
This is evident from Figs.~\ref{fig:g07-tr7-rht}, \ref{fig:g07-tr7-evm}, and \ref{fig:g07-tr7-sb}, showing the output results for RHT, EVM, and our algorithm (CSA) respectively, for the image {\tt g07-tr7.tif} that actually contains $6$ circles and $8$ semicircular arcs, as shown in its ground-truth image (Figure~\ref{fig:g07-tr7-sb}(b)).
Hence, for each individual image, depending on the extent of noise and the distortion of circular arcs and circles, the threshold value has to be properly set for RHT and EVM.
From Figure~\ref{fig:g07-tr7-evm} and Table~\ref{tab:g07-tr7-evm}, we can see that the best possible result for the image {\tt g07-tr7.tif} is obtained when $T_e$ lies in the range $[0.32, 0.43]$; but even in this range, the result is
not perfect, as four circular arcs still remain undetected.
In case of RHT, the situation is even worse, as each run of the algorithm under the same condition does not always produce the same result (due to randomization) when the value of $T_r$ is low.
A summary of the statistical details for ten different (randomized) runs of the algorithm with $T_r = 0.46$ is given in Table~\ref{tab:g07-tr7-rht},
and a sample set of output images for five of these runs is given in Figure~\ref{fig:g07-tr7-rht}.

\begin{figure}[!t]\center\scriptsize
\begin{minipage}{.5\textwidth}\hspace*{-0mm}
        \includegraphics[width=\textwidth]{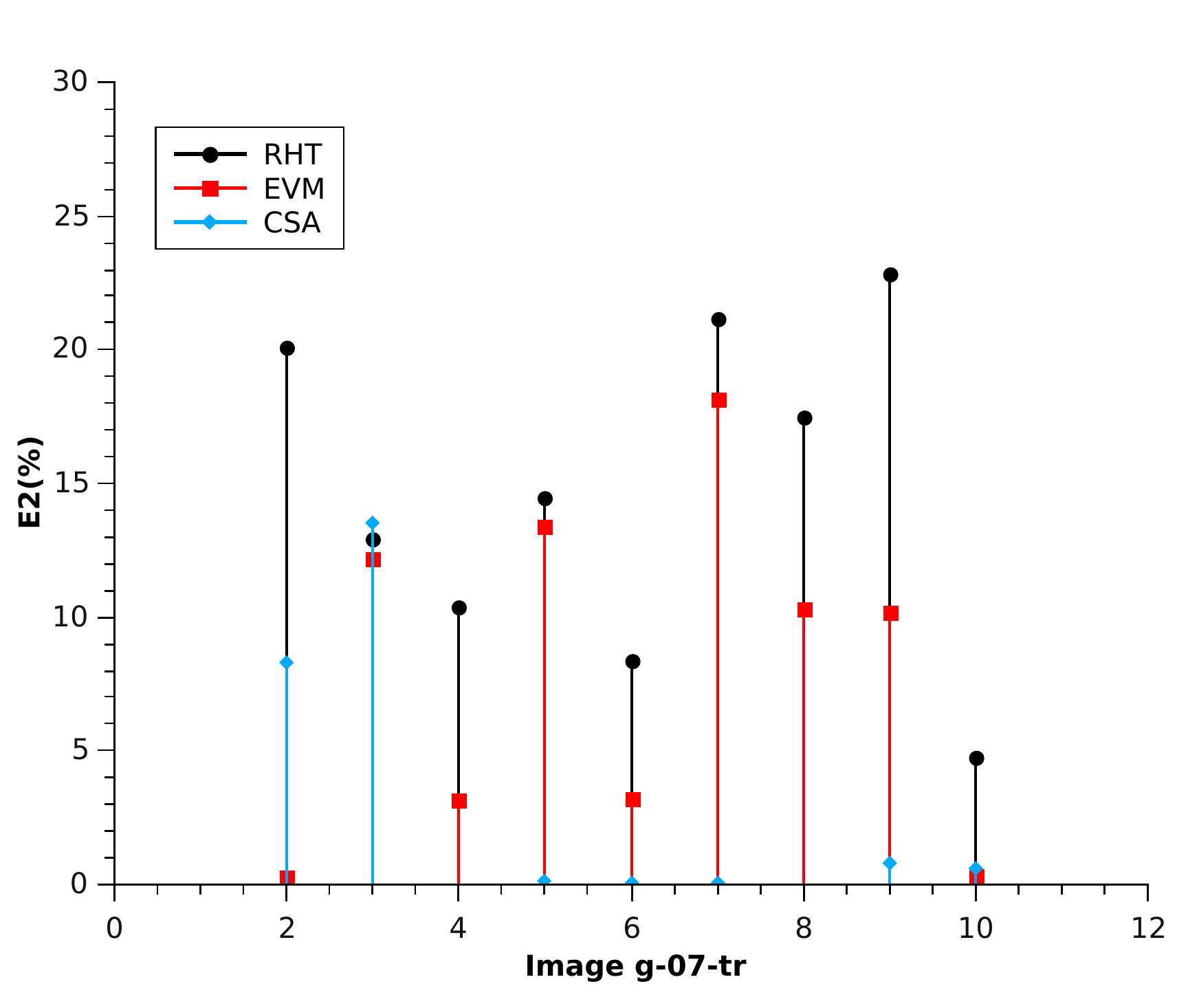}
\end{minipage}
\begin{minipage}{.47\textwidth}\scriptsize
    \begin{tabular}{c|rrr|rrr}\toprule
        Image    & \multicolumn{6}{c}{Error}\\\cline{2-7}
        {\tt g-07}&\multicolumn{3}{c|}{E1\,(\%)}&\multicolumn{3}{c}{E2\,(\%)}\\\cline{2-7}
        {\tt -tr}& EVM     & RHT     & CSA       & EVM      & RHT    & CSA   \\\toprule
        {\tt 2}  & $3.171$ & $ 3.453$& $0.921$  &  $0.294$ & $20.115$&  $8.299$\\\hline
        {\tt 3}  & $3.542$ & $ 7.969$& $4.232$  & $12.202$ & $12.937$& $13.507$\\\hline
        {\tt 4}  & $0.721$ & $ 0.936$& $0.954$  &  $3.158$ & $10.372$&  $0.012$\\\hline
        {\tt 5}  & $1.179$ & $ 1.271$& $0.809$  & $13.385$ & $14.494$&  $0.139$\\\hline
        {\tt 6}  & $7.579$ & $ 4.775$& $9.360$  &  $3.221$ &  $8.375$&  $0.076$\\\hline
        {\tt 7}  & $0.848$ & $ 0.892$& $0.603$  & $18.120$ & $21.189$&  $0.874$\\\hline
        {\tt 8}  & $3.318$ & $ 4.798$& $6.827$  & $10.294$ & $17.446$&  $0.017$\\\hline
        {\tt 9}  & $2.461$ & $12.000$& $1.861$  & $10.182$ & $22.804$&  $0.794$\\\hline
        {\tt 10} & $2.236$ & $ 2.064$& $0.534$  &  $0.303$ &  $4.743$&  $0.604$\\\hline
        {\tt 11} & $0.240$ & $ 0.472$& $0.112$  & $21.556$ & $24.918$&  $1.105$\\\toprule
    \end{tabular}
\end{minipage}
\caption{{Plots of E2 for RHT, EVM, and CSA (proposed algorithm) for images from GREC2007 dataset.}}
\label{fig:comp-E2}
\end{figure}

\begin{figure}[!t]\center\scriptsize
\begin{minipage}{.6\textwidth}\hspace*{-5mm}
        \includegraphics[width=\textwidth]{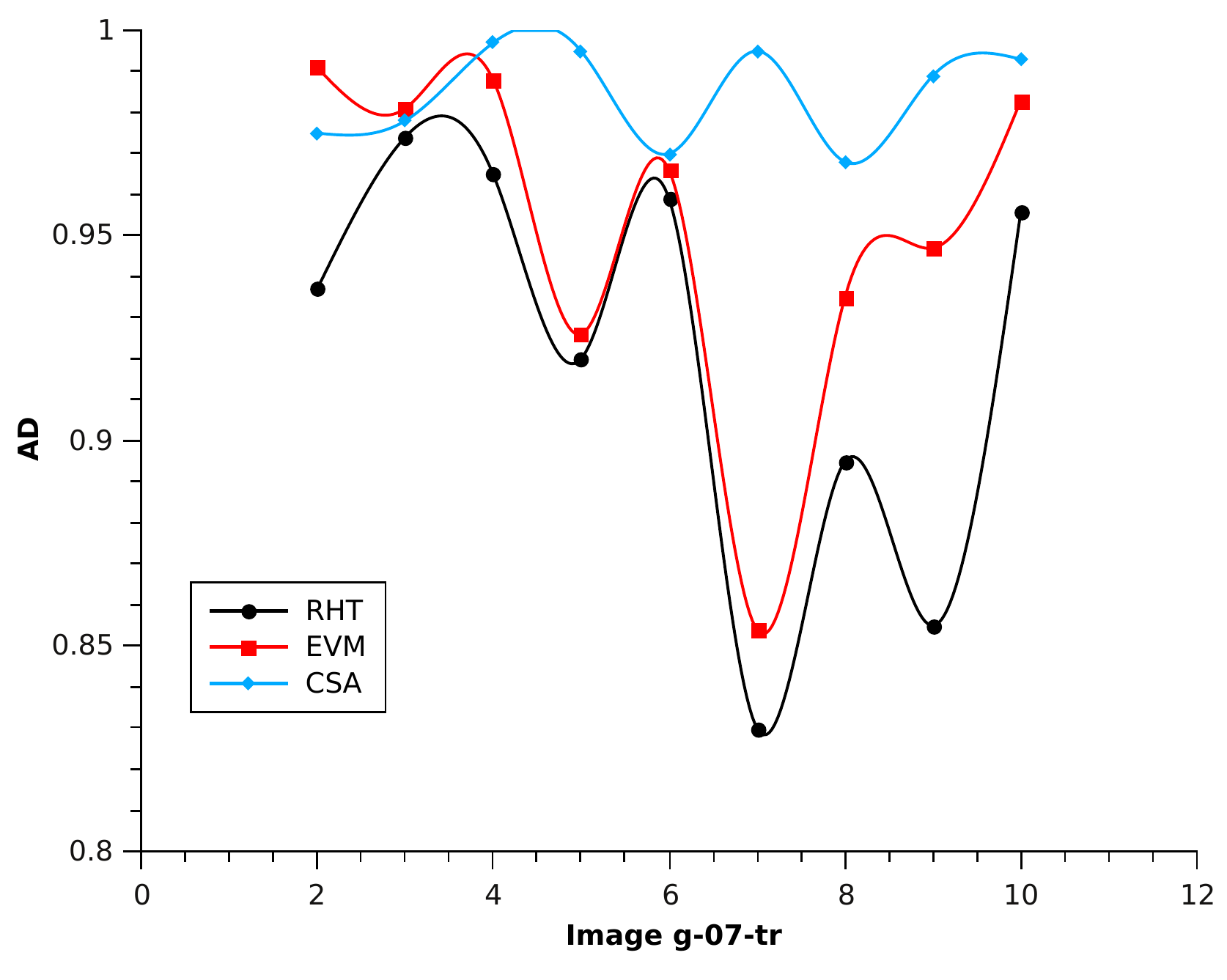}
\end{minipage}
\begin{minipage}{.35\textwidth}\scriptsize
    \begin{tabular}{crrr}\toprule
        Image   &\multicolumn{3}{c}{AD}\\\cline{2-4}
        {\tt go7-tr}&EVM   &   RHT  & CSA \\\toprule
        {\tt 2}  & $0.991$ & $0.937$& $0.975$\\\hline
        {\tt 3}  & $0.981$ & $0.974$& $0.978$\\\hline
        {\tt 4}  & $0.988$ & $0.965$& $0.997$\\\hline
        {\tt 5}  & $0.926$ & $0.920$& $0.995$\\\hline
        {\tt 6}  & $0.966$ & $0.959$& $0.970$\\\hline
        {\tt 7}  & $0.854$ & $0.830$& $0.995$\\\hline
        {\tt 8}  & $0.935$ & $0.895$& $0.968$\\\hline
        {\tt 9}  & $0.947$ & $0.855$& $0.989$\\\hline
        {\tt 10} & $0.983$ & $0.956$& $0.993$\\\hline
        {\tt 11} & $0.858$ & $0.835$& $0.992$\\\toprule
    \end{tabular}
\end{minipage}
\caption{{Plots of AD for RHT, EVM, and CSA (proposed algorithm) for images from GREC2007 dataset.}}
\label{fig:comp-AD}
\end{figure}

\begin{figure}[!t]\center\scriptsize
\begin{minipage}{.6\textwidth}\hspace*{-5mm}
        \includegraphics[width=\textwidth]{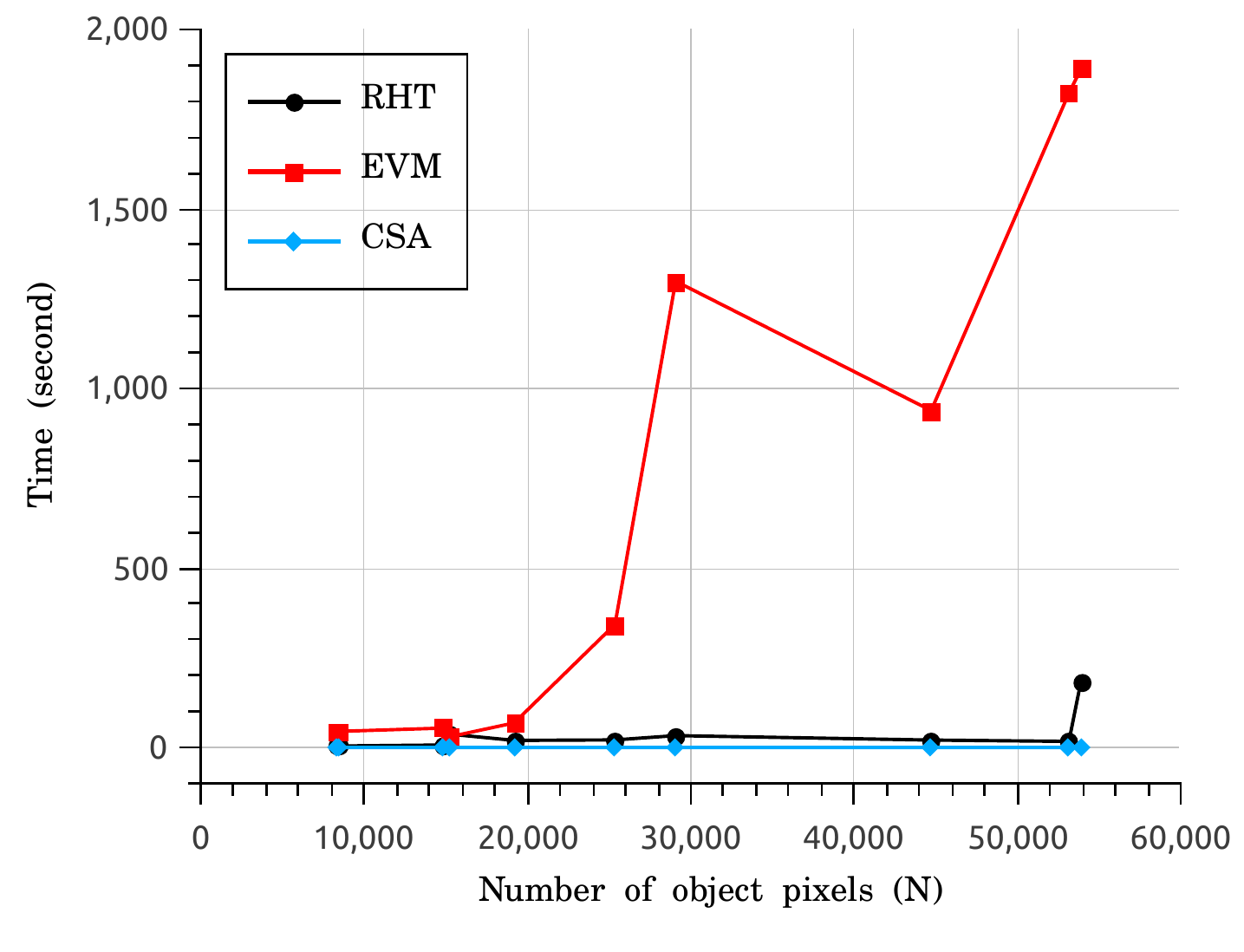}
\end{minipage}
\begin{minipage}{.35\textwidth}\scriptsize
    \begin{tabular}{rcrrr}\toprule
        $N$\ \ &Image   &\multicolumn{3}{c}{CPU time (sec.)}\\\cline{3-5}
           &    {\tt go7-tr}&     EVM   &   RHT   & CSA \\\toprule
         8,321 &     6      &    46.892 &   3.764 & 0.120\\\hline
         8,478 &     5      &    45.104 &   4.389 & 0.125\\\hline
        14,817 &     7      &    54.272 &   7.228 & 0.145\\\hline
        15,189 &    10      &    29.730 &  37.527 & 0.152\\\hline
        19,182 &    11      &    68.625 &  20.021 & 0.231\\\hline
        25,358 &     8      &   343.077 &  21.350 & 0.166\\\hline
        29,072 &     2      & 1,298.658 &  33.309 & 0.121\\\hline
        44,717 &     9      &   938.428 &  20.903 & 0.205\\\hline
        53,156 &     4      & 1,822.414 &  17.229 & 0.217\\\hline
        53,899 &     3      & 1,890.005 & 185.968 & 0.180\\\toprule
    \end{tabular}
\end{minipage}
\caption{CPU time comparison of RHT, EVM, and CSA (proposed algorithm) for images from GREC2007 dataset.}
\label{fig:comp-time}
\end{figure}

For a better understanding, three plots are given in Figure~\ref{fig:comp-E2}, Figure~\ref{fig:comp-AD}, and Figure~\ref{fig:comp-time}, comparing the values of E2, AD, and CPU time for the three algorithms.
In the plots of Figure~\ref{fig:comp-E2} and Figure~\ref{fig:comp-AD}, the $x$-axis represents the image number in GREC2007 dataset.
Against the image number, E2 is plotted using sticks of black, red, and blue (Figure~\ref{fig:comp-E2}), corresponding to RHT, EVM, and CSA, respectively.
In Figure~\ref{fig:comp-AD}, the AD value is plotted for the three algorithms with the same color code as in Figure~\ref{fig:comp-E2}.
A high value of E2 signifies that not all the circles and circular arcs have been detected, or even if they are detected, their centers and radii have not been estimated accurately.
From these plots, it may be noticed that RHT and EVM algorithms are not efficient in detecting circular arcs.
For our algorithm, in most of the cases, the E2 values are significantly lower (less than~$9$, except for the image {\tt g07-tr3.tif}) compared to those for RHT and EVM.
For statistical details, see the adjoining tables given with the respective plots.
In the plots of Figure~\ref{fig:comp-AD}, barring a couple of cases, the AD values corresponding to our algorithm are either at~par or better than those corresponding to RHT.
Although the EVM algorithm is occasionally better than our algorithm in terms of AD, it has a poorer performance compared to our algorithm in terms of CPU time.
This is clearly evident from the plots in Figure~\ref{fig:comp-time} where the CPU time (in seconds) for each algorithm is plotted against the number of object pixels~($N$) processed. 
The reason for such a better performance of our algorithm in terms of CPU time can be attributed to the fact that it works only on arc pixels and not on all pixels of the input image.

\section{Conclusion}
\label{sec:conclu}

We have shown how the {\em chord-and-sagitta property}, with their appropriate adaptations on the digital
plane, can be used for identifying digital circles and circular arcs in a binary image.
A restricted Hough transform is finally used to further improve the accuracy of centers and radii.
We have tested our algorithm on various datasets and the related results demonstrate the efficiency and
robustness of the technique and the strength of the proposed approach in detecting circles and arcs on the
digital plane.

Since for each circle or circular arc, we have computed the center and radius before applying the Hough
transform, the size of the Hough space is very small.
Hence, the requirement of accumulator memory and the computation time are reduced significantly.
This is indeed reflected in the runtime comparison between the proposed algorithm and other existing
algorithms.

\pb{The work discussed in this paper is focused mainly on the detection of (digitally) circular arcs after removal of (digitally) straight pieces from a given binary image. 
It has a broader perspective in the sense that when we have to decompose a given digital curve into a minimum number of line segments and circular arcs, taken in totality, then the problem is quite challenging.
Although there has been some work in recent time related to this problem, e.g., \cite{Kole12,Kole14,Nguyen11}, 
none of these, however, has claimed or proved the minimality of the output.
It may be mentioned in this context that decomposing a digital curve into a sequence of minimum number of 
straight segments can be done in linear time \cite{Feschet05}, although a similar problem of decomposing 
a 3D digital surface into a minimum number of digital plane segments is known to be $NP$-hard~\cite{Sivignon29}.}

\bibliographystyle{abbrv}
\bibliography{arc-seg}

\vfill
\newpage
\section*{\Large Appendix}
\begin{center}

{\bf Results on GREC2013}\medskip

\begin{tabular}{c@{\hspace*{10mm}}c}
\includegraphics[width=.30\textwidth]{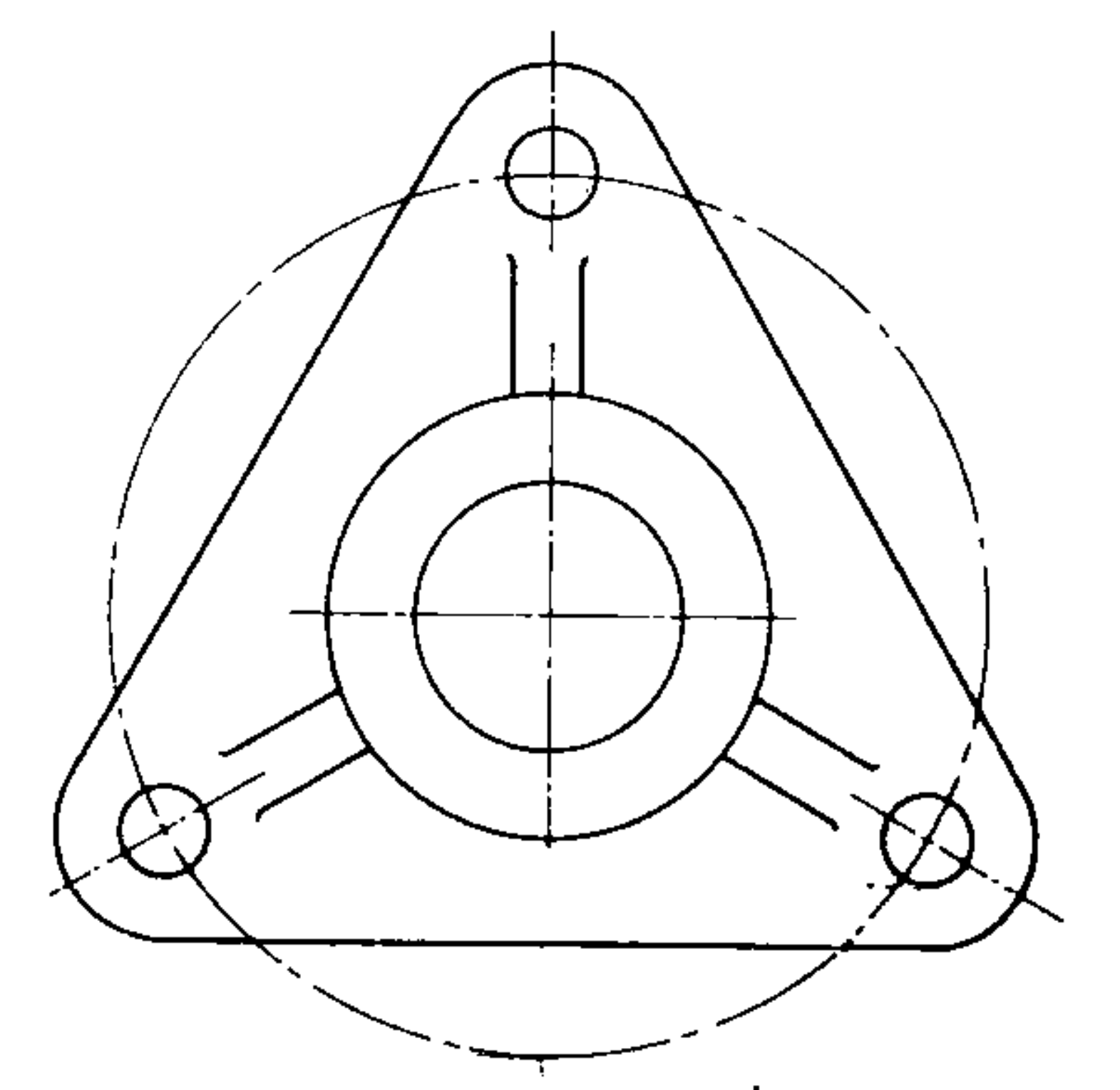}&
\includegraphics[width=.30\textwidth]{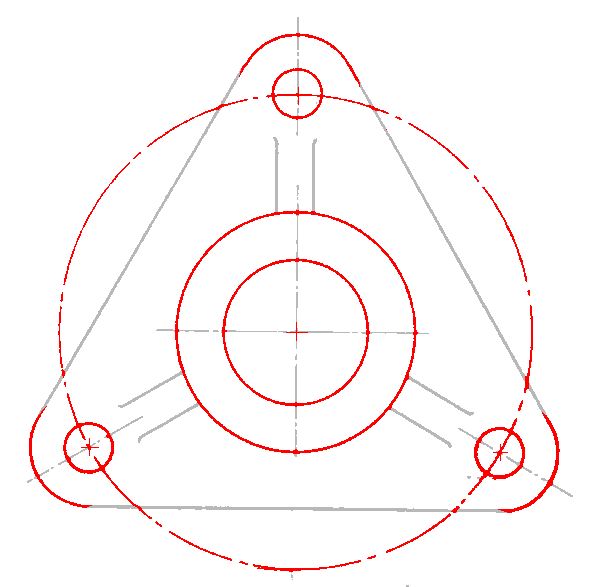}\\
\includegraphics[width=.30\textwidth]{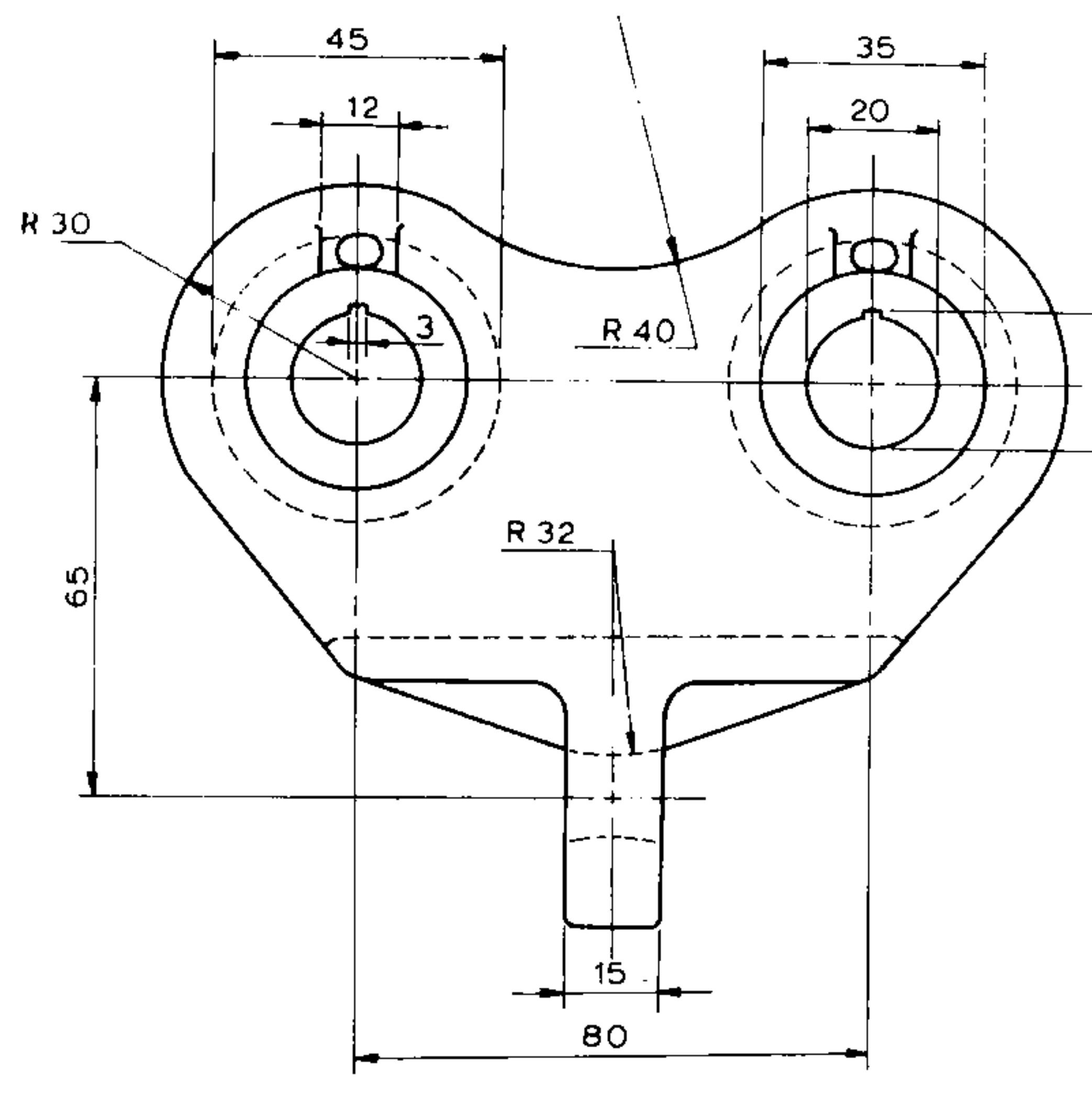}&
\includegraphics[width=.30\textwidth]{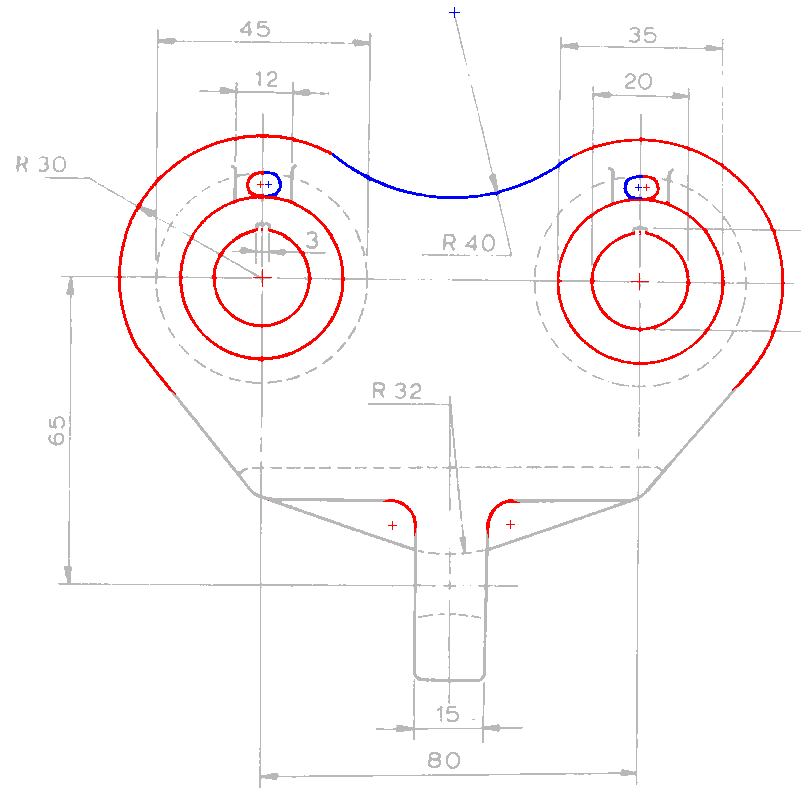}\\
\includegraphics[width=.30\textwidth]{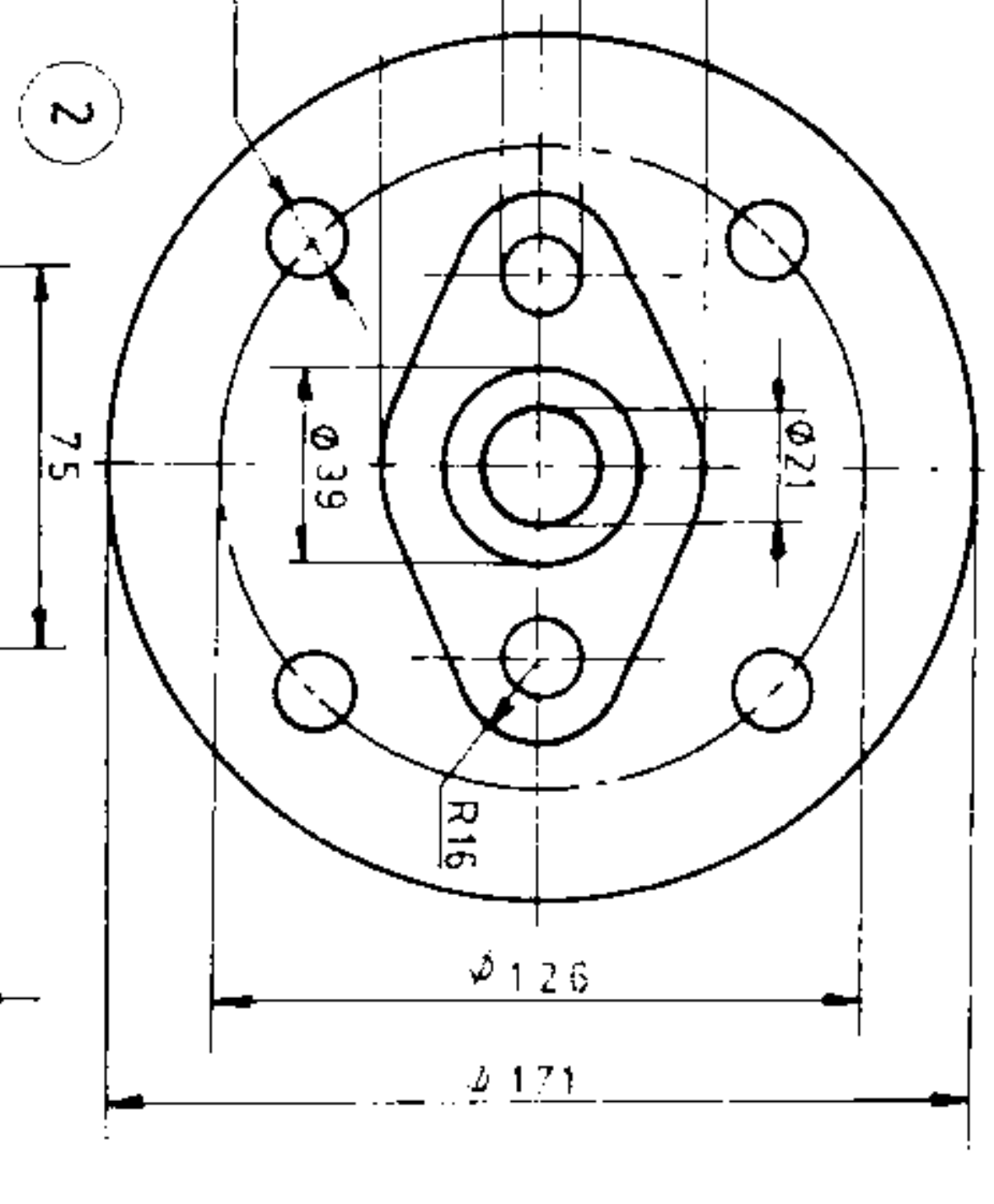}&
\includegraphics[width=.30\textwidth]{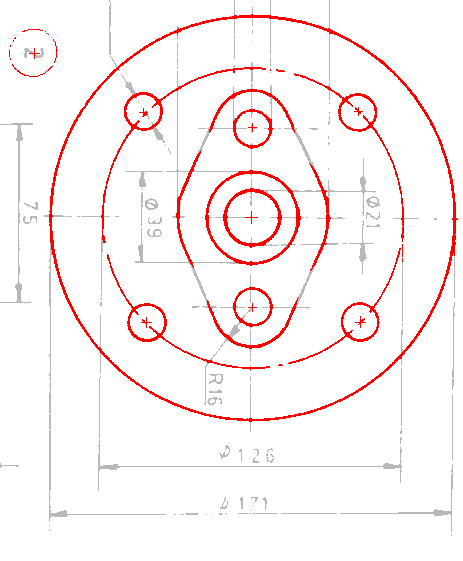}\\
\includegraphics[width=.30\textwidth]{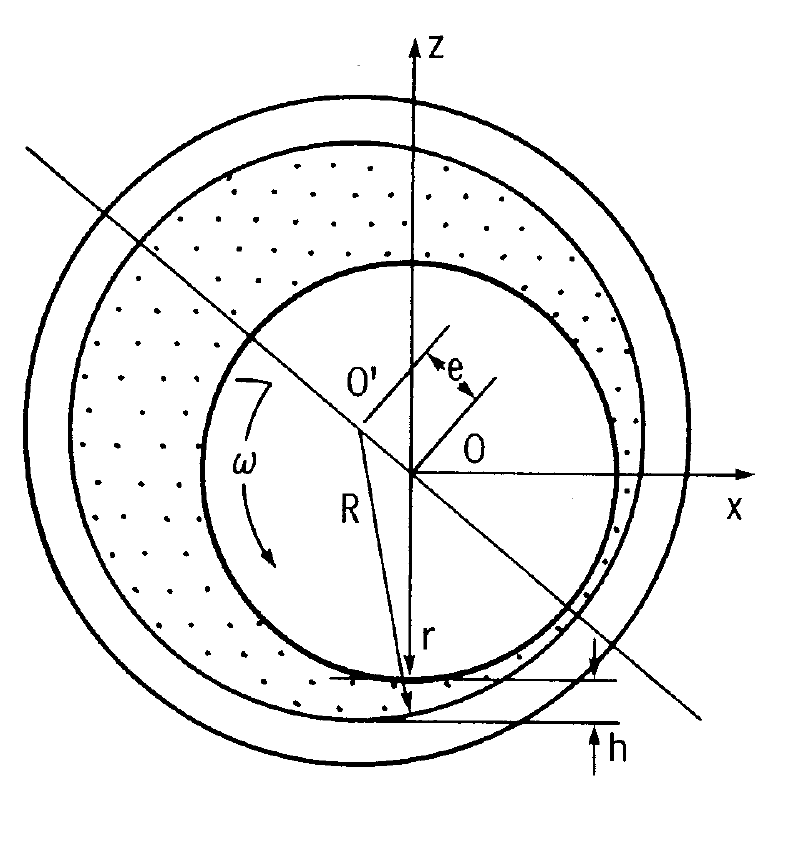}&
\includegraphics[width=.30\textwidth]{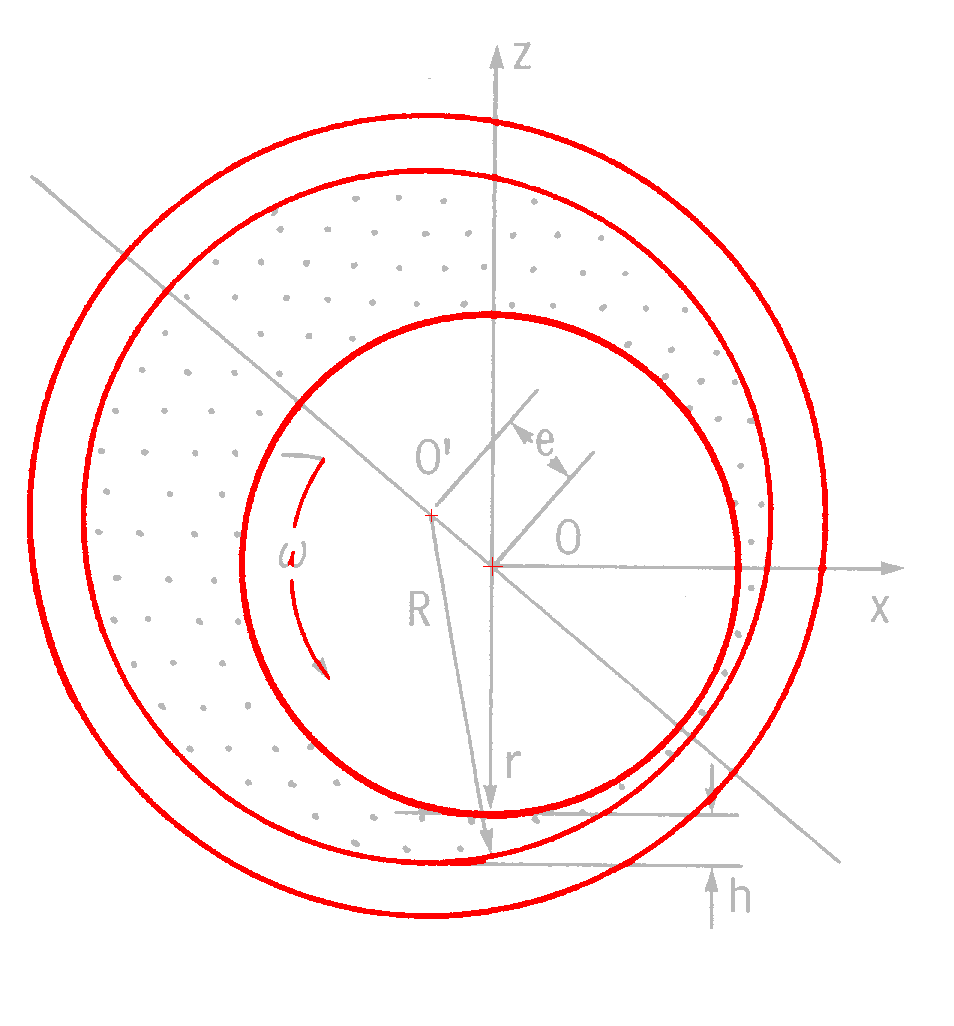}
\end{tabular}

\vfill
\newpage
{\bf Results on SMP}\medskip

\begin{tabular}{c@{\hspace*{10mm}}c}
\includegraphics[width=.40\textwidth]{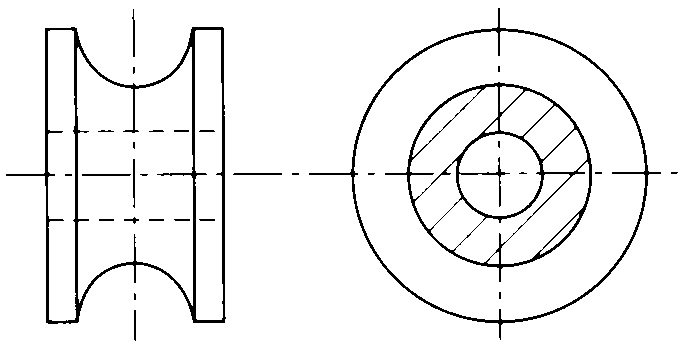}&
\includegraphics[width=.40\textwidth]{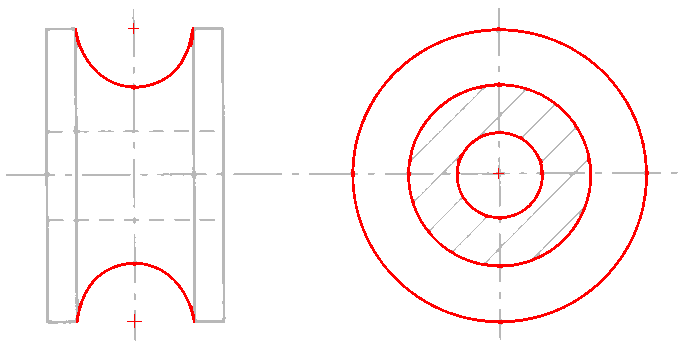}\\
\includegraphics[width=.40\textwidth]{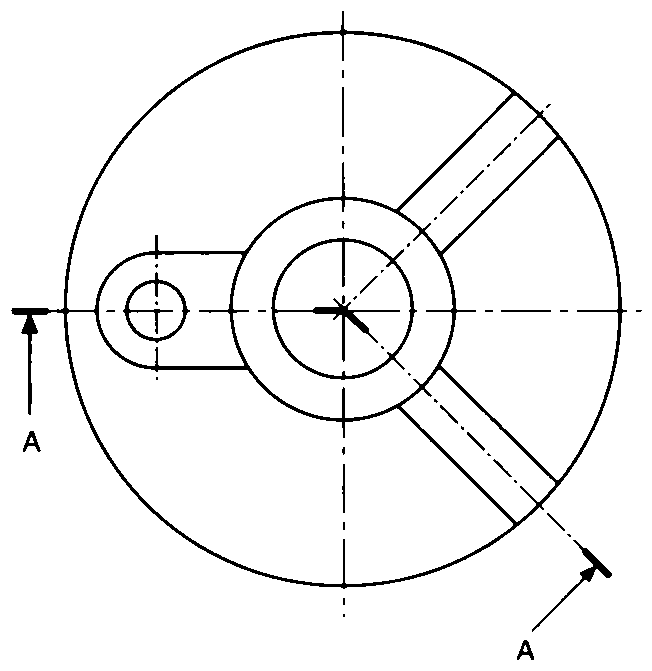}&
\includegraphics[width=.40\textwidth]{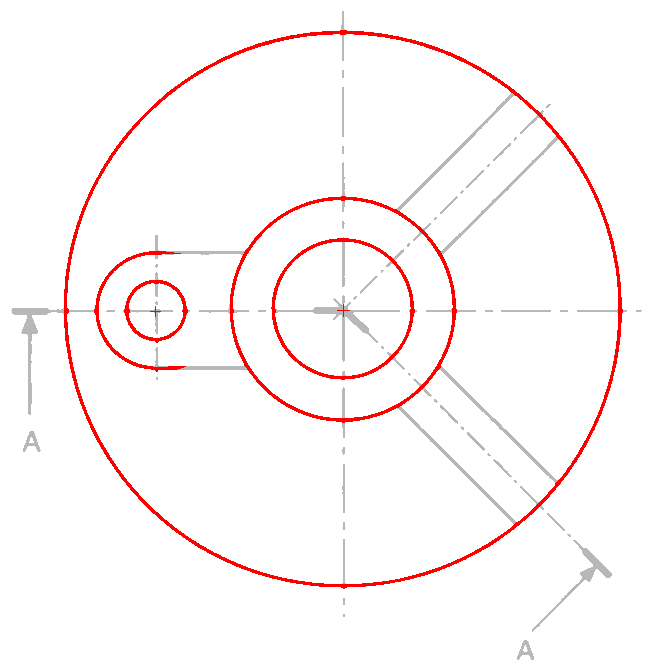}\\
\includegraphics[width=.50\textwidth]{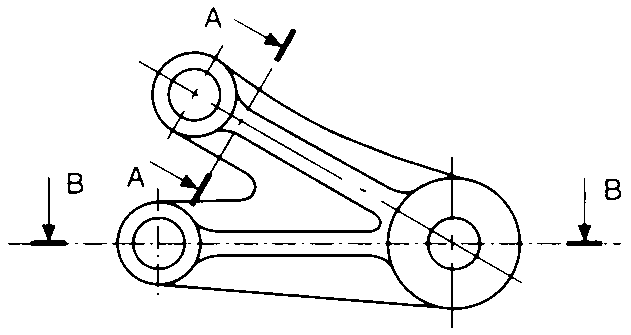}&
\includegraphics[width=.50\textwidth]{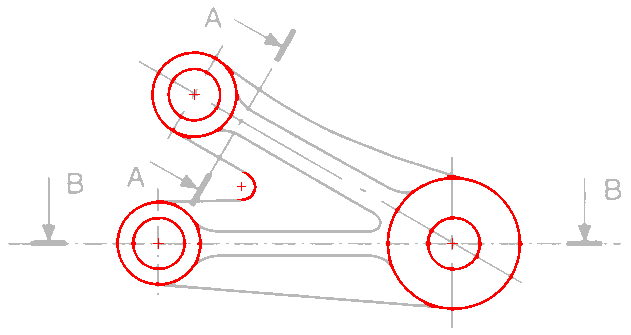}\\
\includegraphics[width=.40\textwidth]{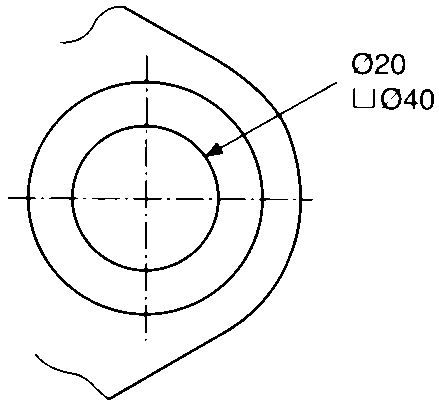}&
\includegraphics[width=.40\textwidth]{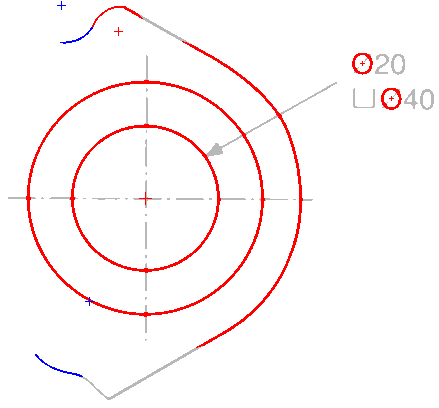}
\end{tabular}


\begin{tabular}{c@{\hspace*{10mm}}c}
\multicolumn{2}{c}{{\bf Results on SMP (continued)}\medskip}\\
\includegraphics[width=.36\textwidth]{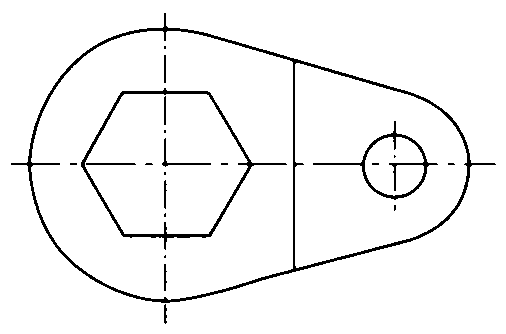}&
\includegraphics[width=.36\textwidth]{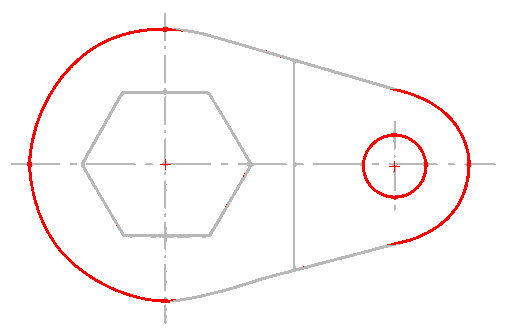}\\
\includegraphics[width=.36\textwidth]{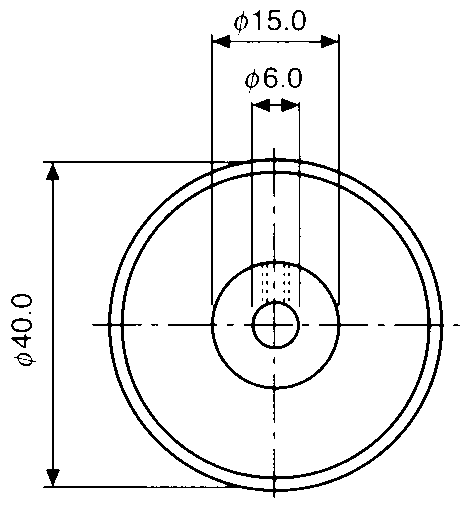}&
\includegraphics[width=.36\textwidth]{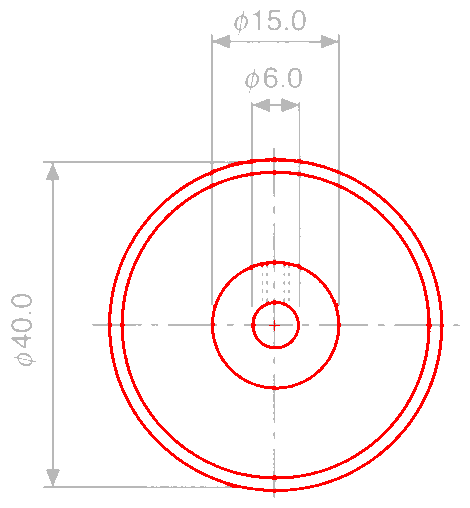}\\
\includegraphics[width=.36\textwidth]{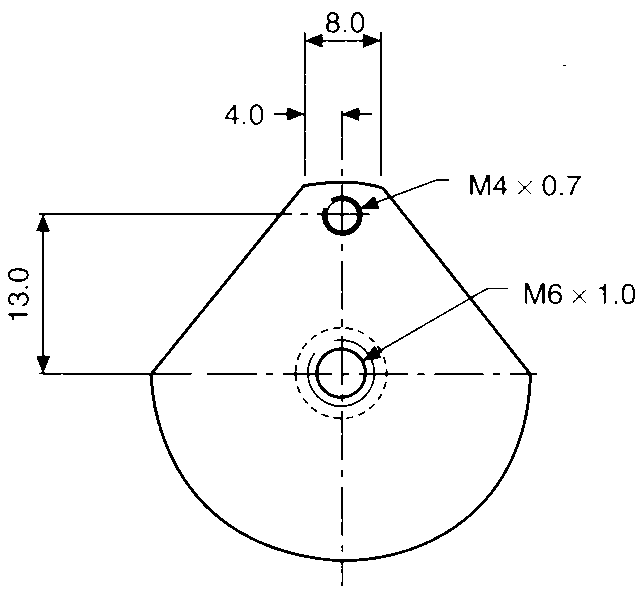}&
\includegraphics[width=.36\textwidth]{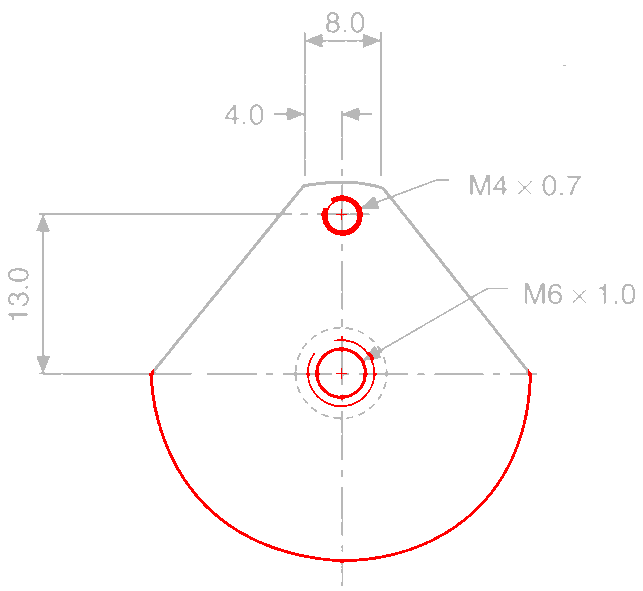}\\
\includegraphics[width=.36\textwidth]{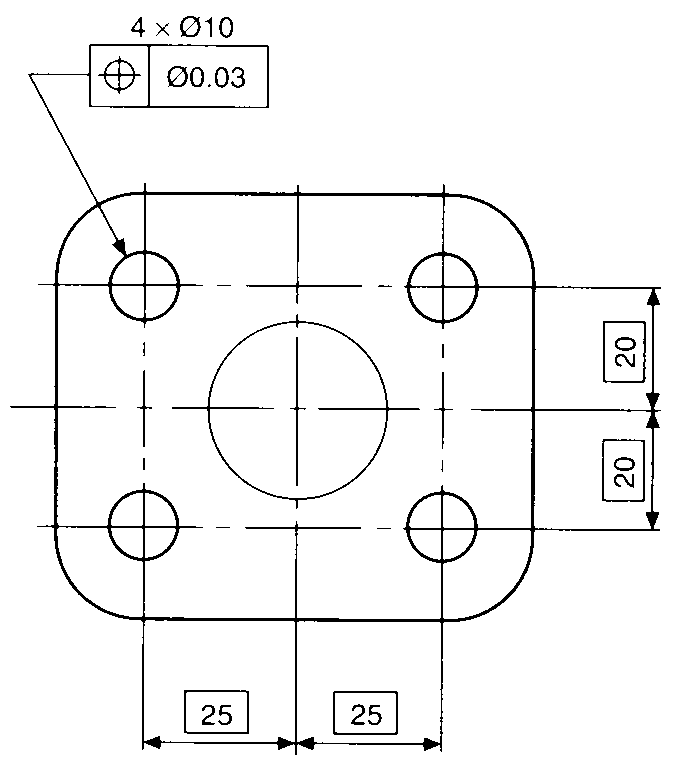}&
\includegraphics[width=.36\textwidth]{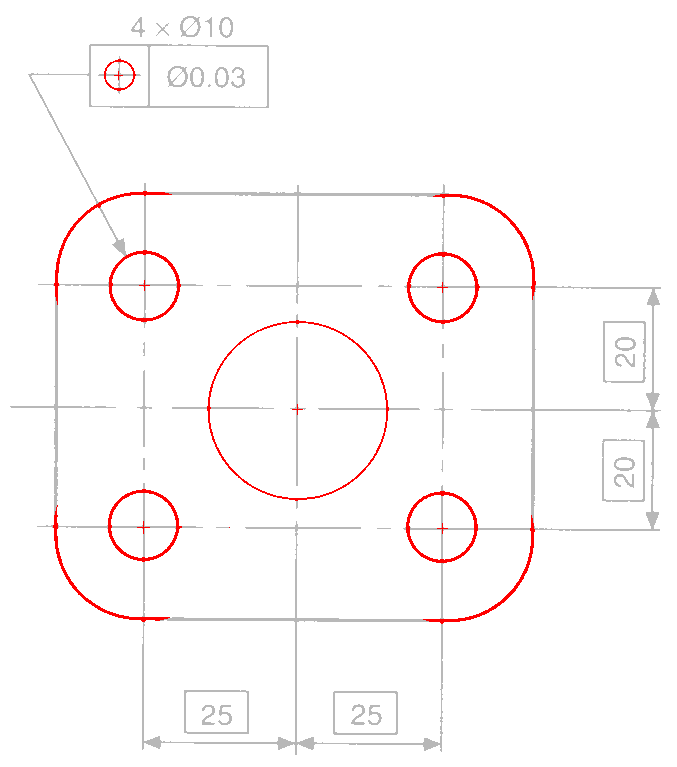}
\end{tabular}

\end{center}

\end{document}